\documentclass[10pt]{article}
\usepackage{mathtools}
\usepackage{amsmath}
\usepackage[shortlabels]{enumitem}

\usepackage{graphicx,epic,eepic,epsfig,latexsym,verbatim,color}
 
\usepackage{amsfonts}       
\usepackage{nicefrac}       
\usepackage[ruled, vlined, linesnumbered]{algorithm2e}
\usepackage[table]{xcolor}

\usepackage{framed}
\definecolor{shadecolor}{rgb}{0.9,0.90,0.9}
\usepackage{bbm}
 
\usepackage{float}
\usepackage{tikz}
\usetikzlibrary{chains}
\usetikzlibrary{fit}

\usepackage{epsfig}
\usetikzlibrary{shapes.symbols,patterns} 
\usepackage{pgfplots}

\usepackage[strict]{changepage}

\usepackage[marginal]{footmisc}
\usepackage{url}
\usepackage{theorem}

\newtheorem{definition}{Definition}
\newtheorem{proposition}{Proposition}
\newtheorem{lemma}[proposition]{Lemma}

\newtheorem{theorem}[proposition]{Theorem}
\newtheorem{corollary}[proposition]{Corollary}

\def\squareforqed{\hbox{\rlap{$\sqcap$}$\sqcup$}}
\def\qed{\ifmmode\squareforqed\else{\unskip\nobreak\hfil
\penalty50\hskip1em\null\nobreak\hfil\squareforqed
\parfillskip=0pt\finalhyphendemerits=0\endgraf}\fi}
\def\endenv{\ifmmode\;\else{\unskip\nobreak\hfil
\penalty50\hskip1em\null\nobreak\hfil\;
\parfillskip=0pt\finalhyphendemerits=0\endgraf}\fi}
\newenvironment{proof}{\noindent \textbf{{Proof~} }}{\hfill $\blacksquare$}

\newcounter{remark}
\newenvironment{remark}[1][]{\refstepcounter{remark}\par\medskip\noindent%
\textbf{Remark~\theremark #1} }{\medskip}

\newcounter{example}

\mathchardef\ordinarycolon\mathcode`\:
\mathcode`\:=\string"8000
\def\vcentcolon{\mathrel{\mathop\ordinarycolon}}
\begingroup \catcode`\:=\active
  \lowercase{\endgroup
  \let :\vcentcolon
  }

\usepackage{graphicx}

\RequirePackage[framemethod=default]{mdframed}
\newmdenv[skipabove=7pt,
skipbelow=7pt,
backgroundcolor=darkblue!15,
innerleftmargin=5pt,
innerrightmargin=5pt,
innertopmargin=5pt,
leftmargin=0cm,
rightmargin=0cm,
innerbottommargin=5pt,
linewidth=1pt]{tBox}

\newmdenv[skipabove=7pt,
skipbelow=7pt,
backgroundcolor=red!15,
innerleftmargin=5pt,
innerrightmargin=5pt,
innertopmargin=5pt,
leftmargin=0cm,
rightmargin=0cm,
innerbottommargin=5pt,
linewidth=1pt]{rBox}

\newmdenv[skipabove=7pt,
skipbelow=7pt,
backgroundcolor=blue2!25,
innerleftmargin=5pt,
innerrightmargin=5pt,
innertopmargin=5pt,
leftmargin=0cm,
rightmargin=0cm,
innerbottommargin=5pt,
linewidth=1pt]{dBox}
\newmdenv[skipabove=7pt,
skipbelow=7pt,
backgroundcolor=darkkblue!15,
innerleftmargin=5pt,
innerrightmargin=5pt,
innertopmargin=5pt,
leftmargin=0cm,
rightmargin=0cm,
innerbottommargin=5pt,
linewidth=1pt]{sBox}
\definecolor{darkblue}{RGB}{0,76,156}
\definecolor{darkkblue}{RGB}{0,0,153}
\definecolor{blue2}{RGB}{102,178,255}
\definecolor{darkred}{RGB}{195,0,0}

\newcommand{\nc}{\newcommand}
\nc{\rnc}{\renewcommand}
\nc{\lbar}[1]{\overline{#1}}
\nc{\bra}[1]{\langle#1|}
\nc{\ket}[1]{|#1\rangle}
\nc{\dketbra}[2]{\vert #1 \rangle \hspace{-.8mm} \rangle \hspace{-.4mm} \langle\hspace{-.8mm}\langle #2 \vert}
\nc{\dbra}[1]{\langle\hspace{-.8mm}\langle #1\vert}
\nc{\dket}[1]{\vert#1\rangle\hspace{-.8mm}\rangle}
\nc{\ketbra}[2]{|#1\rangle\!\langle#2|}
\nc{\braket}[2]{\langle#1|#2\rangle}
\nc{\braandket}[3]{\langle #1|#2|#3\rangle}

\nc{\proj}[1]{| #1\rangle\!\langle #1 |}
\nc{\avg}[1]{\langle#1\rangle}
\nc{\rank}{\operatorname{Rank}}
\nc{\smfrac}[2]{\mbox{$\frac{#1}{#2}$}}
\nc{\tr}{\operatorname{Tr}}
\nc{\ox}{\otimes}
\nc{\dg}{\dagger}
\nc{\dn}{\downarrow}
\nc{\cA}{{\cal A}}
\nc{\cB}{{\cal B}}
\nc{\cC}{{\cal C}}
\nc{\cD}{{\cal D}}
\nc{\cE}{{\cal E}}
\nc{\cF}{{\cal F}}
\nc{\cG}{{\cal G}}
\nc{\cH}{{\cal H}}
\nc{\cI}{{\cal I}}
\nc{\cJ}{{\cal J}}
\nc{\cK}{{\cal K}}
\nc{\cL}{{\cal L}}
\nc{\cM}{{\cal M}}
\nc{\cN}{{\cal N}}
\nc{\cO}{{\cal O}}
\nc{\cP}{{\cal P}}
\nc{\cQ}{{\cal Q}}
\nc{\cR}{{\cal R}}
\nc{\cS}{{\cal S}}
\nc{\cT}{{\cal T}}
\nc{\cU}{{\cal U}}
\nc{\cV}{{\cal V}}
\nc{\cX}{{\cal X}}
\nc{\cY}{{\cal Y}}
\nc{\cZ}{{\cal Z}}
\nc{\cW}{{\cal W}}
\nc{\csupp}{{\operatorname{csupp}}}
\nc{\qsupp}{{\operatorname{qsupp}}}
\nc{\var}{{\operatorname{var}}}
\nc{\rar}{\rightarrow}
\nc{\lrar}{\longrightarrow}
\nc{\polylog}{{\operatorname{polylog}}}
\nc{\idop}{{\mathds{1}}}
\nc{\wt}{{\operatorname{wt}}}
\nc{\av}[1]{{\left\langle {#1} \right\rangle}}
\nc{\supp}{{\operatorname{supp}}}
\nc{\SEP}{{\mathrm{SEP}}}
\nc{\PPT}{{\mathrm{PPT}}}

\nc{\argmin}{{\operatorname{argmin}}}

\def\x{\xi}

\nc{\RR}{{{\mathbb R}}}
\nc{\CC}{{{\mathbb C}}}
\nc{\FF}{{{\mathbb F}}}
\nc{\NN}{{{\mathbb N}}}
\nc{\ZZ}{{{\mathbb Z}}}
\nc{\PP}{{{\mathbb P}}}
\nc{\QQ}{{{\mathbb Q}}}
\nc{\UU}{{{\mathbb U}}}
\nc{\EE}{{{\mathbb E}}}
\nc{\id}{{\operatorname{id}}}

\nc{\be}{\begin{equation}}
\nc{\ee}{\end{equation}}
\nc{\bea}{\begin{equation}\begin{aligned}\hspace{0pt}}
\nc{\eea}{\end{aligned}\end{equation}}
\nc{\eqt}[1]{\stackrel{\mathclap{\scriptsize \mbox{#1}}}{=}}
\nc{\leqt}[1]{\stackrel{\mathclap{\scriptsize \mbox{#1}}}{\leq}}
\nc{\geqt}[1]{\stackrel{\mathclap{\scriptsize \mbox{#1}}}{\geq}}

\nc{\<}{\langle}
\rnc{\>}{\rangle}
\nc{\rU}{\mbox{U}}

\nc{\ob}[1]{#1}

\newcommand{\Choi}{Choi-Jamio\l{}kowski }

\newcommand{\CPTP}{\text{\rm CPTP}}
\newcommand{\NSO}{\mathrm{NS}}

\newcommand{\tB}{\widetilde{B}}
\newcommand{\tA}{\widetilde{A}}

\newcommand{\tS}{\widetilde{S}}
\newcommand{\bB}{\Bar{B}}
\newcommand{\bA}{\Bar{A}}

\usepackage{tikz}

\makeatletter
\def\grd@save@target#1{%
  \def\grd@target{#1}}
\def\grd@save@start#1{%
  \def\grd@start{#1}}
\tikzset{
  grid with coordinates/.style={
    to path={%
      \pgfextra{%
        \edef\grd@@target{(\tikztotarget)}%
        \tikz@scan@one@point\grd@save@target\grd@@target\relax
        \edef\grd@@start{(\tikztostart)}%
        \tikz@scan@one@point\grd@save@start\grd@@start\relax
        \draw[minor help lines,magenta] (\tikztostart) grid (\tikztotarget);
        \draw[major help lines] (\tikztostart) grid (\tikztotarget);
        \grd@start
        \pgfmathsetmacro{\grd@xa}{\the\pgf@x/1cm}
        \pgfmathsetmacro{\grd@ya}{\the\pgf@y/1cm}
        \grd@target
        \pgfmathsetmacro{\grd@xb}{\the\pgf@x/1cm}
        \pgfmathsetmacro{\grd@yb}{\the\pgf@y/1cm}
        \pgfmathsetmacro{\grd@xc}{\grd@xa + \pgfkeysvalueof{/tikz/grid with coordinates/major step}}
        \pgfmathsetmacro{\grd@yc}{\grd@ya + \pgfkeysvalueof{/tikz/grid with coordinates/major step}}
        \foreach \x in {\grd@xa,\grd@xc,...,\grd@xb}
        \node[anchor=north] at (\x,\grd@ya) {\pgfmathprintnumber{\x}};
        \foreach \y in {\grd@ya,\grd@yc,...,\grd@yb}
        \node[anchor=east] at (\grd@xa,\y) {\pgfmathprintnumber{\y}};
      }
    }
  },
  minor help lines/.style={
    help lines,
    step=\pgfkeysvalueof{/tikz/grid with coordinates/minor step}
  },
  major help lines/.style={
    help lines,
    line width=\pgfkeysvalueof{/tikz/grid with coordinates/major line width},
    step=\pgfkeysvalueof{/tikz/grid with coordinates/major step}
  },
  grid with coordinates/.cd,
  minor step/.initial=.2,
  major step/.initial=1,
  major line width/.initial=2pt,
}
\makeatother

\usepackage{thmtools}
\usepackage{thm-restate}
\usepackage{etoolbox}
\makeatletter
\def\problem@s{}
\newcounter{problems@cnt}

\newcommand{\allproblems}{\problem@s}
\makeatother

\usepackage{tikz}
\usetikzlibrary{positioning}
\usetikzlibrary{shapes.geometric}
\usetikzlibrary{calc}
\definecolor{tensorblue}{rgb}{0.8,0.9,1}
\tikzset{ten/.style={fill=tensorblue}}

\usepackage{amsmath,amsfonts,amssymb,array,graphicx,mathtools,multirow,bm,tcolorbox,relsize,booktabs,dsfont}
\usepackage[utf8]{inputenc}
\usepackage[T1]{fontenc}
\usepackage{qcircuit}
\usepackage{mathrsfs}
\usepackage{authblk} 
\usepackage{optidef}
\allowdisplaybreaks
\usepackage[paperwidth=199.8mm,paperheight=297mm,centering,hmargin=20mm,vmargin=2cm]{geometry}

\definecolor{refcolor}{rgb}{0.067,0.5,0.65}
\definecolor{urlcolor}{rgb}{0.1,0,0.9}
\usepackage[
    breaklinks,
    pdftex,
    colorlinks=true,
    linkcolor=refcolor,
    citecolor=refcolor,
    urlcolor=refcolor
]{hyperref}
\usepackage{cleveref}

\begin{document}
\title{Classical communication cost of a bipartite quantum channel assisted by non-signalling correlations}

\author[1]{Chengkai Zhu}
\author[2]{Xuanqiang Zhao}
\author[1]{Xin Wang \thanks{felixxinwang@hkust-gz.edu.cn}}

\affil[1]{\small Thrust of Artificial Intelligence, Information Hub,\par The Hong Kong University of Science and Technology (Guangzhou), Guangzhou 511453, China}
\affil[2]{QICI Quantum Information and Computation Initiative, Department of Computer Science,\par The University of Hong Kong, Pokfulam Road, Hong Kong}
\date{\today}
\maketitle

\begin{abstract}
Understanding the classical communication cost of simulating a quantum channel is a fundamental problem in quantum information theory, which becomes even more intriguing when considering the role of non-locality in quantum information processing. This paper investigates the bidirectional classical communication cost of simulating a bipartite quantum channel assisted by non-signalling correlations, which are permitted across the spatial dimension between the two parties. By introducing non-signalling superchannels, we present lower and upper bounds on the one-shot $\epsilon$-error one-way classical communication cost of a bipartite channel via its smooth max-relative entropy of one-way classical communication, and establish that the asymptotic exact cost is given by its max-relative entropy of one-way classical communication. For the bidirectional scenario, we derive semidefinite programming (SDP) formulations for the one-shot exact bidirectional classical communication cost via non-signalling bipartite superchannels. We further introduce a channel's bipartite conditional min-entropy as an efficiently computable lower bound on the asymptotic cost of bidirectional classical communication. Our results in both one-shot and asymptotic settings provide lower bounds on the entanglement-assisted simulation cost in scenarios where entanglement is available to the two parties. Moreover, we propose a seesaw-based algorithm to compute an upper bound on the minimum simulation error via local operations and shared entanglement, which provides valuable insights into the relationship between non-signalling bipartite superchannels and more physically realizable protocols. Numerical experiments demonstrate the effectiveness of our bounds in estimating communication costs for various quantum channels, showing that our bounds can be tight in different scenarios. Our results elucidate the role of non-locality in quantum communication and pave the way for exploring quantum reverse Shannon theory in bipartite scenarios.
\end{abstract}

\tableofcontents

\section{Introduction}
The study of quantum non-locality has been a central theme in quantum information theory since Bell's seminal work~\cite{Bell_1966}, which demonstrated that any local hidden variable theory cannot explain the correlations exhibited by entangled quantum systems. This non-local nature of quantum correlations has found numerous applications in quantum information processing, including quantum teleportation~\cite{Bennett1993}, superdense coding~\cite{Bennett1992}, and quantum cryptography~\cite{Bennett1984,Ekert1991}. Moreover, the peaceful coexistence between quantum non-locality and the impossibility of superluminal communication has intrigued physicists for decades.

Despite the impossibility of utilizing quantum non-locality for superluminal communication, non-local resources like entanglement can be used to enhance the communication of classical information~\cite{Bennett1999}, private classical information~\cite{Devetak2005}, and quantum information. The most studied scenario is one-way communication between two parties, Alice and Bob, where they share a channel allowing Alice to send information to Bob. The communication is said to be entanglement-assisted if Alice and Bob can share an unlimited amount of entanglement before their communication. While a one-way quantum channel is a natural model of communication between two parties, a \textit{bipartite quantum channel}, which jointly evolves an input from both Alice and Bob, provides the most general setting for two-party communication~\cite{CHILDS_2006} (see Figure~\ref{fig:OW_TW_channel}).
Bipartite quantum channels play a crucial role in describing and understanding a wide range of quantum interactions (e.g.,~\cite{Fawzi2012a,Goold_2016,Stefan2018,Stefan2019,Gilad_2020a,Gour2021,Hirche2023,Ding2023}) and are also reminiscent of two-way classical channels, first studied by Shannon~\cite{shannon1961two} and followed by extensive research on interference channels~\cite{Han1981,Motahari2009} and multiple-access channels (MACs)~\cite{Tse_Viswanath_2005}. A point-to-point channel can be viewed as a special case of a bipartite channel, where Bob's input and Alice's output are trivial one-dimensional subsystems. Another common special case is a bipartite unitary, which describes the interaction between two isolated quantum systems. Similar to the fundamental challenges in classical network information theory~\cite{el2011network}, the analysis of bipartite quantum channels presents significant challenges.

One central problem in quantum information theory is to study how many standard communication resources, e.g., noiseless channels, are needed to produce other specialized resources, e.g., available noisy channels, under different settings. This problem is also known as the channel simulation problem~\cite{Fang_2020}, which is central in quantum Shannon theory. When Alice and Bob can share an unlimited amount of entanglement, the quantum reverse Shannon theorem~\cite{Bennett2014,Berta_2011} completely solves the problem in the asymptotic setting, stating that the optimal rate to simulate a quantum channel is determined by its entanglement-assisted classical capacity. In particular, for bipartite quantum channels, the standard resources are bidirectional communication~\cite{Bennett_2003}. Harrow and Leung~\cite{Harrow_2011} established a protocol for simulating a given bipartite unitary with the assistance of unlimited EPR pairs to study the classical communication cost. How to simulate general bipartite quantum channels with minimum resources, e.g., classical communication, remains unclear (see Figure~\ref{fig:channel_simu}).

\begin{figure}[b]
    \centering
    \includegraphics[width=0.7\linewidth]{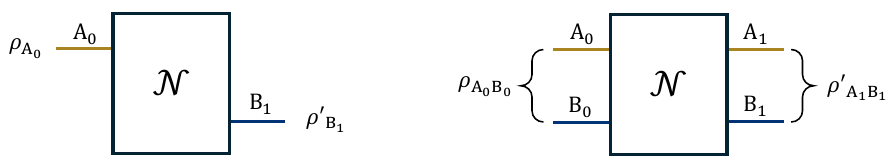}
    \caption{The left side shows a point-to-point quantum channel $\cN$ through which Alice can send quantum states to Bob. The right side illustrates a bipartite quantum channel, where each party provides an input to the channel, which evolves their inputs jointly to produce an output $\rho'_{A_1B_1}$ shared between them.}
    \label{fig:OW_TW_channel}
\end{figure}

Different from a point-to-point communication setting, the role of quantum resources in a multipartite communication or network communication setting becomes more sophisticated. For example, it was shown that entanglement offers no advantage in a point-to-point purely classical communication scenario~\cite{Bennett1999} while can substantially boost the capacity of a classical MAC~\cite{Leditzky_2020}. In addition to allowing Alice and Bob to share finite or unlimited entanglement to assist with information manipulation, a more general framework for investigating the role of non-locality in communication is to consider a setting where they have access to any non-signalling bipartite channels~\cite{Piani_2006}. Non-signalling bipartite channels are also referred to as \textit{non-signalling boxes} or \textit{causal maps} and, while cannot be used to communicate, can exhibit a stronger violation of Bell’s inequalities than local measurements performed on bipartite quantum states~\cite{Beckman2001,Eggeling_2002,Oreshkov_2012}. A well-known classical example is the Popescu-Rohrlich machine for violating Clauser-Horne-Shimony-Holt inequalities~\cite{Clauser1969,Popescu1997,Barrett2005}. Furthermore, the non-signalling correlations are valuable in studying the coding theory of point-to-point channels, known as the \textit{non-signalling codes} when giving the input system and output system access to any quantum resources that cannot be used for communicating between them directly (see, e.g.,~\cite{Leung2015,Duan2016,Wang_2018} for quantum channels and, e.g.,~\cite{Cubitt_2011,Matthews_2012,Barman_2018,Omar2024,Omar2024a} for classical channels).

\begin{figure}[t]
    \centering
    \includegraphics[width=.7\linewidth]{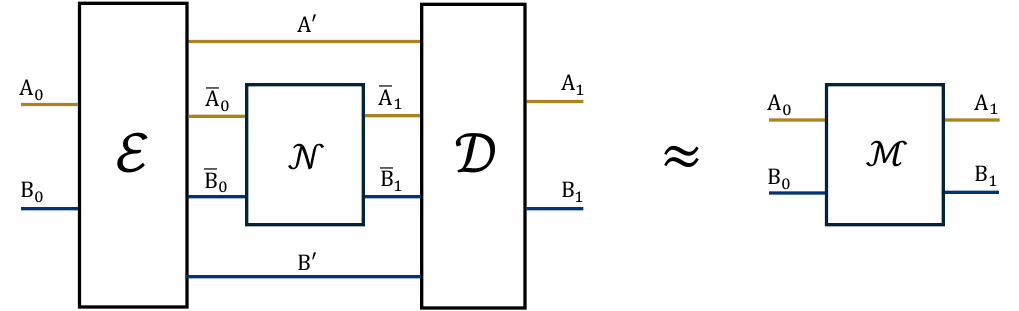}
    \caption{Bipartite quantum channel simulation. Alice (in yellow) and Bob (in blue) aim to simulate a bipartite quantum channel $\cM_{A_0B_0\rightarrow A_1B_1}$ with given $\cN_{\bA_0\bB_0\rightarrow \bA_1\bB_1}$ and available channels $\cE_{A_0B_0\rightarrow A'\bA_0 \bB_0 B'}$ and $\cD_{A'\bA_1\bB_1B'\rightarrow A_1B_1}$. In general, $\cN$ can be a standard communication resource between Alice and Bob, e.g., noiseless bidirectional channels, and we aim to quantify the minimum resource required for simulating the target channel $\cM$.}
    \label{fig:channel_simu}
\end{figure}

\subsection{Summary of results}
By allowing Alice and Bob to consume these non-signalling correlations, we can explore the ultimate limits of non-locality in facilitating communication. In this context, classical communication between the parties becomes the golden resource. Our goal in this work is to determine the minimum number of classical bits that need to be exchanged between the two parties to simulate a target bipartite channel, given access to unlimited non-signalling correlations between two communicating parties and unlimited quantum resources in each party. 

In Section~\ref{sec:pre}, we introduce the notations and preliminaries used throughout the paper. In Section~\ref{sec:ns_channel_superchan}, we introduce the notions of bipartite non-signalling channels and superchannels and characterize their properties. Section~\ref{sec:measures} defines resource measures of classical communication based on the max-relative entropy and relative entropy. 

In Section~\ref{sec:sim_CNSP}, we study the problem of simulating a bipartite quantum channel using one-way classical communication and one-way non-signalling bipartite superchannel, as well as the simulation using bidirectional classical communication and bidirectional non-siganlling bipartite superchannel. We derive both upper and lower bounds for the one-shot $\epsilon$-error one-way classical communication cost of a bipartite channel through its smooth max-relative entropy of one-way classical communication, and present a semidefinite programming (SDP) formulation of the one-shot exact bidirectional classical communication cost of an arbitrary bipartite channel. Based on this, we further have an SDP hierarchy to compute the $\epsilon$-error bidirectional classical communication cost. For the asymptotic setting, we present a closed-form of the asymptotic exact one-way classical communication cost as the channel's max-relative entropy of one-way communication. Then we introduce the channel conditional min-entropy and demonstrate that it serves as an efficiently computable lower bound on the asymptotic exact bidirectional classical communication cost of a bipartite channel. 

In Section~\ref{sec:lose}, we consider the problem of bipartite quantum channel simulation via local operations and shared entanglement (LOSE) to investigate the relationship between the non-signalling correlation and shared entanglement. We propose a bipartite channel simulation framework via LOSE and develop a seesaw-based algorithm to estimate the optimal simulation error of the proposed framework. A numerical example is provided to demonstrate the effectiveness of our method. This analysis provides valuable physical insights that complement our understanding of non-signalling bipartite superchannels.

Section~\ref{sec:examples} presents numerical experiments that demonstrate the performance of our bounds on bidirectional classical communication cost via non-signalling bipartite superchannels. Finally, we conclude with a discussion of our results and future directions in Section~\ref{sec:conclu}.

\section{Preliminaries}\label{sec:pre}

Throughout the paper, the physical systems and their associated Hilbert spaces are denoted by $A,B,C$ and $\cH_A, \cH_B, \cH_C$, respectively. We consider finite-dimensional Hilbert spaces and denote the dimension of a quantum system by $d_A$.
We use the tilde symbol to denote an identical copy of the system under it, i.e., $\cH_{\tA}$ is isomorphic to $\cH_A$. We denote by $\mathscr{L}(\cH_A)$ the set of linear operators and $\mathscr{D}(\cH_A)$ the set of density operators or quantum states in the system $A$. A quantum channel $\cN_{A\to B}$ is a linear map from $\mathscr{L}(\cH_A)$ to $\mathscr{L}(\cH_B)$ that is completely positive (CP) and trace-preserving (TP). We denote by $\id$ the identity channel and by $\CPTP(A,B)$ the set of all quantum channels from system $A$ to system $B$. The \Choi operator of a quantum channel $\cN_{A\to B}$ is given by
\be
J^{\cN}_{A B} \coloneqq \big(\id_{A}\ox \cN_{\tA\rightarrow B}\big)(\ketbra{\Gamma}{\Gamma}_{A \tA}) =  \sum_{i, j=0}^{d_A-1}\ketbra{i}{j}_{A} \ox \cN_{\tA \to B}(\ketbra{i}{j}_{\tA}),
\ee
where $\big\{\ket{i}_A\big\}$ and $\big\{\ket{i}_{\tA}\big\}$ are orthonormal bases of isomorphic Hilbert spaces $\cH_A$ and $\cH_{\tA}$, respectively, and $\ket{\Gamma}_{A \tA}\coloneqq \sum_{i=0}^{d_A-1}\ket{i}_{A}\ket{i}_{\tA}$ is the unnormalized maximally entangled state between systems $A$ and $\tA$. For simplicity, we often omit the identity channel explicitly in equations, e.g., $J^{\cN}_{A B}=\cN_{\tA\rightarrow B}(\ketbra{\Gamma}{\Gamma}_{A \tA})$.

More generally, a multipartite quantum channel refers to a channel in which the input and output systems consist of composite quantum systems shared among multiple parties. For a bipartite quantum channel $\cN_{A_0B_0\rightarrow A_1B_1}$, we consider the input and output shared between Alice and Bob, where we use the subscript `$0$' to denote the input system and the subscript `$1$' to denote the output system for each party. The set of all bipartite quantum channels from system $A_0B_0$ to system $A_1B_1$ is denoted by $\CPTP(A_0B_0, A_1B_1)$. The \Choi operator of a bipartite channel $\cN_{A_0B_0\rightarrow A_1B_1}$ is given by
\be
    J_{A_0B_0A_1B_1}^{\cN} \coloneqq \cN_{\tA_0\tB_0\rightarrow A_1B_1}\big(\proj{\Gamma}_{A_0\tA_0}\ox \proj{\Gamma}_{B_0\tB_0}\big).
\ee
For simplicity, we also use $J_{\mathbf{A}\mathbf{B}}^{\cN}$ to denote the \Choi operator of a bipartite quantum channel, where the bold symbol $\mathbf{A}$ in the subscript denotes all Alice's systems and $\mathbf{B}$ denotes all Bob's systems.

Transitioning to a higher-order perspective, quantum channels are subject to manipulation by quantum superchannels, which transform quantum channels into quantum channels. A supermap $\Theta_{(A\rightarrow B)\rightarrow (C\rightarrow D)}$ is called a superchannel if the output map
\begin{equation}\label{Eq:superchan}
    (\mathbb{I}_R \ox \Theta_{(A\rightarrow B)\rightarrow (C\rightarrow D)})(\cN_{R_0A\rightarrow R_1B})
\end{equation}
is a quantum channel for all input bipartite channels $\cN_{R_0A\rightarrow R_1B}$, where the reference system $R_0,R_1$ are arbitrary and $\mathbb{I}_R$ denotes the identity superchannel for channels from $R_0$ to $R_1$. For a simpler notation, we will often omit the identity supermap when there is no ambiguity, e.g., Eq.~\eqref{Eq:superchan} will be denoted as $\Theta(\cN_{R_0A\rightarrow R_1B})$. It is shown that any superchannel $\Theta_{(A\rightarrow B)\rightarrow (C\rightarrow D)}$ can be realized by a pre-processing channel $\cE_{C\rightarrow AA'}$ and a post-processing channel $\cD_{BA'\rightarrow D}$ such that~\cite{Chiribella2008}
\be
\Theta_{(A\rightarrow B)\rightarrow (C\rightarrow D)}(\cN_{A\rightarrow B}) = \cD_{BA'\rightarrow D}\circ \cN_{A\rightarrow B} \circ \cE_{C\rightarrow AA'}.
\ee
Hence, every superchannel $\Theta_{(A\rightarrow B)\rightarrow (C\rightarrow D)}$ corresponds to a bipartite quantum channel
\be
    \cQ^{\Theta}_{CB\rightarrow AD} \coloneqq \cD_{BA'\rightarrow D}\circ \cE_{C\rightarrow AA'},
\ee
both of which share the same \Choi operator, i.e.,
\begin{equation}
    J_{CBAD}^{\Theta} \coloneqq \cQ^{\Theta}_{\widetilde{C}\tB\rightarrow AD}\big(\proj{\Gamma}_{C\widetilde{C}}\ox \proj{\Gamma}_{B\tB}\big),
\end{equation}
which satisfies
\be
    J_{CBAD}^{\Theta} \geq 0,~J_{CB}^{\Theta} = \idop_{CB},~J_{CBA}^{\Theta} = J_{CA}^{\Theta} \ox \pi_{B},
\ee
where the last equation corresponds to the $B$-to-$A$ non-signalling constraint. There is a one-to-one correspondence between a $B$-to-$A$ non-signalling bipartite channel and a superchannel, which is a special case of the quantum comb formalism~\cite{Chiribella2008a}.

\section{Non-signalling bipartite channel and superchannel}\label{sec:ns_channel_superchan}

In quantum information theory, the concept of non-signalling correlation is crucial for understanding the limitations of information transfer between spatially separated parties. This section introduces two important structures that incorporate non-signalling constraints: non-signalling bipartite quantum channels and non-signalling bipartite superchannels. Non-signalling bipartite channels are quantum operations that act on two separate parties, Alice and Bob, in such a way that neither party's output will be affected by the other party's input, meaning they cannot signal each other by using such a quantum channel. Building upon this, we then explore non-signalling bipartite superchannels, which are higher-order quantum transformations for channels while preserving the two parties' non-signalling properties. These superchannels play a critical role in understanding the manipulation and transformation of non-signalling correlations in quantum information processing tasks.

\subsection{Non-signalling bipartite channel}

A bipartite channel $\cN_{A_0B_0\rightarrow A_1B_1}$ is said to be \textit{non-signalling from Alice to Bob} if Alice cannot send classical information to Bob by using it. It is formally defined as follows~\cite{Piani_2006,khatri2024}.

\begin{definition}[A to B non-signalling channel]\label{def:ABnon_sig}
A bipartite channel $\cN_{A_0B_0\rightarrow A_1B_1}$ is called a non-signalling channel from Alice to Bob if
\be\label{Eq:ABns_cond}
    \cR_{A_1}^{\pi}\circ \cN_{A_0B_0\rightarrow A_1B_1} = \cR_{A_1}^{\pi}\circ\cN_{A_0B_0\rightarrow A_1B_1}\circ \cR_{A_0}^{\pi},
\ee
where $\cR_{A_0}^{\pi}$ is a replacement channel defined as $\cR_{A_0}^{\pi}(\cdot)\coloneqq \tr_{A_0}[\cdot]\pi_{A_0}$ with $\pi_{A_0}\coloneqq \idop_{A_0}/d_{A_0}$ being the maximally mixed state on system $A_0$.
\end{definition}
The condition in Eq.~\eqref{Eq:ABns_cond} states that the reduced state on Bob's output system $B_1$ has no dependence on Alice's input system $A_0$.
Therefore, Alice cannot use $\cN_{A_0B_0\rightarrow A_1B_1}$ to signal Bob. It can be verified that Definition~\ref{def:ABnon_sig} is equivalent to that given in Ref.~\cite{Duan2016}, i.e., for any density operator $\rho_{A_0}^{(0)},\rho_{A_0}^{(1)}\in\mathscr{D}(\cH_{A_0})$ and $\sigma_{B_0}\in\mathscr{D}(\cH_{B_0}), \tr_{A_1}[\cN_{A_0B_0\rightarrow A_1B_1}(\rho_{A_0}^{(0)} \ox \sigma_{B_0})] = \tr_{A_1}[\cN_{A_0B_0\rightarrow A_1B_1}(\rho_{A_0}^{(1)} \ox \sigma_{B_0})]$. Furthermore, $\cN_{A_0B_0\rightarrow A_1B_1}$ is non-signalling from Alice to Bob if and only if its \Choi operator satisfies~\cite{khatri2024}
\be
    \tr_{A_1}\big[J_{A_0A_1B_0B_1}^{\cN}\big] = \pi_{A_0}\ox \tr_{A_0A_1}\big[J_{A_0A_1B_0B_1}^{\cN}\big].
\ee
For simplicity, we will denote $\tr_{A_0A_1}\big[J_{A_0A_1B_0B_1}^{\cN}\big]$ by $J_{B_0B_1}^{\cN}$. We further denote by $\NSO^{\rightarrow}(A_0B_0,A_1B_1)$ the set of all $A$-to-$B$ non-signalling channels, i.e.,
\be
    \NSO^{\rightarrow}(A_0B_0,A_1B_1) \coloneqq \bigg\{\cN \in\CPTP(A_0B_0, A_1B_1)~\Big|~J_{A_0B_0B_1}^{\cN} = \pi_{A_0}\ox J_{B_0B_1}^{\cN}\bigg\}.
\ee
A one-way LOCC channel (1-LOCC) from Bob to Alice is an interesting example of a bipartite channel that is non-signalling from Alice to Bob. Similarly, a non-signalling channel from Bob to Alice is defined as follows.
\begin{definition}[B to A non-signalling channel]
A bipartite channel $\cN_{A_0B_0\rightarrow A_1B_1}$ is called a non-signalling channel from Bob to Alice if
\be\label{Eq:BAns_cond}
    \cR_{B_1}^{\pi}\circ \cN_{A_0B_0\rightarrow A_1B_1} = \cR_{B_1}^{\pi}\circ\cN_{A_0B_0\rightarrow A_1B_1}\circ \cR_{B_0}^{\pi},
\ee
where $\cR_{B_0}^{\pi}$ is a replacement channel defined as $\cR_{B_0}^{\pi}(\cdot)\coloneqq \tr_{B_0}[\cdot]\pi_{B_0}$ with $\pi_{B_0}\coloneqq \idop_{B_0}/d_{B_0}$ being the maximally mixed state on system $B_0$.
\end{definition}
Equivalently, $\cN_{A_0B_0\rightarrow A_1B_1}$ is non-signalling from Bob to Alice if and only if its \Choi operator satisfies $J_{A_0A_1B_0}^{\cN} = \pi_{B_0}\ox J_{A_0A_1}^{\cN}$.

A bipartite quantum channel $\cN_{A_0B_0\rightarrow A_1B_1}$ is called a non-signalling channel if it is non-signalling both from Alice to Bob and from Bob to Alice. We denote by $\NSO(A_0B_0,A_1B_1)$ the set of all bipartite bidirectional non-signalling channels, i.e.,
\begin{equation*}
    \NSO(A_0B_0,A_1B_1) \coloneqq \bigg\{\cN\in\CPTP(A_0B_0, A_1B_1)~\Big|~ J_{A_0B_0B_1}^{\cN} = \pi_{A_0}\ox J_{B_0B_1}^{\cN},\,
    J_{A_0A_1B_0}^{\cN} = \pi_{B_0}\ox J_{A_0A_1}^{\cN}\bigg\}.
\end{equation*}
It is easy to check that $\NSO(A_0B_0,A_1B_1)$ is a convex set, and we will abbreviate it as $\NSO$ when there is no ambiguity about the systems we are considering. A property of non-signalling channels is that the composition of non-signalling channels is still non-signalling.

\begin{lemma}[Composition of non-signalling channels]\label{lem:NS_comp}
For two non-signalling channels $\cM_{A_0B_0\rightarrow A_1B_1}$ and $\cN_{A_1B_1\rightarrow A_2B_2}$, their composition
\be
    \cP_{A_0B_0\rightarrow A_2B_2}\coloneqq \cN_{A_1B_1\rightarrow A_2B_2}\circ\cM_{A_0B_0\rightarrow A_1B_1}
\ee
is also a non-signalling channel.
\end{lemma}
\begin{proof}
Since $\cN_{A_1B_1\rightarrow A_2B_2}$ is a non-signalling channel, we have
\be\label{Eq:N2_NS}
    \cR^{\pi}_{A_2}\circ\cN_{A_1B_1\rightarrow A_2B_2} = \cR^{\pi}_{A_2} \circ \cN_{A_1B_1\rightarrow A_2B_2}\circ \cR^{\pi}_{A_1},
\ee
which yields
\bea
    \cR^{\pi}_{A_2}\circ \cN_{A_1B_1\rightarrow A_2B_2} \circ \cM_{A_0B_0\rightarrow A_1B_1} &= \cR^{\pi}_{A_2}\circ \cN_{A_1B_1\rightarrow A_2B_2} \circ \cR^{\pi}_{A_1}\circ \cM_{A_0B_0\rightarrow A_1B_1}\\
    &= \cR^{\pi}_{A_2}\circ\cN_{A_1B_1\rightarrow A_2B_2} \circ \cR^{\pi}_{A_1}\circ \cM_{A_0B_0\rightarrow A_1B_1}\circ \cR^{\pi}_{A_0},
\eea
where we used the fact that $\cM_{A_0B_0\rightarrow A_1B_1}$ is non-signalling in the second equality. Then using Eq.~\eqref{Eq:N2_NS} again, we have
\be
    \cR^{\pi}_{A_2}\circ\cP_{A_0B_0\rightarrow A_2B_2} = \cR^{\pi}_{A_2} \circ \cP_{A_0B_0\rightarrow A_2B_2}\circ \cR^{\pi}_{A_0},
\ee
which yields $\cP_{A_0B_0\rightarrow A_2B_2}$ is a non-signalling channel from Alice to Bob. Similarly, one can show that $\cP_{A_0B_0\rightarrow A_2B_2}$ is also non-signalling from Bob to Alice, thus a non-signalling channel.
\end{proof}

It is quite expected that sequential uses of non-signalling channels still cannot enable signalling between Alice and Bob. Conversely, if $\cP_{A_0B_0\rightarrow A_2B_2}$ is non-signalling, we cannot assert that both $\cM_{A_0B_0\rightarrow A_1B_1}$ and $\cN_{A_1B_1\rightarrow A_2B_2}$ are non-signalling channels. A simple example is that $\cN_{A_1B_1\rightarrow A_2B_2}$ is a replacement channel, which yields $\cP_{A_0B_0\rightarrow A_2B_2}$ a replacement channel, thus non-signalling even if $\cM_{A_0B_0\rightarrow A_1B_1}$ is signalling.

\begin{figure}[t]
    \centering
    \includegraphics[width=0.85\linewidth]{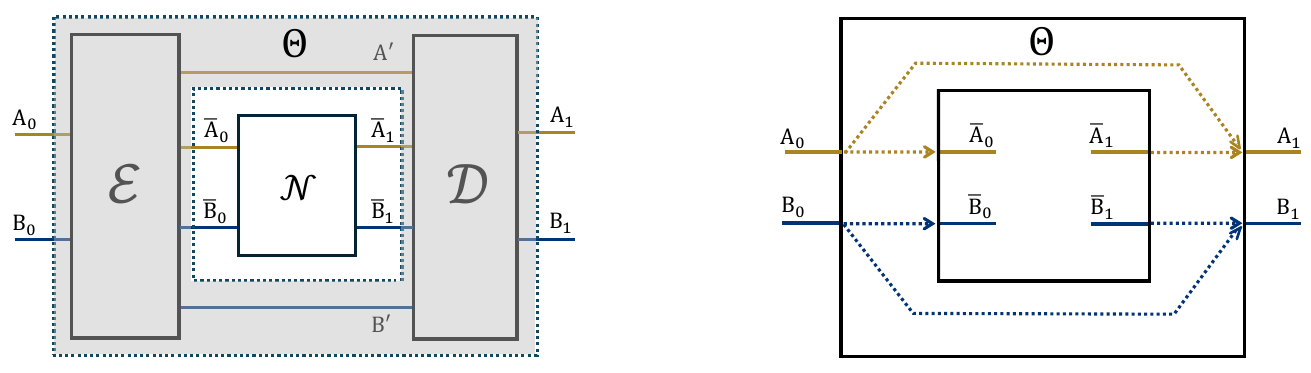}
    \caption{The left side shows the manipulation of a bipartite quantum channel $\cN_{\bA_0\bB_0\rightarrow \bA_1\bB_1}$ via a bipartite superchannel $\Theta_{(\bA_0\bB_0\rightarrow \bA_1\bB_1)\rightarrow(A_0B_0\rightarrow A_1B_1)}$, which can be implemented by a pre-processing channel $\cE_{A_0B_0\rightarrow A'\bA_0\bB_0 B'}$ and a post-processing channel $\cD_{A'\bA_1\bB_1 B'\rightarrow A_1B_1}$. The right side illustrates a non-signalling bipartite superchannel where the input on subsystem $A_0$ can only signal to the outputs on subsystems $\bA_0, A_1$, and the input on subsystem $\bA_1$ can only signal to the output on subsystem $A_1$. Similarly, the input on subsystem $B_0$ can only signal to the outputs on subsystems $\bB_0, B_1$, and the input on subsystem $\bB_1$ can only signal to the output on subsystem $B_1$.}
    \label{fig:bi_superchan}
\end{figure}

\subsection{Non-signalling bipartite superchannel}
A bipartite superchannel $\Theta_{(\bA_0\bB_0\rightarrow \bA_1\bB_1)\rightarrow(A_0B_0\rightarrow A_1B_1)}$ between Alice and Bob, as a higher-order quantum operation, describes the manipulation of bipartite quantum channels. It is natural to consider scenarios where spatially separated parties perform higher-order transformations on quantum channels while respecting local causality constraints.

To explore the full potential of non-signalling correlations, suppose Alice and Bob have access to any quantum resources that cannot be used for direct communication between them. By combining this constraint with a definite causal structure, Alice's input can only affect her own future subsystems ($A_0$ influences $\bA_0$ and $A_1$, while $\bA_1$ affects only $A_1$), and similarly for Bob's subsystems (see Figure~\ref{fig:bi_superchan}). In this context, we introduce the notion of \textit{non-signalling bipartite superchannels}, which provide a rigorous framework for studying distributed quantum information processing tasks, quantum communication protocols, and fundamental limits of causally restricted higher-order quantum operations. Before the formal definition, we remark on how we treat a replacement channel as a superchannel. Given a quantum channel $\cN_{A\rightarrow B}$, we denote
\begin{equation}
    \cR_{A}^{\pi}(\cN_{A\rightarrow B}) \coloneqq \cN_{A\rightarrow B}\circ\cR_{A}^{\pi},\quad \cR_{B}^{\pi}(\cN_{A\rightarrow B}) \coloneqq \cR_{B}^{\pi}\circ\cN_{A\rightarrow B}
\end{equation}
when we treat replacement channels as superchannels.

\begin{definition}[Non-signalling bipartite superchannel]\label{def:NSBSC}
Let $\Theta_{(\bA_0\bB_0\rightarrow \bA_1\bB_1)\rightarrow(A_0B_0\rightarrow A_1B_1)}$ be a superchannel that maps bipartite channels in $\CPTP(\bA_0\bB_0, \bA_1\bB_1)$ to bipartite channels in $\CPTP(A_0B_0, A_1B_1)$. Then it is called a non-signalling bipartite superchannel if
\be
\cR^{\pi}_{A_0A_1}\circ \Theta \circ \cR^{\pi}_{\bA_0} = \cR^{\pi}_{A_1}\circ \Theta \circ \cR^{\pi}_{\bA_0},\quad \cR^{\pi}_{A_1}\circ \Theta \circ \cR^{\pi}_{\bA_1} = \cR^{\pi}_{A_1}\circ \Theta,\label{Eq:nscombA}\\
\ee
and
\be
\cR^{\pi}_{B_0B_1}\circ \Theta \circ \cR^{\pi}_{\bB_0} = \cR^{\pi}_{B_1}\circ \Theta \circ \cR^{\pi}_{\bB_0},\quad \cR^{\pi}_{B_1}\circ \Theta \circ \cR^{\pi}_{\bB_1} = \cR^{\pi}_{B_1}\circ \Theta.\label{Eq:nscombB}
\ee
\end{definition}
It is easy to check that a superchannel $\Theta_{(\bA_0\bB_0\rightarrow \bA_1\bB_1)\rightarrow(A_0B_0\rightarrow A_1B_1)}$ is a non-signalling bipartite superchannel if and only if its \Choi operator satisfies
\be\label{Eq:Theta_choi_NS}
\begin{aligned}
& J_{A_0 \bA_1 \mathbf{B}}^{\Theta} = \pi_{A_0}\ox J_{\bA_1 \mathbf{B}}^{\Theta}, \quad J_{A_0 \bA_0\bA_1 \mathbf{B}}^{\Theta} = \pi_{\bA_1}\ox J_{A_0 \bA_0\mathbf{B}}^{\Theta},\\
& J_{\mathbf{A} B_0 \bB_1}^{\Theta} = \pi_{B_0}\ox J_{\bB_1 \mathbf{A}}^{\Theta}, \quad J_{\mathbf{A} B_0 \bB_0\bB_1}^{\Theta} = \pi_{\bB_1}\ox J_{\mathbf{A} B_0 \bB_0}^{\Theta}.
\end{aligned}
\ee
We note that a typical non-signalling bipartite superchannel is a superchannel realized by a pre-processing non-signalling bipartite channel and a post-processing non-signalling bipartite channel. We also define $\Theta$ as an $A$-to-$B$ non-signalling bipartite superchannel if it satisfies only Eq.~\eqref{Eq:nscombA}, and a $B$-to-$A$ non-signalling bipartite superchannel if it satisfies only Eq.~\eqref{Eq:nscombB}.

\begin{remark}
The causality condition of a superchannel already ensures $J_{A_0\bA_0\bA_1 B_0\bB_0\bB_1}^{\Theta} = \pi_{\bA_1\bB_1}\ox J_{A_0\bA_0 B_0\bB_0}^{\Theta}$, which equivalently states that $\bA_1$ cannot signal $\bA_0$ and $\bB_0$. Noticing that the condition $J_{A_0 \bA_1 \mathbf{B}}^{\Theta} = \pi_{A_0}\ox J_{\bA_1 \mathbf{B}}^{\Theta}$ in the above definition ensures $\bA_1$ cannot signal $B_1$, careful readers might question the necessity of the condition $J_{A_0 \bA_0\bA_1 \mathbf{B}}^{\Theta} = \pi_{\bA_1}\ox J_{A_0 \bA_0\mathbf{B}}^{\Theta}$. However, we note that this condition is indeed crucial considering the non-signalling structures, as there exist scenarios where $\bA_1$ cannot signal $\bA_0$ and $B_1$ individually, yet can signal the composite subsystem $\bA_0B_1$ as a whole. 
\end{remark}

Having established the concept of non-signalling bipartite superchannels, which describes the physical limits of bipartite higher-order transformations, we now turn our attention to a more abstract yet meaningful notion. For the manipulation of bipartite channels, if the output channel of a bipartite superchannel is a non-signalling channel for every input non-signalling channel, we call this superchannel a \textit{non-signalling-preserving (NS-preserving) superchannel}.

\begin{definition}[Completely NS-preserving superchannel]
Let $\Theta_{(\bA_0\bB_0\rightarrow \bA_1\bB_1)\rightarrow(A_0B_0\rightarrow A_1B_1)}$ be a superchannel that maps bipartite channels in $\CPTP(\bA_0\bB_0, \bA_1\bB_1)$ to bipartite channels in $\CPTP(A_0B_0, A_1B_1)$.
\begin{enumerate}
    \item $\Theta$ is called an NS-preserving superchannel if, for any $\cN_{\bA_0\bB_0\rightarrow \bA_1\bB_1}\in\NSO(\bA_0\bB_0,\bA_1\bB_1)$, $\Theta(\cN)$ is a non-signalling channel, i.e., $\Theta(\cN)\in \NSO(A_0B_0,A_1B_1)$.
    \item $\Theta$ is called a completely NS-preserving superchannel (CNSP), if $\mathbb{I}_{A'B'}\ox \Theta$ is an NS-preserving superchannel for any bipartite system $A'A:B' B$, where $\mathbb{I}_{A'B'}$ denotes the identity superchannel on $A'B'$.
\end{enumerate}
\end{definition}
This definition is straightforward from the perspective of quantum dynamics resource theories~\cite{Saxena_2020,Gilad_2020}. When non-signalling bipartite channels are considered free resources, it adopts the perspective that higher-order transformations can only send free channels to free channels, without generating resources. The CNSP constraint ensures the preservation of non-signalling structures between Alice and Bob and with respect to any additional environments of each party, meaning the superchannel can be freely applied to subsystems of larger non-signalling structures. In the following, we shall see that every non-signalling bipartite superchannel is CNSP.

\begin{figure}[t]
    \centering
    \includegraphics[width=1\linewidth]{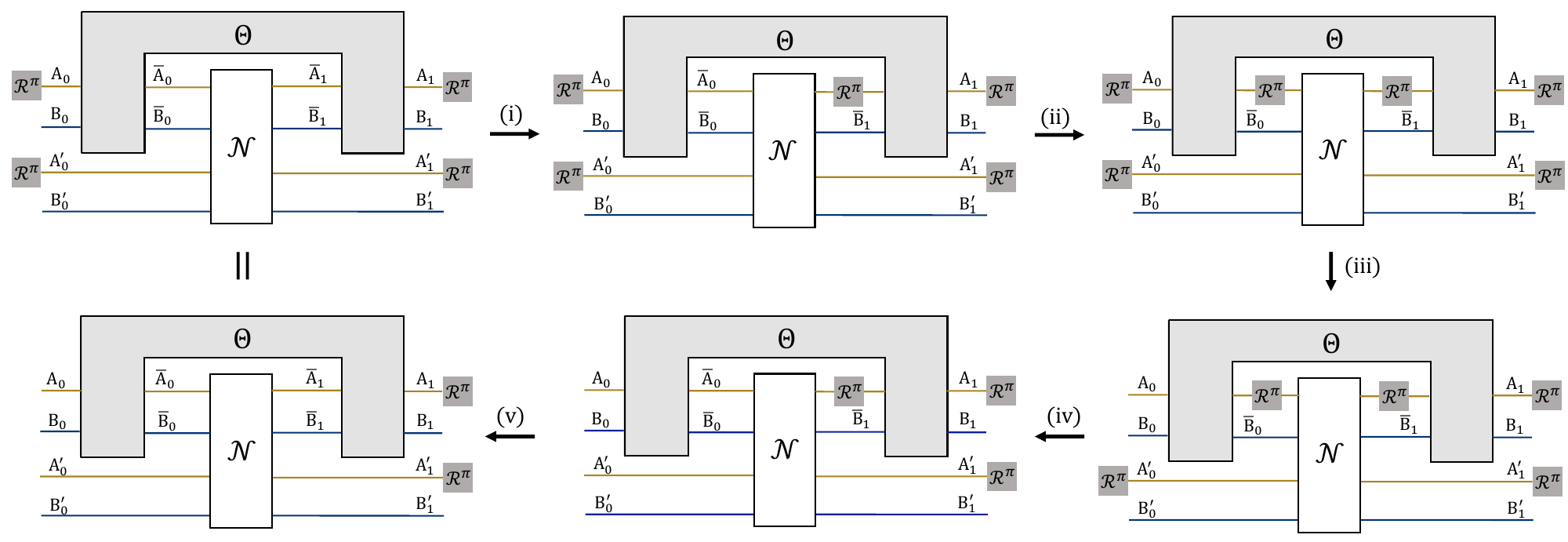}
    \caption{The sketch of different steps in the proof of Theorem~\ref{thm:NSSC_is_CNSP}.}
    \label{fig:if_part}
\end{figure}

\begin{theorem}\label{thm:NSSC_is_CNSP}
Every non-signalling bipartite superchannel $\Theta_{(\bA_0\bB_0\rightarrow \bA_1\bB_1)\rightarrow(A_0B_0\rightarrow A_1B_1)}$ is CNSP.
\end{theorem}
\begin{proof}
Denote 
\be
\cM_{A'_0 B'_0A_0B_0 \rightarrow A'_1 B'_1 A_1 B_1} \coloneqq (\mathbb{I}_{A'B'} \ox \Theta)(\cN_{A'_0 B'_0 \bA_0\bB_0\rightarrow A'_1B'_1 \bA_1\bB_1}).
\ee
If $\cN_{A'_0 B'_0 \bA_0\bB_0\rightarrow A'_1B'_1 \bA_1\bB_1}$ is a non-signalling channel, consider that 
\bea\label{Eq:if_part}
    &\quad\; \cR^{\pi}_{A_1 A'_1}\circ \cM_{A'_0 B'_0A_0B_0 \rightarrow A'_1 B'_1 A_1 B_1} \circ\cR^{\pi}_{A_0A'_0}\\
    &\eqt{(i)} \cR^{\pi}_{A_1}\circ (\mathbb{I}_{A'B'} \ox \Theta)(\cR^{\pi}_{\bA_1 A'_1}\circ\cN_{A'_0 B'_0 \bA_0\bB_0\rightarrow A'_1B'_1 \bA_1\bB_1})\circ\cR^{\pi}_{A_0A'_0}\\
    &\eqt{(ii)} \cR^{\pi}_{A_1}\circ (\mathbb{I}_{A'B'} \ox \Theta)(\cR^{\pi}_{\bA_1 A'_1}\circ\cN_{A'_0 B'_0 \bA_0\bB_0\rightarrow A'_1B'_1 \bA_1\bB_1} \circ \cR^{\pi}_{\bA_0})\circ\cR^{\pi}_{A_0A'_0}\\
    &\eqt{(iii)} \cR^{\pi}_{A_1}\circ (\mathbb{I}_{A'B'} \ox \Theta)(\cR^{\pi}_{\bA_1 A'_1}\circ\cN_{A'_0 B'_0 \bA_0\bB_0\rightarrow A'_1B'_1 \bA_1\bB_1} \circ \cR^{\pi}_{A_0})\circ\cR^{\pi}_{A'_0}\\
    &\eqt{(iv)} \cR^{\pi}_{A_1}\circ (\mathbb{I}_{A'B'} \ox \Theta)(\cR^{\pi}_{\bA_1 A'_1}\circ\cN_{A'_0 B'_0 \bA_0\bB_0\rightarrow A'_1B'_1 \bA_1\bB_1})\\
    &\eqt{(v)} \cR^{\pi}_{A_1}\circ (\mathbb{I}_{A'B'} \ox \Theta)(\cR^{\pi}_{A'_1}\circ\cN_{A'_0 B'_0 \bA_0\bB_0\rightarrow A'_1B'_1 \bA_1\bB_1})\\
    &=\cR^{\pi}_{A_1 A'_1}\circ \cM_{A'_0 B'_0A_0B_0 \rightarrow A'_1 B'_1 A_1 B_1},
\eea
wherein (i) and (v) we used the fact that $\cR^{\pi}_{A_1}\circ \Theta \circ \cR^{\pi}_{\bA_1} = \cR^{\pi}_{A_1}\circ \Theta$; in (ii) and (iv) we used the fact that $\cN_{A'_0 B'_0 A_0B_0\rightarrow A'_1B'_1 A_1B_1}$ is a non-signalling channel; in (iii) we used the fact that $\cR^{\pi}_{A_1}\circ \Theta \circ \cR^{\pi}_{\bA_1} = \cR^{\pi}_{A_1}\circ \Theta$ and $\cR^{\pi}_{A_0A_1}\circ \Theta \circ \cR^{\pi}_{\bA_0} = \cR^{\pi}_{A_1}\circ \Theta \circ \cR^{\pi}_{\bA_0}$ will give
\be
R^{\pi}_{A_1\bA_0}\circ\Theta\circ\cR_{\bA_1A_0} = \cR^{\pi}_{A_1\bA_0}\circ\Theta\circ\cR^{\pi}_{\bA_1}.
\ee
The sketch of different steps in Eq.~\eqref{Eq:if_part} is shown in Figure~\ref{fig:if_part}. Hence, we know that $\cM_{\bA_0\bB_0A'_0 B'_0\rightarrow \bA_1\bB_1 A'_1 B'_1}$ is non-signalling from $\mathbf{A}$ to $\mathbf{B}$ whenever $\cN_{A'_0 B'_0 A_0B_0\rightarrow A'_1B'_1 A_1B_1}$ is a non-signalling channel. Furthermore, since $\cR^{\pi}_{\bB_1B_0}\circ \Theta \circ \cR^{\pi}_{\bB_0 B_1} = \cR^{\pi}_{\bB_1B_0}\circ \Theta$ and $\cR^{\pi}_{\bB_1}\circ \Theta \circ \cR^{\pi}_{B_1} = \cR^{\pi}_{\bB_1}\circ \Theta$, we can have a similar analysis to see that $\cM_{\bA_0\bB_0A'_0 B'_0\rightarrow \bA_1\bB_1 A'_1 B'_1}$ is non-signalling from $\mathbf{B}$ to $\mathbf{A}$. Thus, $\Theta$ is a CNSP superchannel.
\end{proof}

Theorem~\ref{thm:NSSC_is_CNSP} reveals that non-signalling bipartite superchannels, while motivated by physical constraints, also emerge as fundamental free operations in the framework of dynamical resource theories. These superchannels cannot transform a free channel (a non-signalling channel) into a resourceful one (a channel enabling signalling between parties), aligning with the core principle of resource theories wherein free operations should not generate resources. Consequently, non-signalling bipartite superchannels offer a powerful approach to analyzing the classical communication cost in scenarios where non-signalling channels are available. These higher-order transformations provide a natural foundation for investigating the interplay between bipartite quantum channels and classical communication, a theme that we will explore extensively in subsequent sections.

\section{Resource measures of classical communication}\label{sec:measures}

In this section, we define a general measure for the classical communication resource of bipartite quantum channels. In the language of quantum resource theories and dynamical resource theories~\cite{Gour_2019a}, we consider classical communication between two parties as a resource. Consequently, all non-signalling bipartite channels are treated as free objects, encompassing all quantum resources except communication. We note that, as a related work, Ref.~\cite{Milz_2022} considers the resource theory of causal connection where their free objects are non-signalling process matrices, which are different from non-signalling bipartite channels.

To define a resource measure, recall that a generalized divergence $\mathbf{D}(\rho\|\sigma)$ is a function of a quantum state $\rho$ and a positive semidefinite operator $\sigma$ that obeys data-processing inequality~\cite{Polyanskiy2010,Sharma_2013}. Based on this, the corresponding generalized channel divergence~\cite{Leditzky_2018} provides a way of quantifying the distinguishability of two quantum channels $\cN_{A\rightarrow B}$ and $\cM_{A\rightarrow B}$. The generalized channel divergence is defined as
\be
    \mathbf{D}(\cN\| \cM) \coloneqq \sup_{\psi_{AR}} \mathbf{D}\big(\cN_{A\rightarrow B}(\psi_{AR})\|\cM_{A\rightarrow B}(\psi_{AR})\big),
\ee
where the maximization ranges over all pure states $\psi_{AR}$ such that the reference system $R$ is isomorphic to the channel input system $A$.

\begin{definition}[Generalized channel divergence of classical communication]
Let $\cN_{A_0B_0\rightarrow A_1B_1}$ be a bipartite quantum channel. The generalized divergence of one-way classical communication of $\cN_{A_0B_0\rightarrow A_1B_1}$ is defined as
\be
    \mathfrak{D}^{\rightarrow}(\cN) \coloneqq \inf_{\cE\in\NSO^{\rightarrow}} \mathbf{D}(\cN\|\cE),
\ee
where the optimization is with respect to all non-signalling channels in $\NSO^{\rightarrow}(A_0B_0, A_1B_1)$. The generalized divergence of bidirectional classical communication is defined as
\be
    \mathfrak{D}^{\leftrightarrow}(\cN) \coloneqq \inf_{\cE\in\NSO} \mathbf{D}(\cN\|\cE).
\ee
\end{definition}

The generalized divergence of bidirectional communication satisfies monotonicity, i.e., it is nonincreasing under non-signalling-preserving superchannels, which include CNSP superchannels and non-signalling bipartite superchannels.

\begin{lemma}[Monotonicity]\label{lem:monoto}
For every bipartite quantum channel $\cN_{\bA_0\bB_0\rightarrow \bA_1\bB_1}$ and every NS-preserving superchannel $\Theta_{(\bA_0\bB_0\rightarrow \bA_1\bB_1)\rightarrow(A_0B_0\rightarrow A_1 B_1)}$, it satisfies
\be
    \mathfrak{D}^{\leftrightarrow}(\cN) \geq \mathfrak{D}^{\leftrightarrow}(\Theta(\cN)).
\ee
\end{lemma}
\begin{proof}
Consider that 
\be
    \mathfrak{D}^{\leftrightarrow}(\cN) = \inf_{\cE\in\NSO} \mathbf{D}(\cN\|\cE)  \geqt{(i)} \inf_{\cE\in\NSO} \mathbf{D}(\Theta(\cN)\|\Theta(\cE)) \geqt{(ii)} \inf_{\widehat{\cE}\in\NSO} \mathbf{D}(\Theta(\cN)\|\widehat{\cE}) = \mathfrak{D}^{\leftrightarrow}(\Theta(\cN))
\ee
where in (i) we used the monotonicity of the generalized channel divergence~\cite[Section V-A]{Gour_2019}, and (ii) is because $\Theta$ is non-signalling-preserving.
\end{proof}

As a particular case of generalized channel divergences, we consider the channel max-relative entropy~\cite{Cooney_2016}, given by
\be
D_{\max }(\cN \| \cE)\coloneqq\max_{\psi_{R A}} D_{\max}\big(\cN_{A \rightarrow B}\left(\psi_{R A}\right) \| \cE_{A \rightarrow B}\left(\psi_{R A}\right)\big),
\ee
where the max-relative entropy of states is defined as $D_{\max}(\rho \| \sigma) \coloneqq \log \min \{t~|~\rho \leq t\sigma\}$.
It is shown that~\cite{Wilde_2020}
\be\label{eq:chan_max_choi}
D_{\max }(\cN \| \cE)=D_{\max }\big(\cN_{A \rightarrow B}\left(\Phi_{R A}\right) \| \cE_{A \rightarrow B}\left(\Phi_{R A}\right)\big)=\log \min \left\{t~|~J_{A B}^{\cN} \leq t J_{A B}^{\cE}\right\},
\ee
where $\Phi_{R A}$ is the maximally entangled state and $J_{A B}^{\cN}$ is the \Choi operator of the channel $\cN_{A \rightarrow B}$ and similarly for $J_{A B}^{\cE}$. The logarithms in this work are taken in the base $2$. Eq.~\eqref{eq:chan_max_choi} implies the channel max-relative entropy is additive under tensor products, i.e., $D_{\max }(\cN^{\ox n} \| \cE^{\ox n}) = nD_{\max }(\cN \| \cE)$ for all $n\in\mathbb{N}$.

\begin{definition}[Max-relative entropy of classical communication]
Let $\cN_{A_0B_0\rightarrow A_1B_1}$ be a bipartite quantum channel. The max-relative entropy of one-way classical communication of $\cN_{A_0B_0\rightarrow A_1B_1}$ is defined as
\be
    \mathfrak{D}_{\max}^{\rightarrow}(\cN) \coloneqq \min_{\cE\in\NSO^{\rightarrow}} D_{\max}(\cN\|\cE),
\ee
where the minimization ranges over all one-way non-signalling channels in $\NSO^{\rightarrow}(A_0B_0, A_1B_1)$. The max-relative entropy of bidirectional classical communication of $\cN_{A_0B_0\rightarrow A_1B_1}$ is defined as
\be\label{Eq:dmaxNS_def}
\mathfrak{D}_{\max}^{\leftrightarrow}(\cN) \coloneqq \min_{\cE\in\NSO} D_{\max}(\cN\|\cE).
\ee
\end{definition}
We also define the $\epsilon$-smooth max-relative entropy of classical communication as
\begin{equation}
    \mathfrak{D}_{\max}^{*,\epsilon}(\cN) \coloneqq \min_{\cM: \frac{1}{2}\|\cM-\cN\|_{\diamond}\leq \epsilon}\mathfrak{D}_{\max}^{*}(\cM),
\end{equation}
where $*\in\{\rightarrow,\leftrightarrow\}$ and $\cM$ ranges over all bipartite channels in $\CPTP(A_0B_0,A_1B_1)$.

Given a bipartite quantum channel $\cN_{A_0B_0\rightarrow A_1B_1}$, its max-relative entropy of bidirectional classical communication is given by
\be\label{SDP:dmaxNS}
\begin{aligned}
    \mathfrak{D}_{\max}^{\leftrightarrow}(\cN) = \log\, \min &\;\; \lambda\\
     \textrm{s.t.} &\;\; J_{A_0A_1B_0B_1}^{\cN} \leq Y_{A_0 A_1 B_0 B_1},~Y_{A_0B_0} = \lambda \idop_{A_0B_0},\\
     &\;\; Y_{A_0B_0B_1} = \pi_{A_0}\ox Y_{B_0B_1},~Y_{A_0A_1B_0} = \pi_{B_0}\ox Y_{A_0A_1},
\end{aligned}
\ee
whose dual SDP is
\be
\begin{aligned}
   \log\, \max&\; \tr\Big[M_{A_0A_1B_0B_1} J_{A_0A_1B_0B_1}^{\cN}\Big]\\
     \textrm{s.t.} 
     &\; \tr N_{A_0B_0} = 1,~P_{B_0B_1} = 0,~Q_{A_0A_1} = 0\\
     &\;\; 0\leq M_{A_0A_1B_0B_1} \leq N_{A_0B_0}\ox \idop_{A_1B_1}+ P_{A_0B_0B_1}\ox \idop_{A_1} + Q_{A_0A_1B_0}\ox \idop_{B_1},
\end{aligned}
\ee
We can also have primal and dual SDPs for the max-relative entropy of one-way classical communication by removing a non-signalling condition as 
\be\label{Eq:dmax_oneway_pri}
    \mathfrak{D}_{\max}^{\rightarrow}(\cN) = \log\, \min \big\{\lambda~\big|~J_{A_0A_1B_0B_1}^{\cN} \leq Y_{A_0 A_1 B_0 B_1}, Y_{A_0B_0} = \lambda \idop_{A_0B_0}, Y_{A_0B_0B_1} = \pi_{A_0}\ox Y_{B_0B_1}\big\}.
\ee
and
\be\label{SDP:dmax_oneway_dual}
\begin{aligned}
   \log\, \max\;& \tr\big[M_{A_0A_1B_0B_1} J_{A_0A_1B_0B_1}^{\cN}\big]\\
     \textrm{s.t. } 
     & \tr N_{A_0B_0} = 1,~P_{B_0B_1} = 0,~0\leq M_{A_0A_1B_0B_1} \leq N_{A_0B_0}\ox \idop_{A_1B_1}+ P_{A_0B_0B_1}\ox \idop_{A_1}.
\end{aligned}
\ee
A notable property of the max-relative entropy of one-way classical communication is its additivity with respect to the tensor product of two bipartite channels.

\begin{lemma}[Additivity of $\mathfrak{D}^{\rightarrow}_{\max}$]\label{lem:dmax_oneway_add}
The max-relative entropy of one-way communication is additive with respect to the tensor product of quantum channels, i.e., for every $\cN\in\CPTP(A_0B_0:A_1B_1)$ and $\cM\in\CPTP(A'_0B'_0:A'_1B'_1)$,
\be
    \mathfrak{D}^{\rightarrow}_{\max}(\cN\ox \cM) = \mathfrak{D}^{\rightarrow}_{\max}(\cN) + \mathfrak{D}^{\rightarrow}_{\max}(\cM).
\ee
\end{lemma}

\begin{proof}
The proof utilizes the primal and dual SDPs for $\mathfrak{D}^{\rightarrow}_{\max}(\cdot)$. First, we prove the subadditivity. Suppose the optimal solutions for $\cN$ and $\cM$ considering Eq.~\eqref{Eq:dmax_oneway_pri} are $\cE_1$ and $\cE_2$, respectively. Then we have
\be
\begin{aligned}
    \mathfrak{D}^{\rightarrow}_{\max}(\cN\ox \cM) &\leqt{(i)} D_{\max}(\cN\ox \cM \| \cE_1\ox \cE_2)\\
    &\eqt{(ii)} D_{\max}(\cN\|\cE_1) + D_{\max}(\cM\|\cE_2)\\
    &= \mathfrak{D}^{\rightarrow}_{\max}(\cN) + \mathfrak{D}^{\rightarrow}_{\max}(\cM),
\end{aligned}
\ee
where (i) is due to the fact that $\cE_1\ox \cE_2$ is non-signalling from $\mathbf{A}$ to $\mathbf{B}$, thus giving a feasible solution. In (ii), we have used the additivity of the channel max-relative entropy. Second, we prove the superadditivity. Suppose $\{M_{A_0A_1B_0B_1}, N_{A_0B_0}, P_{A_0B_0B_1}\}$ and $\{M'_{A'_0A'_1B'_0B'_1}, N'_{A'_0B'_0}, P'_{A'_0B'_0B'_1}\}$ are feasible solution of SDP~\eqref{SDP:dmax_oneway_dual} for channels $\cN$ and $\cM$ with objective values $x_1 = \tr\big[M_{A_0A_1B_0B_1} J_{A_0A_1B_0B_1}^{\cN}\big]$ and $x_2 = \tr\big[M'_{A'_0A'_1B'_0B'_1} J_{A'_0A'_1B'_0B'_1}^{\cM}\big]$, respectively. Consider $\widehat{M}_{\mathbf{A}\mathbf{A}'\mathbf{B}\mathbf{B}'} := M_{A_0A_1B_0B_1}\ox M'_{A'_0A'_1B'_0B'_1},~\widehat{N}_{A_0B_0A'_0B'_0} := N_{A_0B_0}\ox N'_{A'_0B'_0}$ and $\widehat{P}_{A_0A'_0\mathbf{B}\mathbf{B}'}:= N_{A_0B_0}\ox P'_{A'_0B'_0B'_1}\ox\idop_{B_1} + P_{A_0B_0B_1}\ox N'_{A'_0B'_0}\ox\idop_{B'_1} + P_{A_0B_0B_1}\ox P'_{A'_0B'_0B'_1}$. It can be checked that $\widehat{M}_{\mathbf{A}\mathbf{A}'\mathbf{B}\mathbf{B}'}\geq 0$,
\be
\begin{aligned}
\widehat{M}_{\mathbf{A}\mathbf{A}'\mathbf{B}\mathbf{B}'} &= M_{A_0A_1B_0B_1}\ox M'_{A'_0A'_1B'_0B'_1}\\
&\leq (N_{A_0B_0}\ox \idop_{A_1B_1}+ P_{A_0B_0B_1}\ox \idop_{A_1}) \ox (N_{A'_0B'_0}\ox \idop_{A'_1B'_1}+ P_{A'_0B'_0B'_1}\ox \idop_{A'_1})\\
&= \widehat{N}_{A_0B_0A'_0B'_0} \ox \idop_{A_1B_1A'_1B'_1} + \widehat{P}_{A_0A'_0\mathbf{B}\mathbf{B}'} \ox \idop_{A_1A'_1}.
\end{aligned}
\ee
As $P_{B_0B_1} = P'_{B'_0B'_1} = 0,$ it follows that $\widehat{P}_{\mathbf{B}\mathbf{B}'} = N_{B_0}\ox P'_{B'_0B'_1}\ox\idop_{B_1} + P_{B_0B_1}\ox N'_{B'_0}\ox\idop_{B'_1} + P_{B_0B_1}\ox P'_{B'_0B'_1} = 0$. Hence, $\{\widehat{M}_{\mathbf{A}\mathbf{A}'\mathbf{B}\mathbf{B}'},\widehat{N}_{A_0B_0A'_0B'_0},\widehat{P}_{A_0A'_0\mathbf{B}\mathbf{B}'}\}$ forms a feasible solution of SDP~\eqref{SDP:dmax_oneway_dual} for the channel $\cN\ox\cM$ with an objective value $x_1x_2$. It follows that $\mathfrak{D}^{\rightarrow}_{\max}(\cN\ox \cM) \geq \mathfrak{D}^{\rightarrow}_{\max}(\cN) + \mathfrak{D}^{\rightarrow}_{\max}(\cM)$. Combining with the subadditivity, we complete the proof.
\end{proof}

The max-relative entropy of bidirectional classical communication satisfies subadditivity with respect to the tensor product of bipartite quantum channels.

\begin{lemma}[Subadditivity of $\mathfrak{D}^{\leftrightarrow}_{\max}$]\label{lem:dmax_subadd}
The max-relative entropy of bidirectional communication is subadditive with respect to the tensor product of quantum channels, i.e., for every $\cN\in\CPTP(A_0B_0:A_1B_1)$ and $\cM\in\CPTP(A'_0B'_0:A'_1B'_1)$,
\be
    \mathfrak{D}^{\leftrightarrow}_{\max}(\cN\ox \cM) \leq \mathfrak{D}^{\leftrightarrow}_{\max}(\cN) + \mathfrak{D}^{\leftrightarrow}_{\max}(\cM).
\ee
\end{lemma}
\begin{proof}
The proof follows from that in Lemma~\ref{lem:dmax_oneway_add}.
\end{proof}

Recalling that the max-relative entropy is interlinked with the robustness of resources, we introduce the robustness of classical communication as follows.
\begin{definition}[Robustness of classical communication]
Let $\cN_{A_0B_0\rightarrow A_1B_1}$ be a bipartite quantum channel. The robustness of bidirectional classical communication of $\cN_{A_0B_0\rightarrow A_1B_1}$ is defined as
\be
    R_{\mathrm{NS}}(\cN) \coloneqq \min\Big\{ r \geq 0~\Big|~\frac{\cN+r\cM}{1+r} \in \NSO(A_0B_0,A_1B_1),~\cM\in\CPTP(A_0B_0,A_1B_1) \Big\}.
\ee
\end{definition}
The robustness of bidirectional classical communication has a relationship with the max-relative entropy of bidirectional classical communication as 
\be
    \mathfrak{D}_{\max}^{\leftrightarrow}(\cN) = \log(1+R_{\mathrm{NS}}(\cN)).
\ee
It is worth noting that there is no need to define the \textit{standard} robustness of classical communication in a similar manner:
\be
    \min\Big\{ r \geq 0~\Big|~\frac{\cN+r\cM}{1+r} \in \NSO(A_0B_0,A_1B_1),~\cM\in\NSO(A_0B_0,A_1B_1) \Big\}.
\ee
This is because, due to the linearity of non-signalling constraints, there is no feasible $\cM\in \NSO$ that can make $\cN+r\cM$ non-signalling unless $\cN$ is already non-signalling itself.

Another entropic quantifier of quantum channels is the relative entropy of quantum channels~\cite{Cooney_2016,Leditzky_2018}, given by 
\begin{equation}
    D(\cN\|\cE) \coloneqq \max_{\psi_{RA}}D\big(\cN_{A\rightarrow B}(\psi_{RA}) \| \cE_{A\rightarrow B}(\psi_{RA}) \big),
\end{equation}
where $D(\rho\|\sigma)\coloneqq \tr(\rho\log \rho - \rho\log\sigma)$ is the relative entropy of quantum states. Subsequently, we introduce the relative entropy of classical communication as follows.

\begin{definition}[Relative entropy of classical communication]
Let $\cN_{A_0B_0\rightarrow A_1B_1}$ be a bipartite quantum channel. The relative entropy of one-way classical communication of $\cN_{A_0B_0\rightarrow A_1B_1}$ is defined as
\be
    \mathfrak{D}^{\rightarrow}(\cN) \coloneqq \min_{\cE\in\NSO^{\rightarrow}} D(\cN\|\cE),
\ee
where the minimization ranges over all one-way non-signalling channel in $\NSO^{\rightarrow}(A_0B_0, A_1B_1)$. The relative entropy of bidirectional classical communication of $\cN_{A_0B_0\rightarrow A_1B_1}$ is defined as
\be\label{Eq:dNS_def}
\mathfrak{D}^{\leftrightarrow}(\cN) \coloneqq \min_{\cE\in\NSO} D(\cN\|\cE).
\ee
\end{definition}
Similar quantities can also be defined by choosing $\mathbf{D}(\rho\|\sigma)$ as the Petz-R\'enyi entropy~\cite{Petz1986}, the sandwiched R\'enyi entropy~\cite{M_ller_Lennert_2013, Wilde_2014}, and the the Hilbert $\alpha$-divergences~\cite{Buscemi_2017}.

\section{Classical communication cost of bipartite channel simulation}\label{sec:sim_CNSP}
In this section, we address the problem of quantifying the classical communication resources required to simulate a bipartite quantum channel using non-signalling correlations. Our objective is to determine the minimum number of classical bits that need to be exchanged between Alice and Bob to replicate the action of a target bipartite quantum channel accurately. This analysis assumes access to unlimited non-signalling correlations, which represent the most general form of quantum resources one can exploit without invoking communication, surpassing even unlimited entanglement~\cite{Cubitt_2011,Matthews_2012,Leung2015}. Specifically, we approach this problem by studying the channel simulation task via non-signalling bipartite superchannels.

\begin{definition}[Channel simulation via $A$-to-$B$ non-signalling superchannels]
Given two bipartite quantum channels $\cN_{\bA_0\bB_0\rightarrow \bA_1\bB_1}$ and $\cM_{A_0B_0\rightarrow A_1B_1}$, the minimum error of simulation from $\cN$ to $\cM$ via $A$-to-$B$ non-signalling bipartite superchannels is defined as
\be
    e^{\rightarrow}(\cN, \cM) \coloneqq \frac{1}{2}\inf_{\Theta}\Big\|\Theta(\cN_{\bA_0\bB_0\rightarrow \bA_1\bB_1}) - \cM_{A_0B_0\rightarrow A_1B_1}\Big\|_{\Diamond},
\ee
where $\|\cF\|_{\diamond}\coloneqq \sup_{k\in\mathbb{N}}\sup_{\|X\|_1\leq 1}\|(\cF\ox\id_k)(X)\|_1$ denotes the diamond norm of a linear operator $\cF$ and the minimization ranges over all $A$-to-$B$ non-signalling bipartite superchannels $\Theta_{(\bA_0\bB_0\rightarrow \bA_1\bB_1)\rightarrow(A_0B_0\rightarrow A_1 B_1)}$.
\end{definition}

Similarly, we define $e^{\leftarrow}(\cN, \cM)$ as the minimum error of simulation from $\cN$ to $\cM$ via $B$-to-$A$ non-signalling bipartite superchannels, and $e^{\leftrightarrow}(\cN, \cM)$ as the minimum error when $\Theta$ ranges over all non-signalling bipartite superchannels.

We characterize the simulation cost from both one-shot and asymptotic perspectives. In the one-shot scenario, we show that the $\epsilon$-smooth max-relative entropy of bidirectional classical communication is a lower bound on the $\epsilon$-error simulation cost, and give semidefinite programming (SDP) formulations for the exact bidirectional classical communication cost. For the asymptotic case, we establish a single-letter lower bound on the exact communication cost and relate the vanishing error cost to the relative entropy of bidirectional classical communication.

\subsection{One-shot classical communication cost of bipartite channels}

The classical communication between Alice and Bob can be characterized via a noiseless classical channel. A noiseless classical channel that allows Alice to send $m$ messages to Bob is defined as
\begin{equation}
    \Upsilon_{A\rightarrow B}(\cdot) = \sum_{k=0}^{m-1} \bra{k}(\cdot)\ket{k}\ketbra{k}{k}.
\end{equation}
Consider the bidirectional classical communication between Alice and Bob. If there are $m$ messages they can share with each other at the same time, then the classical communication can be characterized by a bidirectional classical noiseless channel $\Upsilon_m:A_0B_0\rightarrow A_1B_1$ whose \Choi operator is given by 
\be
    C_{A_0B_0A_1B_1} = \sum_{k=0}^{m-1}\sum_{l=0}^{m-1}\ketbra{kk}{kk}_{A_0B_1}\ox\ketbra{ll}{ll}_{B_0A_1}.
\ee
This channel maps each input state $\ketbra{k}{k}_{A_0}\ox\ketbra{l}{l}_{B_0}$ to an output state $\ketbra{l}{l}_{A_1}\ox\ketbra{k}{k}_{B_1}$. In the context of bidirectional classical communication, such a bipartite channel $\Upsilon_m$ is considered a fundamental resource. We aim to study how to simulate an arbitrary bipartite quantum channel using non-signalling superchannels and $\Upsilon_m$. Consequently, we introduce the bidirectional classical communication cost of a bipartite quantum channel as follows.

\begin{definition}[Bidirectional classical communication cost]
Let $\cN_{A_0B_0\rightarrow A_1B_1}$ be a bipartite quantum channel. The one-shot $\epsilon$-error bidirectional classical communication cost of $\cN_{A_0B_0\rightarrow A_1B_1}$ via non-signalling bipartite superchannels is defined as
\be
S_{\leftrightarrow, \epsilon}^{(1)}(\cN)\coloneqq \log\inf\Big\{m\,\Big|\, e^{\leftrightarrow}(\Upsilon_{m}, \cN) \leq \epsilon,\, m\in \mathbb{N} \Big\}.
\ee
\end{definition}

We also consider the one-way classical communication cost where Alice and Bob utilize a unidirectional classical noiseless channel, and their simulation protocol is implemented via an $A$-to-$B$ (or $B$-to-$A$) non-signalling bipartite superchannel. Noticing that $\Upsilon_m = \Upsilon_{A_0\rightarrow B_1}\ox \Upsilon_{B_0\rightarrow A_1}$ and $\Upsilon_{B_0\rightarrow A_1}$ ($\Upsilon_{A_0\rightarrow B_1}$) can be incorporated into the superchannel, we define the one-way communication cost analogously to the bidirectional case as follows.
\begin{definition}[One-way classical communication cost]
Let $\cN_{A_0B_0\rightarrow A_1B_1}$ be a bipartite quantum channel. The one-shot $\epsilon$-error one-way classical communication cost of $\cN_{A_0B_0\rightarrow A_1B_1}$ via $A$-to-$B$ or $B$-to-$A$ non-signalling bipartite superchannels is defined as
\be
S_{*, \epsilon}^{(1)}(\cN)\coloneqq \log\inf\Big\{m\,\Big|\, e^{*}(\Upsilon_{m}, \cN) \leq \epsilon,\, m\in \mathbb{N} \Big\},
\ee
where $*\in\{\rightarrow,\leftarrow\}$.
\end{definition}

Our first result establishes upper and lower bounds on the one-shot $\epsilon$-error one-way classical communication cost via one-way non-signalling bipartite superchannels in terms of the $\epsilon$-smooth max-relative entropy of one-way classical communication.

\begin{theorem}[One-shot $\epsilon$-error one-way classical communication cost]\label{thm:smooth_AtoB}
For a given bipartite quantum channel $\cN_{A_0B_0\rightarrow A_1B_1}$, its one-shot $\epsilon$-error one-way classical communication cost via $A$-to-$B$ or $B$-to-$A$ non-signalling bipartite superchannels satisfies
\begin{equation}
    \mathfrak{D}_{\max}^{*,\epsilon}(\cN) \leq S_{*, \epsilon}^{(1)}(\cN)\leq \mathfrak{D}_{\max}^{*,\epsilon}(\cN) + 1,
\end{equation}
where $*\in\{\rightarrow, \leftarrow\}$.
\end{theorem}
\begin{proof}
    We first prove $\mathfrak{D}_{\max}^{\rightarrow,\epsilon}(\cN) \leq S_{\rightarrow, \epsilon}^{(1)}(\cN)$. Let $m$ be the positive integer satisfying $S_{\rightarrow, \epsilon}^{(1)}(\cN) = \log m$. By the definition of $S_{\rightarrow, \epsilon}^{(1)}(\cN)$, there exists an $A$-to-$B$ non-signalling bipartite superchannel $\Theta$ such that $\frac{1}{2}\|\Theta(\Upsilon_m) - \cN\|_{\diamond}\leq \epsilon$, implying
    \be\label{eq:ABsmooth_low_up_bound_1}
        \mathfrak{D}_{\max}^{\rightarrow,\epsilon}(\cN) = \min_{\cM: \frac{1}{2}\|\cM-\cN\|_{\diamond}\leq \epsilon}\mathfrak{D}_{\max}^{\rightarrow}(\cM) \leq \mathfrak{D}_{\max}^{\rightarrow}(\Theta(\Upsilon_m)).
    \ee
    Let $\cE'\in\CPTP(\bA_0\bB_0,\bA_1\bB_1)$ be a bipartite quantum channel with the \Choi operator
    \be
    J^{\cE'}_{\bA_0\bA_1\bB_0\bB_1} \coloneqq \frac{1}{m} C_{\bA_0\bA_1\bB_0\bB_1} + \frac{1}{m(m-1)} (\idop - C)_{\bA_0\bB_1} \ox (\idop - C)_{\bB_0\bA_1}.
    \ee
    One can see that $\cE'$ is $A$-to-$B$ non-signalling by checking $J^{\cE'}_{\bA_0\bB_0\bB_1} = \frac{1}{m}\idop_{\bA_0\bB_0\bB_1}$.
    For that $\Theta$ is an $A$-to-$B$ non-signalling bipartite superchannel, we know that $\Theta(\cE')$ is also $A$-to-$B$ non-signalling (this is implied by the proof of Theorem~\ref{thm:NSSC_is_CNSP}), and thus
    \be\label{eq:ABsmooth_low_up_bound_2}
    \begin{aligned}
        \mathfrak{D}_{\max}^{\rightarrow}(\Theta(\Upsilon_m)) &= \min_{\cE\in \NSO^{\rightarrow}} D_{\max}(\Theta(\Upsilon_m)\|\cE)\\
        &\leq D_{\max}(\Theta(\Upsilon_m)\|\Theta(\cE'))\\
        &\leqt{(i)} D_{\max}(\Upsilon_m \| \cE')\\
        &= D_{\max}(C_{\bA_0\bA_1\bB_0\bB_1} \| J^{\cE'})\\
        &\leq \log m,
    \end{aligned}
    \ee
    where in (i) we used the data processing inequality of the channel max-relative entropy. Combining Eqs.~\eqref{eq:ABsmooth_low_up_bound_1} and~\eqref{eq:ABsmooth_low_up_bound_2}, we conclude
    \be
        \mathfrak{D}_{\max}^{\rightarrow,\epsilon}(\cN) \leq \log m = S_{\rightarrow, \epsilon}^{(1)}(\cN).
    \ee
    Then we prove $S_{\rightarrow, \epsilon}^{(1)}(\cN) \leq \mathfrak{D}_{\max}^{\rightarrow,\epsilon}(\cN) + 1$. 
    Suppose $\cN'$ is the bipartite quantum channel reaching minimum for $\mathfrak{D}_{\max}^{\rightarrow,\epsilon}(\cN)$ and $\|\cN'-\cN\|_\diamond \leq 2\epsilon$. Let $\mathfrak{D}_{\max}^{\rightarrow}(\cN') = \log \lambda$ such that $J^{\cN'}\leq \lambda \cdot J^{\Bar{\cN'}},~\Bar{\cN'}\in\NSO^{\rightarrow}(A_0B_0,A_1B_1)$. Let $m= \lceil \lambda \rceil, d_{\bA_0} = d_{\bA_1} = d_{\bB_0} = d_{\bB_1} = m,~C_{\bA_0\bB_1} = \sum_{k=1}^m\ketbra{kk}{kk}_{\bA_0\bB_1}$ and all systems $\bA_0,\bB_0,\bA_1,\bB_1$ are isomorphic. Consider a superchannel $\Theta_{(\bA_0\bB_0\rightarrow \bA_1\bB_1)\rightarrow (A_0B_0\rightarrow A_1B_1)}$ such that for any input channel $\cM\in\CPTP(\bA_0\bB_0,\bA_1\bB_1)$,
    \be
    \begin{aligned}
    \Theta(\cM) =\;& \frac{1}{m^2}\tr\big((C_{\bA_0\bB_1}\ox\idop_{\bA_1\bB_0}) J^{\cM}_{\bA_0\bA_1\bB_0\bB_1}\big)\cN' \\
    &+ \Big[1 - \frac{1}{m^2}\tr\big((C_{\bA_0\bB_1}\ox\idop_{\bA_1\bB_0}) J^{\cM}_{\bA_0\bA_1\bB_0\bB_1}\big)\Big]\cdot \frac{m}{m-1}\big(\Bar{\cN'}-\frac{1}{m}\cN'\big).
    \end{aligned}
    \ee
    In the following, we will see that $\Theta$ is an $A$-to-$B$ non-signaling bipartite superchannel. 
    The \Choi operator of $\Theta$ is given by
    \be
    \begin{aligned}
    J^\Theta =\;& \frac{1}{m^2} C_{\bA_0\bB_1}\ox\idop_{\bA_1\bB_0}\ox J^{\cN'}_{A_0A_1B_0B_1} \\
    &+ \frac{1}{m^2}\Big[\idop_{\bA_0\bB_1\bA_1\bB_0} -C_{\bA_0\bB_1}\ox\idop_{\bA_1\bB_0}\Big]\ox \frac{m}{m-1}\big(J^{\Bar{\cN'}}_{A_0A_1B_0B_1}-\frac{1}{m}J^{\cN'}_{A_0A_1B_0B_1}\big).
    \end{aligned}
    \ee
    Note that $J^{\Bar{\cN'}}-\frac{1}{m}J^{\cN'}\geq 0$ as $m\geq \lambda$ and $0\leq C_{\bA_0\bB_1}\ox\idop_{\bA_1\bB_0}\leq \idop_{\bA_0\bA_1\bB_0\bB_1}$. It follows that $J^\Theta\geq 0$. Meanwhile, we have
    \begin{equation}
    \begin{aligned}
    J_{A_0\bA_0 \bA_1 B_0\bB_0 \bB_1}^{\Theta} &= \frac{1}{m^2}C_{\bA_0\bB_1}\ox\idop_{\bA_1\bB_0A_0B_0} + \frac{1}{m^2}\Big[(\idop_{\bA_0\bB_1} -C_{\bA_0\bB_1})\ox\idop_{\bA_1\bB_0}\Big]\ox \idop_{A_0B_0}\\
    &= \frac{1}{m^2}\idop_{\bA_0\bB_1 \bA_1\bB_0 A_0B_0},
    \end{aligned}
    \end{equation}
    which indicates that $\Theta$ is a valid superchannel.
    We can also check that
    \begin{equation}\label{Eq:Choi_P_theta1}
    \begin{aligned}
    J_{A_0\bA_0 \bA_1 \mathbf{B}}^{\Theta} &= \frac{1}{m^2}C_{\bA_0\bB_1}\ox\idop_{\bA_1\bB_0}\ox J^{\cN'}_{A_0B_0B_1}
    + \frac{1}{m^2}\Big[(\idop_{\bA_0\bB_1} -C_{\bA_0\bB_1})\ox\idop_{\bA_1\bB_0}\Big]\ox \frac{m}{m-1}\big(J^{\Bar{\cN'}}_{A_0B_0B_1}-\frac{1}{m}J^{\cN'}_{A_0B_0B_1}\big)\\
    &= \pi_{\bA_1}\ox J_{A_0\bA_0 \mathbf{B}}^{\Theta}
    \end{aligned}
    \end{equation}
    and
    \begin{equation}\label{Eq:Choi_P_theta2}
    \begin{aligned}
    J_{A_0 \bA_1 \mathbf{B}}^{\Theta} &= \frac{1}{m^2}\idop_{\bB_1\bA_1\bB_0}\ox J^{\cN'}_{A_0B_0B_1}
    + \frac{1}{m^2}\Big[(m-1)\idop_{\bB_1\bA_1\bB_0}\Big]\ox \frac{m}{m-1}\big(J^{\Bar{\cN'}}_{A_0B_0B_1}-\frac{1}{m}J^{\cN'}_{A_0B_0B_1}\big)\\
    &=\frac{1}{m^2}\idop_{\bB_1\bA_1\bB_0}\ox J^{\cN'}_{A_0B_0B_1}
    +\frac{1}{m}\idop_{\bB_1\bA_1\bB_0}\ox J^{\Bar{\cN'}}_{A_0B_0B_1} - \frac{1}{m^2}\idop_{\bB_1\bA_1\bB_0}\ox J^{\cN'}_{A_0B_0B_1}\\
    &= \frac{1}{m}\idop_{\bB_1\bA_1\bB_0}\ox J^{\Bar{\cN'}}_{A_0B_0B_1}\\
    &= \frac{1}{m}\idop_{\bB_1\bA_1\bB_0}\ox \pi_{A_0} \ox J^{\Bar{\cN'}}_{B_0B_1},
    \end{aligned}
    \end{equation}
    where we used the fact that $\Bar{\cN'}$ is an $A$-to-$B$ non-signalling channel in the last equality. 
    Eqs.~\eqref{Eq:Choi_P_theta1} and~\eqref{Eq:Choi_P_theta2} are the conditions of $\Theta$ being an $A$-to-$B$ non-signalling bipartite superchannel as shown in Eq.~\eqref{Eq:Theta_choi_NS}.
    Therefore, $\Theta$ is an $A$-to-$B$ non-signalling bipartite superchannel.
    Finally, noticing that $\|\Theta(\Upsilon_m) - \cN\|_{\diamond} = \|\cN' - \cN\|_{\diamond} \leq 2\epsilon$ and $\Theta$ is a feasible simulation protocol, we have
    \begin{equation}
        S_{\rightarrow, \epsilon}^{(1)}(\cN) \leq \log m \leq \log(1+ \lambda) \leq 1+\log\lambda = \mathfrak{D}_{\max}^{\rightarrow,\epsilon}(\cN) + 1.
    \end{equation}
    Hence, we arrive at
    \begin{equation}
        \mathfrak{D}_{\max}^{\rightarrow,\epsilon}(\cN) \leq S_{\rightarrow, \epsilon}^{(1)}(\cN)\leq \mathfrak{D}_{\max}^{\rightarrow,\epsilon}(\cN) + 1.
    \end{equation}

    Now we prove the case for the one-shot $\epsilon$-error cost via $B$-to-$A$ non-signalling bipartite superchannels to arrive at $\mathfrak{D}_{\max}^{\leftarrow,\epsilon}(\cN) \leq S_{\leftarrow, \epsilon}^{(1)}(\cN)\leq \mathfrak{D}_{\max}^{\leftarrow,\epsilon}(\cN) + 1$. The proof of the converse bound remains the same, and the constructed superchannel $\Theta$ for the upper bound is given by 
    \be
    \begin{aligned}
    \Theta(\cM) =\;& \frac{1}{m^2}\tr\big((C_{\bB_0\bA_1}\ox\idop_{\bA_0\bB_1}) J^{\cM}_{\bA_0\bA_1\bB_0\bB_1}\big)\cN' \\
    &+ \Big[1 - \frac{1}{m^2}\tr\big((C_{\bB_0\bA_1}\ox\idop_{\bA_0\bB_1}) J^{\cM}_{\bA_0\bA_1\bB_0\bB_1}\big)\Big]\cdot \frac{m}{m-1}\big(\Bar{\cN'}-\frac{1}{m}\cN'\big).
    \end{aligned}
    \ee
    Followed by a similar derivation, we complete the proof.
\end{proof}

\begin{theorem}[One-shot $\epsilon$-error bidirectional classical communication cost]\label{thm:smooth_low_bound}
For a given bipartite quantum channel $\cN_{A_0B_0\rightarrow A_1B_1}$, its one-shot $\epsilon$-error bidirectional classical communication cost via non-signalling bipartite superchannels satisfies
\begin{equation}
    S_{\leftrightarrow, \epsilon}^{(1)}(\cN) \geq \mathfrak{D}_{\max}^{\leftrightarrow,\epsilon}(\cN). 
\end{equation}
\end{theorem}
\begin{proof}
The proof is similar to that in Theorem~\ref{thm:smooth_AtoB}.
\end{proof}

\vspace{2mm}
How to obtain an upper bound on the $\epsilon$-error simulation cost via non-signalling bipartite superchannels is still open, as we may need other techniques to arrive at an achievability result. Theorem~\ref{thm:smooth_low_bound} also provides a lower bound on entanglement-assisted one-shot exact bidirectional simulation cost when the communicating parties, Alice and Bob, are assisted by unlimited entanglement, which they can share not only between each other but also across the timeline of their own simulation protocol. When extended to multipartite channels, such a setting could serve as a general framework for network communication assisted by quantum entanglement. Let $\epsilon= 0$, we have that $S_{\leftrightarrow,0}^{(1)}(\cN)\geq \mathfrak{D}_{\max}^{\leftrightarrow}(\cN)$. The channel max-relative entropy is also called the log-robustness for point-to-point channels in literature~\cite{liu2019resource,Gour_2019a}. It characterizes the one-shot simulation cost of channels in many cases, including the resource theory of coherence with MIO as free channels~\cite{D_az_2018} and the resource theory of thermal nonequilibrium with Gibbs-preserving maps as free channels~\cite{Faist2018,Faist_2019}. To further characterize the one-shot bidirectional classical communication, we derive the SDP formulation of the one-shot minimum simulation error via non-signalling bipartite superchannels.

\begin{proposition}[Minimum simulation error]\label{prop:min_err_SDP}
For a given bipartite quantum channel $\cN_{A_0B_0\rightarrow A_1B_1}$, the one-shot minimum error of simulation from the bidirectional classical noiseless channel $\Upsilon_m:\bA_0\bB_0\rightarrow \bA_1\bB_1$ to $\cN_{A_0B_0\rightarrow A_1B_1}$ via non-signalling bipartite superchannels is given by the following SDP:
\be
\begin{aligned}
    e^{\leftrightarrow}(\Upsilon_m, \cN) = \min & \;\; \mu \\
    {\rm s.t.} &\; \tr_{A_1B_1} Y_{A_0A_1B_0B_1} \leq \mu \mathds{1}_{A_0B_0},\, Y_{A_0A_1B_0B_1}\geq 0,\\
    &\;\; Y_{A_0A_1B_0B_1} \geq Q_{A_0A_1B_0B_1} - J^{\cN},\, Q_{A_0A_1B_0B_1} \geq 0,\\
    &\;\; Q_{A_0B_0} = \mathds{1}_{A_0B_0},\, V_{A_0A_1B_0B_1}\geq Q_{A_0A_1B_0B_1},\\
    &\;\; W_{A_0A_1B_0B_1}\geq Q_{A_0A_1B_0B_1}, \, X_{A_0A_1B_0B_1}\geq Q_{A_0A_1B_0B_1},\\
    &\;\; X_{A_0B_0B_1} = \pi_{A_0} \ox X_{B_0B_1},\, X_{A_0A_1B_0} = \pi_{B_0} \ox X_{A_0A_1},\\
    &\;\; mQ_{A_0B_0B_1} = W_{A_0B_0B_1}, \; X_{A_0B_0B_1} = V_{A_0B_0B_1},\\
    &\;\; mQ_{A_0A_1B_0} = V_{A_0A_1B_0}, \; X_{A_0A_1B_0}=W_{A_0A_1B_0}.
\end{aligned}
\ee
\end{proposition}

\begin{proof}
For any two quantum channels $\cN_1, \cN_2$ from system $A$ to system $B$, their difference in terms of the diamond norm can be expressed as an SDP as follows~\cite{watrous2009semidefinite}:
\be
    \frac{1}{2}\big\|\cN_1 - \cN_2\big\|_{\Diamond} = \min \big\{\mu \,|\, \tr_B Y \leq \mu \mathds{1}_A,\, Y\geq 0,\, Y\geq J_{\cN_1} - J_{\cN_2} \big\},
\ee
where $J_{\cN_1}$ and $J_{\cN_2}$ are the \Choi operators of $\cN_1$ and $\cN_2$, respectively. It follows that the minimum error $e^{\leftrightarrow}(\Upsilon_m, \cN)$ can be expressed as the following SDP:
\begin{subequations}
\begin{align}
    e^{\leftrightarrow}(\Upsilon_m, \cN) = \min & \;\; \mu\\
    {\rm s.t.} &\; \tr_{A_1B_1} Y_{A_0A_1B_0B_1} \leq \mu \mathds{1}_{A_0B_0},\\
    &\;\; Y_{A_0A_1B_0B_1}\geq 0,\, Y_{A_0A_1B_0B_1} \geq J^{\widetilde{\cN}} - J^{\cN},\\
    &\;\; J^{\widetilde{\cN}} = \tr_{\bA_0\bA_1\bB_0\bB_1}J^{\Theta}_{\mathbf{A}\mathbf{B}}(C_{\bA_0\bB_0\bA_1\bB_1}\ox \mathds{1}_{A_0A_1B_0B_1}) \label{Eq:sdp_JNout}\\
    &\;\; J^{\Theta}_{\mathbf{A}\mathbf{B}} \text{ satisfies Eq.}~\eqref{Eq:condition},
\end{align}
\end{subequations}
where $J^{\Theta}_{\mathbf{A}\mathbf{B}}$ is the \Choi operator of $\Theta$. 
For any $m \times m$ permutations $\tau$ and $\tau'$ in the symmetric group $S_m$, if $J^{\Theta}_{\mathbf{A}\mathbf{B}}$ is feasible, then
\be
    J^{'}_{\mathbf{A}\mathbf{B}} =\left(\tau_{\bA_0} \ox \tau_{\bB_1} \ox \tau'_{\bB_0} \ox \tau'_{\bA_1} \ox \mathds{1}_{A_0A_1B_0B_1}\right) J^{\Theta}_{\mathbf{A}\mathbf{B}} \left(\tau_{\bA_0} \ox \tau_{\bB_1} \ox \tau'_{\bB_0} \ox \tau'_{\bA_1}\ox \mathds{1}_{A_0A_1B_0B_1}\right)^{\dagger}
\ee
is also feasible. Furthermore, if $J'_{\mathbf{A}\mathbf{B}}$ and $J''_{\mathbf{A}\mathbf{B}}$ are feasible, so is any convex combination $\lambda J^{'}_{\mathbf{A}\mathbf{B}} + (1-\lambda) J^{''}_{\mathbf{A}\mathbf{B}}$ for $0 \leq \lambda \leq 1$. With these facts, we can always construct a new feasible $\widetilde{J}^{\Theta}_{\mathbf{A}\mathbf{B}}$ by performing the following twirling operation on a feasible solution $J^{\Theta}_{\mathbf{A}\mathbf{B}}$
\begin{equation*}
\widetilde{J}^{\Theta}_{\mathbf{A}\mathbf{B}} \coloneqq \frac{1}{(m!)^2} \sum_{\tau,\tau' \in S_m}\left(\tau_{\bA_0} \ox \tau_{\bB_1} \ox \tau'_{\bB_0} \ox \tau'_{\bA_1} \ox \mathds{1}_{A_0B_0A_1B_1}\right) J^{\Theta}_{\mathbf{A}\mathbf{B}} \left(\tau_{\bA_0} \ox \tau_{\bB_1} \ox \tau'_{\bB_0} \ox \tau'_{\bA_1}\ox \mathds{1}_{A_0B_0A_1B_1}\right)^{\dagger}.
\end{equation*}
Note that the symmetry group $G\coloneqq\left\{\tau_{\bA_0} \ox \tau_{\bB_1}\ox \tau'_{\bB_0}\ox\tau'_{\bA_1}: \tau, \tau' \in S_m\right\}$ acts independently on two pairs of systems $(\bA_0, \bB_1)$ and $(\bA_1,\bB_0)$. For each pair, the action of the symmetric group $S_m$ decomposes the space into two irreducible subspaces, i.e., the symmetric subspace with a projector $C_{\bA_0\bB_1}$ and the antisymmetric subspace with a projector $(\idop - C)_{\bA_0\bB_1}$. Here $C_{\bA_0\bB_1}= \sum_{k=1}^m\ketbra{kk}{kk}_{\bA_0\bB_1}$ is the \Choi matrix of the noiseless classical channel from system $\bA_0$ to system $\bB_1$, so is $C_{\bB_0\bA_1}$. Since the group $G$ acts independently in two pairs of systems, we can consider combinations of symmetric and antisymmetric subspaces for each pair. By applying Schur's Lemma to the symmetry group $G$, we have that $\widetilde{J}^{\Theta}_{\mathbf{A}\mathbf{B}}$ can be chosen as
\be\label{Eq:new_Ztilde}
\begin{aligned}
    \widetilde{J}^{\Theta}_{\mathbf{A}\mathbf{B}} &= C_{\bA_0\bB_1} \ox C_{\bB_0 \bA_1} \ox E^{(1)}_{A_0A_1B_0B_1} + (\mathds{1}-C)_{\bA_0\bB_1} \ox C_{\bB_0 \bA_1} \ox E^{(2)}_{A_0A_1B_0B_1} \\
    &\qquad + C_{\bA_0\bB_1} \ox (\mathds{1}-C)_{\bB_0 \bA_1} \ox  E^{(3)}_{A_0A_1B_0B_1} + (\mathds{1}-C)_{\bA_0 \bB_1} \ox (\mathds{1} - C)_{\bB_0 \bA_1} \ox E^{(4)}_{A_0A_1B_0B_1}.
\end{aligned}
\ee
Taking this into condition Eq.~\eqref{Eq:sdp_JNout}, we can obtain
\be
    J^{\widetilde{\cN}} = m^2 E^{(1)}_{A_0A_1B_0B_1}.
\ee
By definition, the \Choi operator of $\Theta$ should satisfy
\begin{subequations}\label{Eq:condition}
\begin{align}
    &J^{\Theta}_{\mathbf{A}\mathbf{B}} \geq 0, \quad (\rm CP)\\
    & J^{\Theta}_{A_0\bA_1B_0\bB_1} = \mathds{1}_{A_0\bA_1B_0\bB_1}, \quad (\rm TP)\\
    &J^{\Theta}_{A_0\bA_1 \mathbf{B}} = \pi_{A_0} \ox J^{\Theta}_{\bA_1 \mathbf{B}},\quad J_{A_0 \bA_0\bA_1\mathbf{B}}^{\Theta} = \pi_{\bA_1}\ox J_{A_0 \bA_0\mathbf{B}}^{\Theta}, \quad (\mathbf{A}\nrightarrow \mathbf{B})\\
    &J^{\Theta}_{\mathbf{A} B_0 \bB_1} = \pi_{B_0} \ox J^{\Theta}_{\bB_1 \mathbf{A}},\quad J_{\mathbf{A} B_0 \bB_0\bB_1}^{\Theta} = \pi_{\bB_1}\ox J_{\mathbf{A} B_0 \bB_0}^{\Theta}. \quad (\mathbf{B}\nrightarrow \mathbf{A})
\end{align}
\end{subequations}
We now simplify the constraints in Eq.~\eqref{Eq:condition}.
In this process, we will frequently use the facts that the marginal operator of $C_{\bA_0\bB_1}$ (or $C_{\bB_0\bA_1}$) on one system is the identity operator $\mathds{1}$ on that system and the marginal operator of $(\mathds{1} - C)_{\bA_0\bB_1}$ (or $(\mathds{1} - C)_{\bB_0\bA_1}$) on one system is $(m-1)\mathds{1}$ on that system.

For the CP constraint, we have $E^{(1)}_{A_0A_1B_0B_1} \geq 0,\, E^{(2)}_{A_0A_1B_0B_1}\geq 0, \, E^{(3)}_{A_0A_1B_0B_1}\geq 0, \, E^{(4)}_{A_0A_1B_0B_1}\geq 0$.
For the TP constraint, we have $E^{(1)}_{A_0B_0} + (m-1)(E^{(2)}_{A_0B_0} + E^{(3)}_{A_0B_0}) + (m-1)^2 E^{(4)}_{A_0B_0} = \mathds{1}_{A_0B_0}$.

For the non-signalling condition $\mathbf{A}\nrightarrow \mathbf{B}$, notice that
\begin{equation}
\begin{aligned}
    \widetilde{J}_{A_0\bA_1\mathbf{B}}^{\Theta} &= \idop_{\bB_1} \ox C_{\bB_0 \bA_1} \ox \left(E^{(1)}_{A_0B_0B_1} + (m-1)E^{(2)}_{A_0B_0B_1}\right)\\
    &\qquad + \mathds{1}_{\bB_1} \ox (\mathds{1} - C)_{\bB_0 \bA_1} \ox \left(E^{(3)}_{A_0B_0B_1} + (m-1)E^{(4)}_{A_0B_0B_1}\right)
\end{aligned}
\end{equation}
and
\begin{equation}
\begin{aligned}
    \widetilde{J}_{\bA_1\mathbf{B}}^{\Theta} &= \idop_{\bB_1} \ox C_{\bB_0 \bA_1} \ox \left(E^{(1)}_{B_0B_1} + (m-1)E^{(2)}_{B_0B_1}\right)\\
    &\qquad + \mathds{1}_{\bB_1} \ox (\mathds{1} - C)_{\bB_0 \bA_1} \ox \left(E^{(3)}_{B_0B_1} + (m-1)E^{(4)}_{B_0B_1}\right).
\end{aligned}
\end{equation}
Therefore, to satisfy the constraint $J^{\Theta}_{A_0\bA_1 \mathbf{B}} = \pi_{A_0} \ox J^{\Theta}_{\bA_1 \mathbf{B}}$, it must hold that
\be\label{Eq:condAB1}
\begin{aligned}
    &E^{(1)}_{A_0B_0B_1} + (m-1) E^{(2)}_{A_0B_0B_1} = \pi_{A_0} \ox \left(E^{(1)}_{B_0B_1} + (m-1) E^{(2)}_{B_0B_1}\right),\\
    &E^{(3)}_{A_0B_0B_1} + (m-1) E^{(4)}_{A_0B_0B_1} = \pi_{A_0} \ox \left(E^{(3)}_{B_0B_1} + (m-1) E^{(4)}_{B_0B_1}\right).
\end{aligned}
\ee
Moreover, notice that
\be
\begin{aligned}
    \widetilde{J}_{A_0\bA_0\bA_1\mathbf{B}}^{\Theta} &= C_{\bA_0\bB_1} \ox C_{\bB_0 \bA_1} \ox E^{(1)}_{A_0B_0B_1} + (\mathds{1}-C)_{\bA_0\bB_1} \ox C_{\bB_0 \bA_1} \ox E^{(2)}_{A_0B_0B_1}\\
    &\qquad + C_{\bA_0\bB_1} \ox (\mathds{1}-C)_{\bB_0 \bA_1} \ox  E^{(3)}_{A_0B_0B_1} + (\mathds{1}-C)_{\bA_0 \bB_1} \ox (\mathds{1} - C)_{\bB_0 \bA_1} \ox E^{(4)}_{A_0B_0B_1}\\
    &= C_{\bA_0\bB_1} \ox \Big(\idop_{\bB_0 \bA_1} \ox  E^{(3)}_{A_0B_0B_1} + C_{\bB_0 \bA_1} \ox \big(E^{(1)}_{A_0B_0B_1} - E^{(3)}_{A_0B_0B_1}\big)\Big)\\
    &\qquad + (\mathds{1}-C)_{\bA_0\bB_1} \ox \left(\idop_{\bB_0 \bA_1} \ox  E^{(4)}_{A_0B_0B_1} + C_{\bB_0 \bA_1} \ox \big(E^{(2)}_{A_0B_0B_1} - E^{(4)}_{A_0B_0B_1}\big)\right)
\end{aligned}
\ee
and
\be
\begin{aligned}
    \widetilde{J}_{A_0\bA_0\mathbf{B}}^{\Theta} &= C_{\bA_0\bB_1} \ox \idop_{\bB_0} \ox \left(E^{(1)}_{A_0B_0B_1} + (m-1)E^{(3)}_{A_0B_0B_1}\right) \\
    &\qquad + (\mathds{1}-C)_{\bA_0 \bB_1} \ox \mathds{1}_{\bB_0} \ox \left(E^{(2)}_{A_0B_0B_1} + (m-1)E^{(4)}_{A_0B_0B_1}\right).
\end{aligned}
\ee
To satisfy the constraint $J_{A_0 \bA_0\bA_1\mathbf{B}}^{\Theta} = \pi_{\bA_1}\ox J_{A_0 \bA_0\mathbf{B}}^{\Theta}$, it must hold that 
\begin{equation}\label{Eq:condAB2}
    E^{(1)}_{A_0B_0B_1} = E^{(3)}_{A_0B_0B_1} \quad\text{and}\quad E^{(2)}_{A_0B_0B_1}=E^{(4)}_{A_0B_0B_1}.
\end{equation}
Combining Eqs.~\eqref{Eq:condAB1} and~\eqref{Eq:condAB2}, the $\mathbf{A}\nrightarrow \mathbf{B}$ constraint is equivalent to 
\be
\begin{aligned}
    &E^{(1)}_{A_0B_0B_1} =  E^{(3)}_{A_0B_0B_1},\quad E^{(2)}_{A_0B_0B_1} =  E^{(4)}_{A_0B_0B_1},\\
    &E^{(1)}_{A_0B_0B_1} + (m-1) E^{(4)}_{A_0B_0B_1} = \pi_{A_0} \ox \left(E^{(1)}_{B_0B_1} + (m-1) E^{(4)}_{B_0B_1}\right).
\end{aligned}
\ee

For the non-signalling condition $\mathbf{B}\nrightarrow \mathbf{A}$, following the same steps as those for $\mathbf{A}\nrightarrow \mathbf{B}$, we have that the $\mathbf{B}\nrightarrow \mathbf{A}$ constraint is equivalent to
\be
\begin{aligned}
    & E^{(1)}_{A_0A_1B_0} = E^{(2)}_{A_0A_1B_0},\quad E^{(3)}_{A_0A_1B_0}=E^{(4)}_{A_0A_1B_0},\\
    & E^{(1)}_{A_0A_1B_0} + (m-1) E^{(4)}_{A_0A_1B_0} = \pi_{B_0} \ox (E^{(1)}_{A_0A_1} + (m-1) E^{(4)}_{A_0A_1}).
\end{aligned}
\ee
Note that with the simplified non-signalling constraints, the TP constraint can be further simplified to $m^2 E^{(1)}_{A_0B_0} = \mathds{1}_{A_0B_0}$. 
So far, the SDP has been simplified to
\bea
    \min & \;\; \mu \\
    {\rm s.t.} &\; \tr_{A_1B_1} Y_{A_0A_1B_0B_1} \leq \mu \mathds{1}_{A_0B_0},\, Y_{A_0A_1B_0B_1}\geq 0,\\
    &\;\; Y_{A_0A_1B_0B_1} \geq m^2 E^{(1)}_{A_0A_1B_0B_1} - J^{\cN},\\
    &\;\; E^{(1)}_{A_0A_1B_0B_1} \geq 0,\, m^2 E^{(1)}_{A_0B_0} = \mathds{1}_{A_0B_0},\\
    &\;\; E^{(2)}_{A_0A_1B_0B_1}\geq 0, \, E^{(3)}_{A_0A_1B_0B_1}\geq 0, \, E^{(4)}_{A_0A_1B_0B_1}\geq 0,\\
    &\;\; E^{(1)}_{A_0B_0B_1} + (m-1)E^{(4)}_{A_0B_0B_1} = \pi_{A_0} \ox \left( E^{(1)}_{B_0B_1} + (m-1) E^{(4)}_{B_0B_1}\right),\\
    &\;\; E^{(1)}_{A_0A_1B_0} + (m-1)E^{(4)}_{A_0A_1B_0} = \pi_{B_0} \ox \left(E^{(1)}_{A_0A_1}+ (m-1) E^{(4)}_{A_0A_1}\right),\\
    &\;\; E^{(1)}_{A_0B_0B_1} = E^{(3)}_{A_0B_0B_1}, \; E^{(2)}_{A_0B_0B_1} = E^{(4)}_{A_0B_0B_1},\\
    &\;\; E^{(1)}_{A_0A_1B_0} = E^{(2)}_{A_0A_1B_0}, \; E^{(3)}_{A_0A_1B_0}=E^{(4)}_{A_0A_1B_0}.
\eea
Let $Q_{A_0A_1B_0B_1} \coloneqq m^2 E^{(1)}_{A_0A_1B_0B_1}$, $V_{A_0A_1B_0B_1} \coloneqq m^2 E^{(1)}_{A_0A_1B_0B_1} + m^2(m-1) E^{(2)}_{A_0A_1B_0B_1}$, $W_{A_0A_1B_0B_1} \coloneqq m^2 E^{(1)}_{A_0A_1B_0B_1} + m^2(m-1) E^{(3)}_{A_0A_1B_0B_1}$ and $X_{A_0A_1B_0B_1} \coloneqq m^2 E^{(1)}_{A_0A_1B_0B_1} + m^2(m-1) E^{(4)}_{A_0A_1B_0B_1}$. It follows that the SDP can be rewritten as
\bea
    \min & \;\; \mu \\
    {\rm s.t.} &\; \tr_{A_1B_1} Y_{A_0A_1B_0B_1} \leq \mu \mathds{1}_{A_0B_0},\, Y_{A_0A_1B_0B_1}\geq 0,\\
    &\;\; Y_{A_0A_1B_0B_1} \geq Q_{A_0A_1B_0B_1} - J^{\cN},\, Q_{A_0A_1B_0B_1} \geq 0,\\
    &\;\; Q_{A_0B_0} = \mathds{1}_{A_0B_0},\, V_{A_0A_1B_0B_1}\geq Q_{A_0A_1B_0B_1},\\
    &\;\; W_{A_0A_1B_0B_1}\geq Q_{A_0A_1B_0B_1}, \, X_{A_0A_1B_0B_1}\geq Q_{A_0A_1B_0B_1},\\
    &\;\; X_{A_0B_0B_1} = \pi_{A_0} \ox X_{B_0B_1},\, X_{A_0A_1B_0} = \pi_{B_0} \ox X_{A_0A_1},\\
    &\;\; mQ_{A_0B_0B_1} = W_{A_0B_0B_1}, \; X_{A_0B_0B_1} = V_{A_0B_0B_1},\\
    &\;\; mQ_{A_0A_1B_0} = V_{A_0A_1B_0}, \; X_{A_0A_1B_0}=W_{A_0A_1B_0}.
\eea
Hence, we complete the proof.
\end{proof}

\begin{remark}
During the proof, we can observe that the one-shot minimum error of simulation via $A$-to-$B$ non-signalling bipartite superchannels is given by 
\bea
    e^{\rightarrow}(\Upsilon_m, \cN) = \min & \;\; \mu \\
    {\rm s.t.} &\; \tr_{A_1B_1} Y_{A_0A_1B_0B_1} \leq \mu \mathds{1}_{A_0B_0},\, Y_{A_0A_1B_0B_1}\geq 0,\\
    &\;\; Y_{A_0A_1B_0B_1} \geq Q_{A_0A_1B_0B_1} - J^{\cN},\, Q_{A_0A_1B_0B_1} \geq 0,\\
    &\;\; Q_{A_0B_0} = \mathds{1}_{A_0B_0},\, V_{A_0A_1B_0B_1}\geq Q_{A_0A_1B_0B_1},\\
    &\;\; W_{A_0A_1B_0B_1}\geq Q_{A_0A_1B_0B_1}, \, X_{A_0A_1B_0B_1}\geq Q_{A_0A_1B_0B_1},\\
    &\;\; X_{A_0B_0B_1} = \pi_{A_0} \ox X_{B_0B_1},\\
    &\;\; mQ_{A_0B_0B_1} = W_{A_0B_0B_1}, \; X_{A_0B_0B_1} = V_{A_0B_0B_1},
\eea
by removing the $\mathbf{B}\nrightarrow \mathbf{A}$ condition. Similarly, we can obtain the SDP for $e^{\leftarrow}(\Upsilon_m, \cN)$.
\end{remark}

Based on the minimum simulation error, we present an SDP hierarchy for calculating the one-shot $\epsilon$-error bidirectional classical communication cost in Corollary~\ref{cor:sdp_hichy}, and then formulate the one-shot \textit{exact} bidirectional classical communication cost as a single SDP in Theorem~\ref{thm:one_shot_exactcost}.

\begin{corollary}\label{cor:sdp_hichy}
For a given bipartite quantum channel $\cN_{A_0B_0\rightarrow A_1B_1}$, there is an SDP hierarchy for its one-shot $\epsilon$-error bidirectional classical communication cost via non-signalling bipartite superchannels.
\end{corollary}
\begin{proof}
According to Proposition~\ref{prop:min_err_SDP}, we know that the one-shot $\epsilon$-error bidirectional classical communication cost via non-signalling bipartite superchannels is given by the following optimization problem.
\bea\label{Eq:smooth_cost_SDP}
S_{\leftrightarrow, \epsilon}^{(1)}(\cN) = \log\min & \;\; \lceil m\rceil \\
{\rm s.t.} &\; \tr_{A_1B_1} Y_{A_0A_1B_0B_1} \leq \epsilon \mathds{1}_{A_0B_0},\, Y_{A_0A_1B_0B_1}\geq 0,\\
&\;\; Y_{A_0A_1B_0B_1} \geq Q_{A_0A_1B_0B_1} - J^{\cN},\, Q_{A_0A_1B_0B_1} \geq 0,\\
&\;\; Q_{A_0B_0} = \mathds{1}_{A_0B_0},\, V_{A_0A_1B_0B_1}\geq Q_{A_0A_1B_0B_1},\\
&\;\; W_{A_0A_1B_0B_1}\geq Q_{A_0A_1B_0B_1}, \, X_{A_0A_1B_0B_1}\geq Q_{A_0A_1B_0B_1},\\
&\;\; X_{A_0B_0B_1} = \pi_{A_0} \ox X_{B_0B_1},\, X_{A_0A_1B_0} = \pi_{B_0} \ox X_{A_0A_1},\\
&\;\; mQ_{A_0B_0B_1} = W_{A_0B_0B_1}, \; X_{A_0B_0B_1} = V_{A_0B_0B_1},\\
&\;\; mQ_{A_0A_1B_0} = V_{A_0A_1B_0}, \; X_{A_0A_1B_0}=W_{A_0A_1B_0}.
\eea
This is not an SDP as there is a bilinear term $mQ_{A_0B_0B_1}$ in the constraints where both $m$ and $Q_{A_0A_1B_0B_1}$ are variables. However, it implies a hierarchy of SDPs with fixed $k\in\NN$ as follows
\bea
\Delta(\cN, k) = \min & \;\; \delta \\
{\rm s.t.} &\; \tr_{A_1B_1} Y_{A_0A_1B_0B_1} \leq \epsilon \mathds{1}_{A_0B_0},\, Y_{A_0A_1B_0B_1}\geq 0,\\
&\;\; Y_{A_0A_1B_0B_1} \geq Q_{A_0A_1B_0B_1} - J^{\cN},\, Q_{A_0A_1B_0B_1} \geq 0,\\
&\;\; Q_{A_0B_0} = \mathds{1}_{A_0B_0},\, V_{A_0A_1B_0B_1}\geq Q_{A_0A_1B_0B_1},\\
&\;\; W_{A_0A_1B_0B_1}\geq Q_{A_0A_1B_0B_1}, \, X_{A_0A_1B_0B_1}\geq Q_{A_0A_1B_0B_1},\\
&\;\; X_{A_0B_0B_1} = \pi_{A_0} \ox X_{B_0B_1},\, X_{A_0A_1B_0} = \pi_{B_0} \ox X_{A_0A_1},\\
&\;\; X_{A_0B_0B_1} = V_{A_0B_0B_1},\; X_{A_0A_1B_0}=W_{A_0A_1B_0},\\
&\; -\delta \idop_{A_0B_0B_1}\leq kQ_{A_0B_0B_1} - W_{A_0B_0B_1}\leq \delta \idop_{A_0B_0B_1},\\
&\; -\delta \idop_{A_0B_0B_1}\leq kQ_{A_0A_1B_0} - V_{A_0A_1B_0}\leq \delta \idop_{A_0B_0B_1}.
\eea
It can be verified that the above optimization problem is a valid SDP. Further, we have $S_{\leftrightarrow, \epsilon}^{(1)}(\cN) = \log k_{\min}$ where $k_{\min}$ is the smallest $k\in\NN$ such that $\Delta(\cN, k) = 0$ or is minimized to a negligible tolerance as that is when all conditions in Eq.~\eqref{Eq:smooth_cost_SDP} are satisfied.
\end{proof}

Notably, we can also have SDP hierarchies for calculating the one-shot $\epsilon$-error one-way classical communication cost via $A$-to-$B$ or $B$-to-$A$ non-signalling bipartite superchannels. According to Theorem~\ref{thm:smooth_AtoB}, both of these SDP hierarchies converge in at most $ 2^{\lceil\mathfrak{D}_{\max}^{\rightarrow(\leftarrow),\epsilon}(\cN) + 1\rceil}$ steps.

\begin{theorem}[One-shot exact bidirectional classical communication cost]\label{thm:one_shot_exactcost}
For a given bipartite quantum channel $\cN_{A_0B_0\rightarrow A_1B_1}$, its one-shot exact bidirectional classical communication cost via non-signalling bipartite superchannels is given by the following SDP:
\be\label{SDP:cost_simp}
\begin{aligned}
    S_{\leftrightarrow, 0}^{(1)}(\cN) = \log\min & \;\; \lceil m \rceil\\
    {\rm s.t.} &\;\; V_{A_0A_1B_0B_1}\geq J^{\cN}_{A_0A_1B_0B_1}, \, W_{A_0A_1B_0B_1}\geq J^{\cN}_{A_0A_1B_0B_1},\\
    &\;\; X_{A_0A_1B_0B_1} \geq J^{\cN}_{A_0A_1B_0B_1},\\
    &\;\; X_{A_0B_0B_1} = \pi_{A_0} \ox X_{B_0B_1}, \, X_{A_0A_1B_0} = \pi_{B_0} \ox X_{A_0A_1},\\
    &\;\; m J^{\cN}_{A_0B_0B_1} = W_{A_0B_0B_1},~X_{A_0B_0B_1} = V_{A_0B_0B_1},\\
    &\;\; m J^{\cN}_{A_0A_1B_0} = V_{A_0A_1B_0},~X_{A_0A_1B_0} = W_{A_0A_1B_0}.
\end{aligned}
\ee
\end{theorem} 

\begin{proof}
By the definition of the one-shot exact bidirectional classical communication cost and Proposition~\ref{prop:min_err_SDP}, we have
\begin{subequations}
\begin{align}
    S_{\leftrightarrow, 0}^{(1)}(\cN) = \log \min & \;\; \lceil m \rceil\\
    {\rm s.t.} &\;\; Q_{A_0A_1B_0B_1} = J^{\cN}_{A_0A_1B_0B_1},\, Q_{A_0A_1B_0B_1} \geq 0,\\
    &\;\; Q_{A_0B_0} = \mathds{1}_{A_0B_0},\, V_{A_0A_1B_0B_1}\geq Q_{A_0A_1B_0B_1},\\
    &\;\; W_{A_0A_1B_0B_1}\geq Q_{A_0A_1B_0B_1}, \, X_{A_0A_1B_0B_1}\geq Q_{A_0A_1B_0B_1},\\
    &\;\; X_{A_0B_0B_1} = \pi_{A_0} \ox X_{B_0B_1},\, X_{A_0A_1B_0} = \pi_{B_0} \ox X_{A_0A_1},\\
    &\;\; mQ_{A_0B_0B_1} = W_{A_0B_0B_1}, \; X_{A_0B_0B_1} = V_{A_0B_0B_1},\\
    &\;\; mQ_{A_0A_1B_0} = V_{A_0A_1B_0}, \; X_{A_0A_1B_0}=W_{A_0A_1B_0}.
\end{align}
\end{subequations}
Since $Q_{A_0A_1B_0B_1} = J^{\cN}_{A_0A_1B_0B_1}$, we can replace the variable $Q_{A_0A_1B_0B_1}$ with the input channel $J^{\cN}_{A_0A_1B_0B_1}$, and the constraints $Q_{A_0A_1B_0B_1} \geq 0$ and $Q_{A_0B_0} = \mathds{1}_{A_0B_0}$ are guaranteed. This replacement of $Q_{A_0A_1B_0B_1}$ with $J^{\cN}_{A_0A_1B_0B_1}$ results in the SDP specified in the statement of this theorem.
\end{proof}

Different from previous studies on point-to-point channel simulation cost assisted by non-signalling correlations~\cite{Duan2016,Fang_2020}, the simulation cost we have derived takes into account not only the non-signalling correlations between the simulation codes across the timeline of the protocol, but also the non-signaling correlations between two parties communicating through the target bipartite quantum channel. This network communication scenario is more practical when we consider both the space separation and the time separation of a bipartite quantum process. Notably, if $A_1$ and $B_0$ are taken as trivial one-dimensional subsystems (see Figure~\ref{fig:OW_TW_channel}), the SDP~\eqref{SDP:cost_simp} will reduce to
\begin{equation}
\begin{aligned}
    \log \min & \;\; \lceil m \rceil\\
    {\rm s.t.} &\;\; J^{\cN}_{A_0B_1}\leq V_{A_0B_1},\, J^{\cN}_{A_0B_1}\leq W_{A_0B_1},\, J^{\cN}_{A_0B_1}\leq X_{A_0B_1},\\
    &\;\; X_{A_0B_1} = \pi_{A_0} \ox X_{B_1},\\
    &\;\; m J^{\cN}_{A_0B_1} = W_{A_0B_1},\, X_{A_0B_1} = V_{A_0B_1},\\
    &\;\; m J^{\cN}_{A_0} = V_{A_0},\, X_{A_0} = W_{A_0}.
\end{aligned}
\end{equation}
Notice that $V_{A_0B_1}$ and $W_{A_0B_1}$ are redundant variables and $J^{\cN}_{A_0} = \idop_{A_0}$, so the SDP can be further simplified to
\begin{equation}
    \log\min \left\{\lceil m \rceil \,\middle|\, J^{\cN}_{A_0B_1}\leq \pi_{A_0} \ox X_{B_1},\, \tr [X_{B_1}] \pi_{A_0} = m\idop_{A_0} \right\}.
\end{equation}
Redefining $X_{B_1} \coloneqq X_{B_1} / d_{A_0}$ results in
\begin{equation}
    \log\min \left\{\lceil \tr [X_{B_1}] \rceil \,\middle|\, J^{\cN}_{A_0B_1}\leq \idop_{A_0} \ox X_{B_1} \right\}.
\end{equation}
This is identical to $\log\left\lceil2^{-H_{\min}(A|B)_{J^\cN}}\right\rceil$, which has been shown by Duan and Winter~\cite{Duan2016} to characterize the one-shot exact classical simulation cost of $\cN_{A_0\to B_1}$ under non-signalling assistance with $H_{\min}(A|B)_{J^\cN}$ being the conditional min-entropy~\cite{renner2005security, datta2009min}. Thus, Theorem~\ref{thm:one_shot_exactcost} can be reduced to the case for a point-to-point channel from $A_0$ to $B_1$. From Theorem~\ref{thm:one_shot_exactcost}, we can also obtain that the max-relative entropy of bidirectional classical communication is a lower bound on the one-shot exact bidirectional classical communication cost.

\subsection{Asymptotic classical communication cost of bipartite channels}

While the one-shot exact bidirectional classical communication cost provides insights into the resources required for a single instance of channel simulation, it is also essential to investigate the asymptotic behavior of the cost when simulating a large number of identical and independently distributed (i.i.d.) target channels. The asymptotic bidirectional classical communication cost of $\cN_{A_0 B_0\rightarrow A_1 B_1}$ via non-signalling bipartite superchannel is given by the regularization
\be
    S_{\leftrightarrow}(\cN) \coloneqq \lim_{\epsilon\rightarrow 0}\lim_{n\rightarrow \infty}\frac{1}{n} S_{\leftrightarrow, \epsilon}^{(1)}(\cN^{\ox n}),
\ee
and the asymptotic exact bidirectional classical communication cost of $\cN_{A_0 B_0\rightarrow A_1B_1}$ is given by
\be
    S_{\leftrightarrow, 0}(\cN) \coloneqq \lim_{n\rightarrow \infty}\frac{1}{n} S_{\leftrightarrow, 0}^{(1)}(\cN^{\ox n}),
\ee
Similarly, we have the asymptotic (exact) one-way classical communication cost of $\cN_{A_0 B_0\rightarrow A_1 B_1}$ via $A$-to-$B$ and $B$-to-$A$ non-signalling bipartite superchannels
\begin{equation}
    S_{*}(\cN) \coloneqq \lim_{\epsilon\rightarrow 0}\lim_{n\rightarrow \infty}\frac{1}{n} S_{*, \epsilon}^{(1)}(\cN^{\ox n}),~S_{*, 0}(\cN) \coloneqq \lim_{n\rightarrow \infty}\frac{1}{n} S_{*, 0}^{(1)}(\cN^{\ox n}),
\end{equation}
where $*\in \{\rightarrow,\leftarrow\}$.

We first establish the asymptotic one-way classical communication cost via one-way non-signalling bipartite superchannels.

\begin{theorem}[Asymptotic one-way classical communication cost]\label{thm:vani_err_lowerAB}
For a given bipartite quantum channel $\cN_{\bA_0\bB_0\rightarrow \bA_1\bB_1}$, its asymptotic one-way classical communication cost via $A$-to-$B$ or $B$-to-$A$ non-signalling bipartite superchannels satisfies 
\begin{equation}
    S_{*}(\cN) = \lim_{\epsilon\rightarrow 0}\lim_{n \rightarrow \infty} \frac{1}{n} \mathfrak{D}_{\max}^{*,\epsilon}(\cN^{\ox n}),
\end{equation}
where $*\in\{\rightarrow,\leftarrow\}$.
\end{theorem}
\begin{proof}
This is a direct consequence of Theorem~\ref{thm:smooth_AtoB}.
\end{proof}

\begin{theorem}[Asymptotic exact one-way classical communication cost]
For a given bipartite quantum channel $\cN_{\bA_0\bB_0\rightarrow \bA_1\bB_1}$, its asymptotic exact one-way classical communication cost via $A$-to-$B$ or $B$-to-$A$ non-signalling bipartite superchannels satisfies 
\begin{equation}
    S_{*,0}(\cN) = \mathfrak{D}_{\max}^{*}(\cN),
\end{equation}
where $*\in\{\rightarrow,\leftarrow\}$.
\end{theorem}
\begin{proof}
This is a direct consequence of Theorem~\ref{thm:smooth_AtoB} and the additivity of the max-relative entropy of one-way classical communication given in Lemma~\ref{lem:dmax_oneway_add}.
\end{proof}

Now, we study the bidirectional classical communication cost. We introduce the \textit{channel's bipartite conditional min-entropy} and demonstrate its role as an additive lower bound on the one-shot exact bidirectional classical communication cost of $\cN_{A_0 B_0\rightarrow A_1 B_1}$ via non-signalling bipartite superchannels. It will favorably give a lower bound on the asymptotic exact cost.

\begin{definition}[Bipartite conditional min-entropy]
Let $\cN_{A_0 B_0\rightarrow A_1B_1}$ be a bipartite quantum channel. The bipartite conditional min-entropies of $\cN_{A_0 B_0\rightarrow A_1B_1}$ are defined as
\be
\begin{aligned}
    &H_{\min}(\mathbf{A}|\mathbf{B})_{\cN}\coloneqq -\min_{\cX} D_{\max}\left(\cR^{\pi}_{A_1}\circ\cN_{A_0 B_0\rightarrow A_1B_1} \,\middle\|\, \cR^{\pi}_{A_0\to A_1} \ox \cX_{B_0\rightarrow B_1}\right),\\
    &H_{\min}(\mathbf{B}|\mathbf{A})_{\cN}\coloneqq -\min_{\cY} D_{\max}\left(\cR^{\pi}_{B_1}\circ\cN_{A_0 B_0\rightarrow A_1B_1} \,\middle\|\, \cY_{A_0\rightarrow A_1} \ox \cR^{\pi}_{B_0\rightarrow B_1}\right),
\end{aligned}
\ee
where $\cX$ ranges over all quantum channels in $\CPTP(B_0,B_1)$ for $H_{\min}(\mathbf{A}|\mathbf{B})_{\cN}$, and $\cY$ ranges over all quantum channels in $\CPTP(A_0,A_1)$ for $H_{\min}(\mathbf{B}|\mathbf{A})_{\cN}$.
\end{definition}
By the property of the max-relative entropy of quantum channels, the bipartite conditional min-entropies of $\cN_{A_0 B_0\rightarrow A_1B_1}$ can be expressed as follows:
\be\label{SDP:chan_cond_minetrpy}
\begin{aligned}
    H_{\min}(\mathbf{A}|\mathbf{B})_{\cN} &= -\min_{\cX} D_{\max}\left(\idop_{A_1}\ox J^{\cN}_{A_0B_0B_1} \,\middle\|\, \idop_{A_0A_1} \ox J^{\cX}_{B_0B_1}\right)\\
    &= -\log\min \left\{m \,\middle|\, X_{B_0} = m\cdot\mathds{1}_{B_0},~J^{\cN}_{A_0B_0B_1} \leq \idop_{A_0} \ox X_{B_0B_1}\right\},\\
    H_{\min}(\mathbf{B}|\mathbf{A})_{\cN} &= -\min_{\cY} D_{\max}\left(\idop_{B_1}\ox J^{\cN}_{A_0A_1B_0} \,\middle\|\, \idop_{B_0 B_1} \ox J^{\cY}_{A_0A_1}\right)\\
    &= -\log\min \left\{m \,\middle|\, Y_{A_0} = m\cdot\mathds{1}_{A_0},~J^{\cN}_{A_0A_1B_0} \leq \idop_{B_0} \ox Y_{A_0A_1}\right\}.
\end{aligned}
\ee
Notably, if $A_1$ and $B_0$ are taken as one-dimensional trivial subsystems, the bipartite conditional min-entropy $H_{\min}(\mathbf{A}|\mathbf{B})_{\cN}$ will reduce to the common channel conditional min-entropy of a point-to-point channel~\cite{Duan2016}:
\be
    H_{\min}(\mathbf{A}|\mathbf{B})_{\cN} = -\min_{\cX} D_{\max}\left(\cN_{A_0\rightarrow B_1} \,\middle\|\, \cR^{\sigma}_{A_0\rightarrow B_1}\right) = -\min_{\sigma_{B_1}} D_{\max}\left(J^{\cN}_{A_0 B_1} \,\middle\|\, \idop_{A_0} \ox \sigma_{B_1}\right),
\ee
where $\cR^{\sigma}_{A_0\rightarrow B_1}$ is a replacement channel that always outputs $\sigma_{B_1}$. It is also related to the channel's max-information~\cite{Fang_2020} as $H_{\min}(\mathbf{A}|\mathbf{B})_{\cN} = -I_{\max}(A:B)_{\cN}$. In the following, we show a valuable property of the bipartite channel conditional min-entropy, additivity with respect to the tensor product of quantum channels.

\begin{lemma}[Additivity of bipartite conditional min-entropy]\label{lem:Hmin_add}
The channel's bipartite conditional min-entropy is additive with respect to the tensor product of quantum channels, i.e., for any $\cN_1\in\CPTP(A_0B_0:A_1B_1)$ and $\cN_2\in\CPTP(A'_0B'_0:A'_1B'_1)$,
\begin{equation}
H_{\min}(\mathbf{A}|\mathbf{B})_{\cN_1\ox \cN_2}= H_{\min}(\mathbf{A}|\mathbf{B})_{\cN_1} + H_{\min}(\mathbf{A}|\mathbf{B})_{\cN_2}.
\end{equation}
\end{lemma}
\begin{proof}
Suppose $\left\{m_1, X^{(1)}_{B_0B_1}\right\}$ and $\left\{m_2, X^{(2)}_{B'_0B'_1}\right\}$ form two feasible solutions of the SDP~\eqref{SDP:chan_cond_minetrpy} for $H_{\min}(\mathbf{A}|\mathbf{B})_{\cN_1}$ and $H_{\min}(\mathbf{A}|\mathbf{B})_{\cN_2}$, respectively. It is easy to verify that $\left\{m_1m_2, X^{(1)}_{B_0B_1} \ox X^{(2)}_{B'_0B'_1}\right\}$ forms a feasible solution of the SDP for $H_{\min}(\mathbf{A}|\mathbf{B})_{\cN_1\ox \cN_2}$. Hence, we have
\be\label{eq:hmin_superadd}
    H_{\min}(\mathbf{A}|\mathbf{B})_{\cN_1\ox \cN_2} \geq -\log m_1m_2 = -\log m_1 - \log m_2 = H_{\min}(\mathbf{A}|\mathbf{B})_{\cN_1} + H_{\min}(\mathbf{A}|\mathbf{B})_{\cN_2}.
\ee
Using Lagrangian method, the dual SDP of $H_{\min}(\mathbf{A}|\mathbf{B})_{\cN}$ is given by
\begin{equation}
    -\log\max\Big\{ \tr \big[M_{A_0B_0B_1} J^{\cN}_{A_0B_0B_1}\big]~\Big|~\tr N_{B_0} = 1,~M_{A_0B_0B_1} \geq 0,~N_{B_0}\ox \idop_{B_1} \geq M_{B_0B_1}\Big\}.
\end{equation}
It is easy to see that the strong duality holds by Slater's condition. For any feasible solution $\big\{M^{(1)}_{A_0B_0B_1}, N^{(1)}_{B_0}\big\}$ and $\big\{M^{(2)}_{A'_0B'_0B'_1}, N^{(2)}_{B'_0}\big\}$ for the dual SDPs of $H_{\min}(\mathbf{A}|\mathbf{B})_{\cN_1}$ and $H_{\min}(\mathbf{A}|\mathbf{B})_{\cN_2}$, respectively, we can verify that $\big\{M^{(1)}_{A_0B_0B_1}\ox M^{(2)}_{A'_0B'_0B'_1}, N^{(1)}_{B_0}\ox N^{(2)}_{B'_0}\big\}$ forms a feasible solution of the dual SDP for $H_{\min}(\mathbf{A}|\mathbf{B})_{\cN_1\ox \cN_2}$. Hence, we have $H_{\min}(\mathbf{A}|\mathbf{B})_{\cN_1\ox \cN_2} \leq H_{\min}(\mathbf{A}|\mathbf{B})_{\cN_1} +H_{\min}(\mathbf{A}|\mathbf{B})_{\cN_2}$. Together with Eq.~\eqref{eq:hmin_superadd}, we have $H_{\min}(\mathbf{A}|\mathbf{B})_{\cN_1\ox \cN_2}= H_{\min}(\mathbf{A}|\mathbf{B})_{\cN_1} + H_{\min}(\mathbf{A}|\mathbf{B})_{\cN_2}$. Similar results can be obtained for $H_{\min}(\mathbf{B}|\mathbf{A})_{\cN_1\ox \cN_2}$. Hence, we complete the proof.
\end{proof}

Based on the additivity of the channel's bipartite conditional min-entropy, we present a single-letter lower bound on the asymptotic exact bidirectional classical communication cost of an arbitrary bipartite quantum channel.

\begin{theorem}[Asymptotic exact bidirectional classical communication cost]\label{thm:cost_lowerbound}
For a given bipartite quantum channel $\cN_{A_0 B_0\rightarrow A_1B_1}$, its asymptotic exact bidirectional classical communication cost via non-signalling bipartite superchannels satisfies 
\begin{equation}
     S_{\leftrightarrow, 0}(\cN) \geq -\min\big\{H_{\min}(\mathbf{A}|\mathbf{B})_{\cN},~H_{\min}(\mathbf{B}|\mathbf{A})_{\cN}\big\}.
\end{equation}
\end{theorem}
\begin{proof}
Recall the definition of the max-relative entropy of bidirectional communication:
\be\label{SDP:dmaxNS1}
\begin{aligned}
    \mathfrak{D}_{\max}^{\leftrightarrow}(\cN) = \log\, \min &\;\; \lambda\\
     \textrm{s.t.} 
     &\;\; J_{A_0 A_1 B_0 B_1}^{\cN} \leq Y_{A_0 A_1B_0B_1},~Y_{A_0B_0} = \lambda \idop_{A_0B_0},\\
     &\;\; Y_{A_0B_0B_1} = \pi_{A_0}\ox Y_{B_0B_1},~Y_{A_0A_1B_0} = \pi_{B_0}\ox Y_{A_0A_1}.
\end{aligned}
\ee
Suppose $\{\lambda, Y_{A_0A_1B_0B_1}\}$ forms a feasible solution of SDP~\eqref{SDP:dmaxNS1}. Then let 
\begin{equation}
    \bar{X}_{B_0 B_1} \coloneqq Y_{B_0B_1}/d_{A_0},~\bar{Y}_{A_0 A_1} \coloneqq Y_{A_0A_1}/d_{B_0}.
\end{equation}
It follows that $\{\lambda, \bar{X}_{B_0 B_1}\}$ and $\{\lambda, \bar{Y}_{\bA_0 \bA_1}\}$ are feasible solutions of the optimization for $-H_{\min}(\mathbf{A}|\mathbf{B})_{\cN}$ and $-H_{\min}(\mathbf{B}|\mathbf{A})_{\cN}$ as given in Eq.~\eqref{SDP:chan_cond_minetrpy}, respectively. Thus, we have $-H_{\min}(\mathbf{A}|\mathbf{B})_{\cN}\leq \mathfrak{D}_{\max}^{\leftrightarrow}(\cN)$ and $-H_{\min}(\mathbf{B}|\mathbf{A})_{\cN}\leq \mathfrak{D}_{\max}^{\leftrightarrow}(\cN)$ which yields $\mathfrak{D}_{\max}^{\leftrightarrow}(\cN)\geq \max\big\{-H_{\min}(\mathbf{A}|\mathbf{B})_{\cN},~-H_{\min}(\mathbf{B}|\mathbf{A})_{\cN}\big\}$. Combining with Theorem~\ref{thm:smooth_low_bound}, we have
\begin{equation}\label{Eq:S1_geq_Hmin}
   S_{\leftrightarrow, 0}^{(1)}(\cN) \geq -\min\big\{H_{\min}(\mathbf{A}|\mathbf{B})_{\cN},~H_{\min}(\mathbf{B}|\mathbf{A})_{\cN}\big\}.
\end{equation}
Then, consider that
\begin{equation*}
S_{\leftrightarrow, 0}(\cN) = \lim_{n\rightarrow \infty} \frac{1}{n} S_{\leftrightarrow, 0}^{(1)}(\cN^{\ox n}) \geqt{(i)} - \lim_{n\rightarrow \infty} \frac{1}{n} H_{\min}(\mathbf{A}|\mathbf{B})_{\cN^{\ox n}} \eqt{(ii)} -\lim_{n\rightarrow \infty} \frac{n}{n} H_{\min}(\mathbf{A}|\mathbf{B})_{\cN}
= -H_{\min}(\mathbf{A}|\mathbf{B})_{\cN},
\end{equation*}
where we used Eq.~\eqref{Eq:S1_geq_Hmin} in (i), and Lemma~\ref{lem:Hmin_add} in (ii). Applying the same inequality for $H_{\min}(\mathbf{B}|\mathbf{A})_{\cN}$, we complete the proof.
\end{proof}

Using the relationship between the $\epsilon$-smooth max-relative entropy of resources and the relative entropy of resources in dynamical resource theory, we have the following result on the asymptotic vanishing error bidirectional classical communication cost.
\begin{theorem}[Asymptotic bidirectional classical communication cost]\label{thm:vani_err_lowerbd}
For a given bipartite quantum channel $\cN_{\bA_0\bB_0\rightarrow \bA_1\bB_1}$, its asymptotic bidirectional classical communication cost via non-signalling bipartite superchannels satisfies 
\begin{equation}
    S_{\leftrightarrow}(\cN) \geq \lim_{n \rightarrow \infty} \frac{1}{n} \mathfrak{D}^{\leftrightarrow}(\cN^{\ox n}).
\end{equation}
\end{theorem}
\begin{proof}
It is proved in Ref.~\cite[Lemma 13]{Gour_2019a} that for any convex quantum resource theory $\mathfrak{F}$, 
\begin{equation}
    \lim _{\epsilon \rightarrow 0^{+}} \lim_{n \rightarrow \infty} \frac{1}{n} \mathfrak{D}_{\max,
    \mathfrak{F}}^{\epsilon}(\cN^{\ox n}) \geq \lim_{n \rightarrow \infty} \frac{1}{n} \mathfrak{D}_{\mathfrak{F}}(\cN^{\ox n}).
\end{equation}
Hence, we directly have
\begin{equation}
    S_{\leftrightarrow}(\cN) \geq \lim _{\epsilon \rightarrow 0^{+}} \lim _{n \rightarrow \infty} \frac{1}{n} \mathfrak{D}_{\max}^{\leftrightarrow, \epsilon}(\cN^{\ox n}) \geq \lim_{n \rightarrow \infty} \frac{1}{n} \mathfrak{D}^{\leftrightarrow}(\cN^{\ox n}).
\end{equation}
\end{proof}

\section{Numerical optimization on physical simulation protocols}\label{sec:lose}
In this section, we investigate the practical implementation of bipartite quantum channel simulation with bidirectional classical communication. Consider two parties, Alice and Bob, located in separate laboratories with access to shared entanglement and local quantum operations in their respective labs. Suppose Alice and Bob can use entanglement in their own labs through encoding and decoding.

Let $E_A$ and $E_B$ denote their ancillary systems for memories between encoding and decoding operations, and $S_A$ and $S_B$ represent their ancillary systems for shared entanglement. The encoding operations performed by Alice and Bob are described by channels $\cE^{A}_{S_A A_0\rightarrow E_A\bA_0}$ and $\cE^{B}_{S_B B_0\rightarrow E_B\bB_0}$, respectively. Similarly, their decoding operations are given by channels $\cD^{A}_{E_A\tS_A \bA_1\rightarrow A_1}$ and $\cD^{B}_{E_B\tS_B \bB_1\rightarrow B_1}$, respectively. Figure~\ref{fig:LOSE} illustrates an overview of the simulation protocol. This framework enables a natural protocol for bipartite channel simulation using local operations and shared entanglement (LOSE).

\begin{figure}[t]
    \centering
    \includegraphics[width=0.85\linewidth]{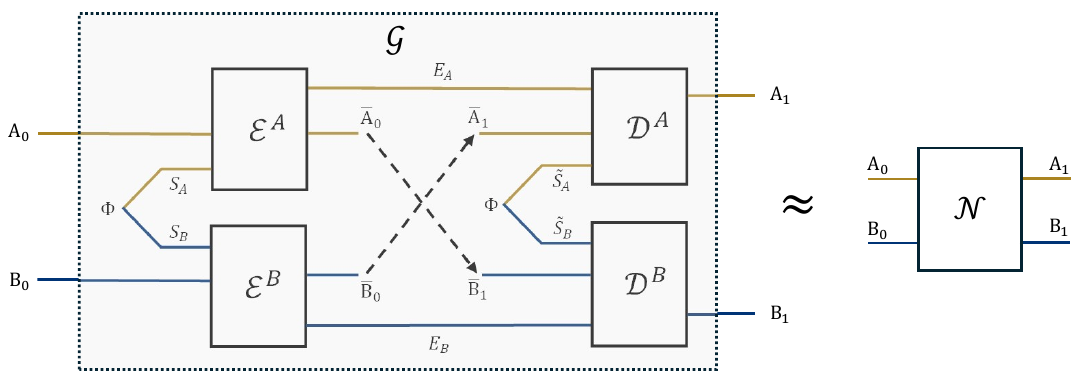}
    \caption{Bipartite quantum channel simulation via local operations and shared entanglement.}
    \label{fig:LOSE}
\end{figure}

For a target channel $\cN_{A_0B_0\rightarrow A_1B_1}$, the one-shot exact bidirectional classical communication cost of $\cN_{A_0B_0\rightarrow A_1B_1}$ via LOSE is defined as the minimum $m$ such that,
\begin{equation}
\begin{aligned}
    \cN_{A_0B_0\rightarrow A_1B_1}(\psi_{RA_0B_0}) = (\cD^{A}_{E_A\tS_A \bA_1\rightarrow A_1}\ox \cD^{B}_{E_B\tS_B \bB_1\rightarrow B_1}) \circ \Upsilon_m \circ & (\cE^{A}_{S_A A_0\rightarrow E_A\bA_0}\ox \cE^{B}_{S_B B_0\rightarrow E_B\bB_0})\\
    & (\Phi_{S_AS_B}\ox \Phi'_{\tS_A\tS_B}\ox \psi_{RA_0B_0})
\end{aligned}
\end{equation}
for any input state $\psi_{RA_0B_0}$ with a reference system $R$. Here $\Upsilon_m: \bA_0\bB_0\rightarrow \bA_1\bB_1$ denotes the classical swap channel between systems $\bA_0\bB_0$ and $\bA_1\bB_1$, and the optimization for $m$ is over all possible encoders, decoders and shared entangled states $\Phi_{S_AS_B},\Phi'_{\tS_A\tS_B}$ between Alice and Bob. The one-shot $\epsilon$-error classical communication cost via LOSE follows analogously. We denote the effective channel resulting from this simulation protocol by
\begin{equation}
\begin{aligned}
    &\cG_{A_0B_0 \rightarrow A_1B_1}(\cdot)\\
    =& (\cD^{A}_{E_A\tS_A \bA_1\rightarrow A_1}\ox \cD^{B}_{E_B\tS_B \bB_1\rightarrow B_1}) \circ \Upsilon_m \circ (\cE^{A}_{S_A A_0\rightarrow E_A\bA_0}\ox \cE^{B}_{S_B B_0\rightarrow E_B\bB_0})\big(\Phi_{S_AS_B}\ox \Phi'_{\tS_A\tS_B} \ox (\cdot)\big).
\end{aligned}
\end{equation}
Then, for a fixed $m\in\NN$, the minimum simulation error via LOSE is defined as 
\begin{equation}\label{Eq:LOSE_err}
e_{\mathrm{LOSE}}(\cN_{A_0B_0\rightarrow A_1B_1}) = \min\left\{\frac{1}{2}\big\|\cN - \cG\big\|_{\diamond}~\big|~\cE^{A},\cE^{B},\cD^{A},\cD^{B} \text{are quantum channels}\right\}, 
\end{equation}
where the optimization also ranges over all possible dimensions of the ancillary systems for the shared entanglement. We note that compared to the minimum simulation error in Proposition~\ref{prop:min_err_SDP}, which allows unlimited shared entanglement and non-signalling correlations, the error metric in Eq.~\eqref{Eq:LOSE_err} incorporates more physically realizable operations. The \Choi operator of $\cG_{A_0B_0\rightarrow A_1B_1}$ can be expressed as the link product of the \Choi operators of $\cE^A$, $\cE^B$, $\cD^A$, and $\cD^B$, which we denote as 
\begin{equation}
    J^{\cG}_{A_0 A_1B_0B_1} = \mathrm{LinProd}(J^{\cE^A}, J^{\cE^B}, J^{\cD^A}, J^{\cD^B}).
\end{equation}
The minimum simulation error via LOSE is then given by 
\be\label{sdp:see-saw}
\begin{aligned}
    e_{\mathrm{LOSE}}(\cN_{A_0B_0\rightarrow A_1B_1}) = &\min_{\omega,Y,J^{\cE^A}, J^{\cE^B}, J^{\cD^A}, J^{\cD^B}} \; \omega\\
     \textrm{s.t.}
     &\;\; Y_{A_0 A_1B_0B_1} \geq 0,~Y_{A_0B_0} \leq \omega \idop_{A_0B_0},\\
     &\;\; J^{\cM} \geq 0,~J^{\cM}_{\rm in} = \idop_{\rm in},~\forall \cM\in\{\cE^A, \cE^B, \cD^A, \cD^B\},\\
     &\;\; Y_{A_0 A_1B_0B_1} \geq \mathrm{LinProd}(J^{\cE^A}, J^{\cE^B}, J^{\cD^A}, J^{\cD^B}) - J_{A_0 A_1 B_0 B_1}^{\cN}.
\end{aligned}
\ee

While this optimization problem is generally intractable due to the optimization over unbounded dimension and the bilinearity on variables, we develop a seesaw-based algorithm~\cite{Konno1976ACP} (see~\cite{Werner2001} for an example application in quantum information theory) to estimate Eq.~\eqref{Eq:LOSE_err} for fixed $m$ and fixed dimensions of all ancillary systems, thereby enabling quantification of the performance gap between LOSE and non-signalling protocols. When any three of these operators are fixed, the optimization over the remaining Choi operator becomes a valid SDP. 
In Algorithm~\ref{alg:lose}, we implement this optimization by iteratively inputting three distinct \Choi operators from the set $\{J^{\cE^A}, J^{\cE^B}, J^{\cD^A}, J^{\cD^B}\}$ and optimizing over the remaining operator at each step. We note that other methods using SDP hierarchies to deal with the bilinearity that appeared from channel coding theory may also be applied~\cite{Berta_2021}.

\begin{algorithm}\label{alg:lose}
    \renewcommand{\algorithmcfname}{Algorithm}
    \SetKwInOut{Input}{Input}\SetKwInOut{Output}{Output}
    \caption{Seesaw algorithm for estimating $e_{\mathrm{LOSE}}(\cN_{A_0B_0\rightarrow A_1B_1})$.}
    \Input{A target bipartite quantum channel $\cN_{A_0B_0\rightarrow A_1B_1}$; Encoder $\cE^{B}_{S_B B_0\rightarrow E_B\bB_0}$, decoder $\cD^{A}_{E_A\tS_A \bA_1\rightarrow A_1},\cD^{B}_{E_B\tS_B \bB_1\rightarrow B_1}$; number of iterations $K$;
    }
    \Output{An upper bound on $e_{\mathrm{LOSE}}(\cN_{A_0B_0\rightarrow A_1B_1})$;}
    $\hat{J}^{\cE^{B}} \leftarrow J^{\cE^{B}}, \hat{J}^{\cD^{A}} \leftarrow J^{\cD^{A}}$ and $\hat{J}^{\cD^{B}} \leftarrow J^{\cD^{B}}$\;
    \For{$k$ from $1$ to $K$}{
        Input $\hat{J}^{\cE^{B}},\hat{J}^{\cD^{A}},\hat{J}^{\cD^{B}}$ to the SDP~\eqref{sdp:see-saw} and output the solution $\hat{J}^{\cE^{A}}$\;
        Input $\hat{J}^{\cE^{A}},\hat{J}^{\cD^{A}},\hat{J}^{\cD^{B}}$ to the SDP~\eqref{sdp:see-saw} and output the solution $\hat{J}^{\cE^{B}}$\;
        Input $\hat{J}^{\cE^{A}},\hat{J}^{\cE^{B}},\hat{J}^{\cD^{B}}$ to the SDP~\eqref{sdp:see-saw} and output the solution $\hat{J}^{\cD^{A}}$\;
        Input $\hat{J}^{\cE^{A}},\hat{J}^{\cE^{B}},\hat{J}^{\cD^{A}}$ to the SDP~\eqref{sdp:see-saw} and output the solution $\hat{J}^{\cD^{B}}$ with the optimal value $\omega^{(k)}$ of~\eqref{sdp:see-saw}\;
    }
    Return $\omega^{(K)}$\;
\end{algorithm}

\paragraph{Numerical results.} We consider a target bipartite channel $\cN_{p}(\cdot)$ defined as
\begin{equation}
    \cN_p(\rho) = (1-p) \cdot \mathrm{CNOT_{\leftrightarrow}} \rho \mathrm{CNOT_{\leftrightarrow}}^{\dagger} + p\cdot \frac{\idop}{4}
\end{equation}
where $\mathrm{CNOT_{\leftrightarrow}}=\mathrm{CNOT_{BA}}\cdot \mathrm{CNOT_{AB}}$ represents the sequential application of two CNOT gates: first a CNOT gate with Alice as the control, followed by a CNOT gate with Bob as the control. We set $K=300, d_{E_A} = d_{E_B} = d_{S_A} = d_{S_B} = d_{\tS_A} = d_{\tS_B} = 2$. We choose both $\Phi_{S_AS_B}$ and $\Phi'_{\tS_A\tS_B}$ as the Bell state, and generate random channels for $\cE^{B}_{S_B B_0\rightarrow E_B\bB_0},\cD^{A}_{E_A\tS_A \bA_1\rightarrow A_1},\cD^{B}_{E_B\tS_B \bB_1\rightarrow B_1}$ using the \texttt{RandomSuperoperator} function in QETLAB~\cite{qetlab} to run Algorithm~\ref{alg:lose} for $\cN_p$. The numerical calculations are implemented in MATLAB~\cite{MATLAB} with CVX~\cite{cvx} and QETLAB~\cite{qetlab}.

\begin{figure}[t]
    \centering
    \includegraphics[width=0.55\linewidth]{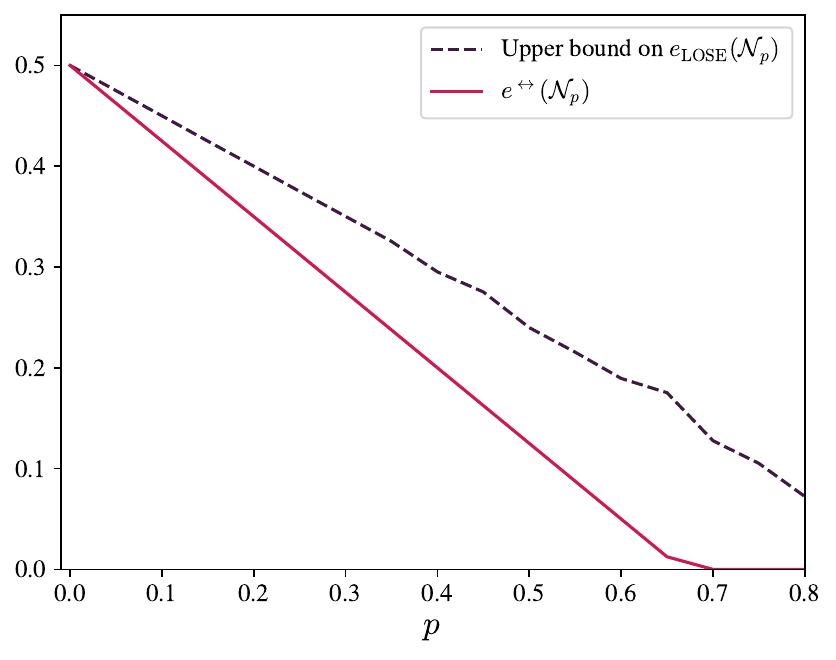}
    \caption{Comparison of optimized simulation errors for the quantum channel $\cN_p$ under two protocols: non-signalling bipartite superchannel and local operations with shared entanglement. The $x$-axis corresponds to the channel parameter $p$ and the $y$-axis corresponds to the simulation error. The red solid line represents the optimal simulation error achieved by the non-signalling bipartite superchannel. The dashed line shows an upper bound obtained via LOSE protocols using Algorithm~\ref{alg:lose}.}
    \label{fig:noisy_CNOTAB}
\end{figure}

Figure~\ref{fig:noisy_CNOTAB} illustrates the simulation error given by an optimized LOSE protocol using Algorithm~\ref{alg:lose} (the dashed line), and the simulation error given by an optimized non-signalling bipartite superchannel using the SDP given in Proposition~\ref{prop:min_err_SDP}. We can observe that when $p=0$ (no depolarizing noise), both optimized protocols yield identical simulation errors, suggesting that our non-signalling bipartite superchannel framework provides a tight converse bound for physically implementable protocols such as LOSE. However, the gap between the optimized simulation errors widens as $p$ increases. This divergence may be attributed to various factors: the need for increased entanglement resources, larger dimensions of ancillary systems $E_A$ and $E_B$, as the non-siganlling protocol encompasses unlimited such ancillary resources. While our seesaw-based optimization algorithm lacks formal convergence guarantees, it is noteworthy that our algorithm offers valuable insights into the performance gap between LOSE and non-signalling assisted protocols. Our results on non-signalling assisted bipartite channel simulation could also guide the development of more physically realizable protocols like LOSE.

\section{Examples}\label{sec:examples}

In this section, we investigate some representative examples of bipartite quantum channels and estimate their bidirectional classical communication cost assisted by non-signalling correlations.

\paragraph{Noisy $\mathrm{SWAP}^{\alpha}$ gate.} The SWAP gate is defined as $S\ket{\phi}\ket{\psi} = \ket{\psi}\ket{\phi}$. In a two-qubit case, the Heisenberg exchange between two electron spin qubits results in a $\mathrm{SWAP}^{\alpha}$ gate~\cite{Fan_2005} whose matrix form is given by 
\begin{equation}
S^{\alpha} = \begin{pmatrix}
1 & 0 & 0 & 0\\
0 & \frac{1+e^{i\pi\alpha}}{2} & \frac{1-e^{i\pi\alpha}}{2} & 0\\
0 & \frac{1-e^{i\pi\alpha}}{2} & \frac{1+e^{i\pi\alpha}}{2} & 0\\
0 & 0 & 0 & 1
\end{pmatrix}.
\end{equation}
Consider a $\mathrm{SWAP}^{\alpha}$ gate affected by a global depolarizing noise, resulting in a noisy bipartite channel $\cN_{\alpha,p}(\cdot)$ as
\be
\cN_{\alpha,p}(\rho) = (1-p)\cdot\mathrm{SWAP}^{\alpha} \rho (\mathrm{SWAP}^{\alpha})^{\dagger} + p \cdot\frac{\idop}{4}.
\ee
We numerically computed the one-shot exact bidirectional communication cost, as derived in Theorem~\ref{thm:one_shot_exactcost}, and the lower bound on the asymptotic exact bidirectional classical communication cost, as established in Theorem~\ref{thm:cost_lowerbound}, for the channel $\cN_{\alpha,p}(\cdot)$ with varying values of $p$ and $\alpha$. The results are illustrated in Figure~\ref{fig:noisy_swapalpha}. The red solid line represents the lower bound $-\min\{H_{\min}(\mathbf{A}|\mathbf{B}){\cN}, H_{\min}(\mathbf{B}|\mathbf{A})_{\cN}\}$, while the gray dotted line shows $S_{\leftrightarrow, 0}^{(1)}(\cdot)$. 

It can be seen that the classical communication cost increases with increasing $\alpha$ or decreasing noise level $p$. Notably, when $\alpha=1$ and $p=0$, the channel $\cN_{\alpha,p}(\cdot)$ corresponds to a SWAP gate, which requires 2 bits of classical information to implement with the assistance of non-signalling correlations, as shown in the figure. This scenario can be realized through a standard bidirectional teleportation protocol, where Alice and Bob each send 2 messages to the other to exchange their quantum states perfectly. This result further indicates that non-signalling correlations offer no advantage over local operations assisted by shared entanglement for such a SWAP operation. We also observe that our lower bound on the asymptotic bidirectional classical communication cost becomes increasingly tight as the depolarizing noise increases. Specifically, when $p=0.2$, the one-shot cost closely approaches the lower bound, and at $p=0.4$, they almost coincide in our plot. As the one-shot exact cost is an upper bound on the asymptotic cost, this convergence suggests that the channel's bipartite conditional min-entropy provides a reliable estimate of the asymptotic classical communication cost, particularly in higher noise regimes. Furthermore, when $p=0.4$, the target noisy channel can be simulated with less than $\log 3$ bits of classical information, which is better than trivially applying a bidirectional teleportation.

\begin{figure}[t]
    \centering
    \includegraphics[width=1\linewidth]{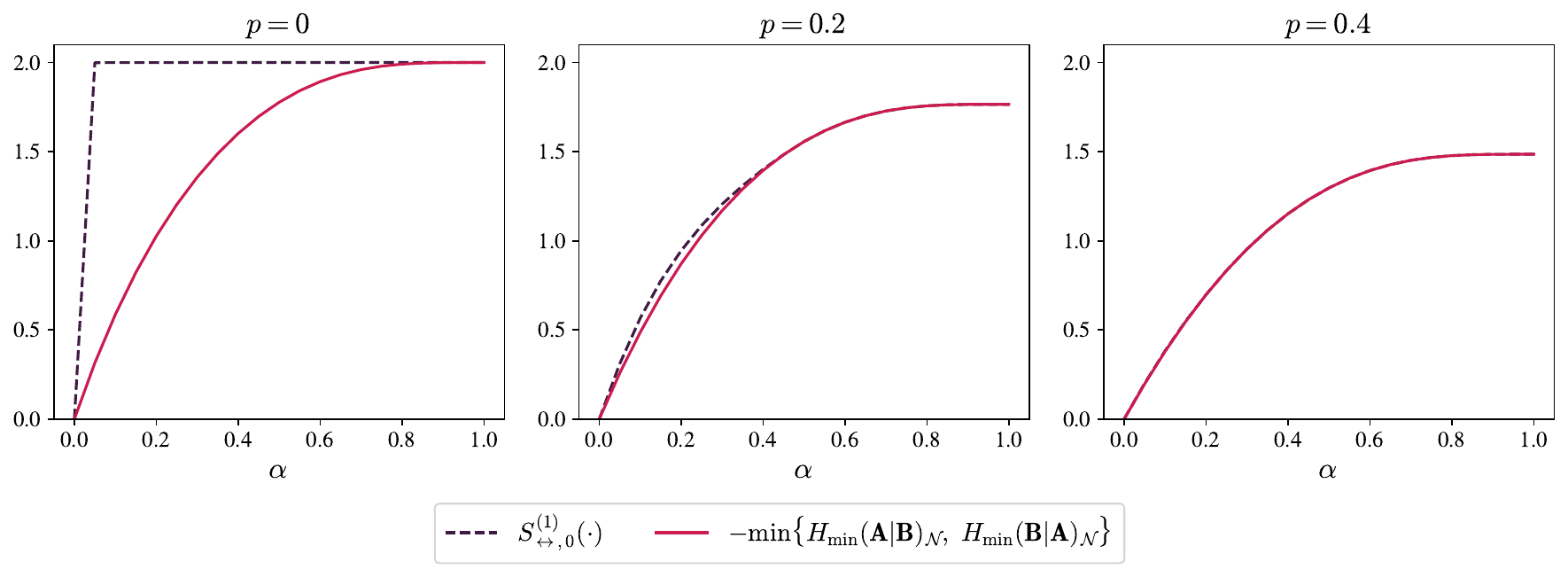}
    \caption{Asymptotic exact bidirectional classical communication cost estimates as a function of the $\mathrm{SWAP}^{\alpha}$ gate parameter $\alpha$ under varying depolarizing noise $p$. The graphs illustrate how increasing $p$ affects the communication cost, with dashed lines representing the one-shot exact cost and the red solid lines representing our lower bounds.}
    \label{fig:noisy_swapalpha}
\end{figure}
\begin{figure}[t]
    \centering
    \includegraphics[width=1\linewidth]{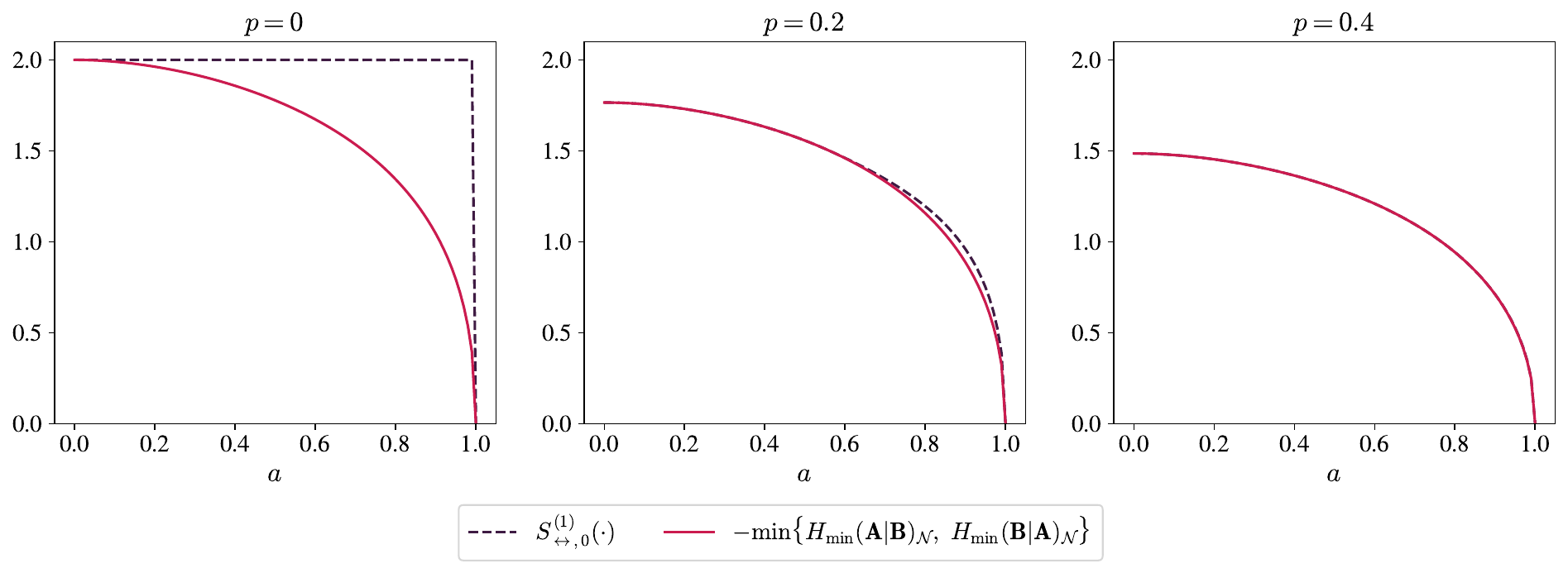}
    \caption{Asymptotic exact bidirectional classical communication cost estimates as a function of the partial swap operation parameter $a$ under varying depolarizing noise $p$. The graphs illustrate how increasing $p$ affects the communication cost, with dashed lines representing the one-shot exact cost and the red solid lines representing our lower bounds.}
    \label{fig:noisy_partialswap}
\end{figure}

\paragraph{Noisy partial swap operation.} Another bipartite operation we consider is the qubit partial swap operation~\cite{Audenaert_2016}, defined as $U_a=\sqrt{a}\idop + i \sqrt{1-a}S$. It is known to resemble a beamsplitter~\cite{K_nig_2013}. We investigate the bidirectional classical communication cost for this operation under various noise levels. Similar to the $\mathrm{SWAP}^{\alpha}$ gate, the noisy partial swap channel is modeled as:
\be
\cN_{a,p}(\rho) = (1-p)\cdot U_a \rho U_a^{\dagger} + p \cdot\frac{\idop}{4},
\ee
where $p$ represents the depolarizing noise parameter. We computed the one-shot exact bidirectional classical communication cost and the lower bound for the asymptotic case.

Figure~\ref{fig:noisy_partialswap} illustrates these results for different values of $a$ and noise levels $p$. As $a$ approaches 1, corresponding to the identity operation, the communication cost decreases. Increasing noise levels generally reduces the required communication cost. The lower bound on the asymptotic exact cost approaches the one-shot exact cost as the noise level increases, similar to our observations for the $\mathrm{SWAP}^{\alpha}$ gate. These results further support our earlier conclusions about the relationship between noise levels and the tightness of our lower bound.

\section{Discussion}\label{sec:conclu}

In this work, we have investigated how much bidirectional classical communication resources are required to simulate a desired bipartite quantum channel when non-signalling correlations are available. This channel simulation problem is central in quantum Shannon theory, and non-signalling correlations encompass all possible resources with no communication capabilities. To this end, we have introduced non-signalling bipartite channels and superchannels, which serve as the foundation for defining resource measures for classical communication. We have shown that the one-shot exact $\epsilon$-error bidirectional classical communication cost is lower bounded by the $\epsilon$-error max-relative entropy of classical communication and have derived an SDP formulation for the one-shot exact cost. Furthermore, we have introduced a channel's bipartite conditional min-entropy as an efficiently computable lower bound on the asymptotic exact bidirectional classical communication cost, which recovers and generalizes the results for point-to-point channels. To bridge the gap between non-signalling assisted protocols and physically realizable implementations, we have developed a seesaw-based optimization algorithm for estimating the minimum simulation error achievable via local operations and shared entanglement.

Numerical experiments showcase the effectiveness of our bound on various bipartite quantum channels, particularly for a class of noisy bipartite unitary operations. Our results shed light on the fundamental interplay between non-locality and classical communication in a quantum setting. The resource-theoretic approach allows for a systematic characterization of the communication capabilities of bipartite quantum channels under non-signalling correlations. Several intriguing questions remain unsolved:
\begin{enumerate}
    \item The non-signalling bipartite superchannel given in Definition~\ref{def:NSBSC} is physically motivated by non-signalling correlations. Another natural consideration is a superchannel realized by non-signalling pre- and post-processing channels, which we may call a non-signalling realizable bipartite superchannel. Furthermore, a CNSP superchannel is motivated by dynamical resource theory. It is clear that a non-signalling realizable bipartite superchannel is a non-signalling bipartite superchannel, and thus a CNSP superchannel. However, it remains an open question whether the converse statement also holds.
    \item Investigating the additivity of different resource measures of classical communication is an interesting problem. For instance, the additivity of the max-relative entropy bidirectional classical communication remains unclear. If it holds, the max-relative entropy would provide a lower bound on the asymptotic exact bidirectional classical communication cost. Furthermore, the additivity of the relative entropy of classical communication is still unknown. If it is additive, we can obtain a single-letter lower bound on the vanishing error bidirectional classical communication cost according to Theorem~\ref{thm:vani_err_lowerbd}. It is also very interesting to explore non-asymptotic analysis~\cite{Li2014a,Tomamichel2015,Cheng2017b,Datta2016,Wang2017a,Cheng2023a} for the communication and simulation of bipartite quantum channels. 
    \item In addition to the problem of simulating target channels using classical swap channels, it is also of interest to investigate the ability of a bipartite channel to exchange information. Similar methods may be applied to examine the maximum amount of classical information that can be reliably transmitted in both directions through a bipartite quantum channel, with the help of non-signalling correlations. Remarkably, Shannon's noisy channel coding theorem~\cite{Shannon1948} and its dual, the reverse Shannon theorem~\cite{Bennett2002}, characterize the capacity of classical channels and the resources required to simulate them using noiseless channels, respectively. In the quantum realm, the quantum reverse Shannon theorem~\cite{Bennett2014,Berta_2011} solves the case where communicating parties share unlimited entanglement, establishing that the optimal asymptotic simulation cost equals its entanglement-assisted classical capacity. This result directly implies a reverse Shannon theorem for the case of channels assisted by non-signaling correlations. When considering a more general and practical bipartite scenario, it becomes crucial and fascinating to investigate whether a quantum reverse Shannon theorem holds for bipartite channels assisted by non-local resources, such as entanglement and non-signaling correlations.
\end{enumerate}

\section*{Acknowledgement}
This work was supported by the National Key R\&D Program of China (Grant No. 2024YFE0102500), the National Natural Science Foundation of China (Grant. No.~12447107), the Guangdong Provincial Quantum Science Strategic Initiative (Grant No.~GDZX2403008, GDZX2403001), the Guangdong Natural Science Foundation (Grant No.~2025A1515012834), the Guangdong Provincial Key Lab of Integrated Communication, Sensing and Computation for Ubiquitous Internet of Things (Grant No. 2023B1212010007), the Quantum Science Center of Guangdong-Hong Kong-Macao Greater Bay Area, and the Education Bureau of Guangzhou Municipality.

\bibliographystyle{alpha}
\bibliography{main}

\newcommand{\etalchar}[1]{$^{#1}$}
\begin{thebibliography}{MLDS{\etalchar{+}}13}

\bibitem[ADO16]{Audenaert_2016}
Koenraad Audenaert, Nilanjana Datta, and Maris Ozols.
\newblock Entropy power inequalities for qudits.
\newblock {\em Journal of Mathematical Physics}, 57(5), May 2016.

\bibitem[BB84]{Bennett1984}
Charles~H Bennett and Gilles Brassard.
\newblock {Quantum cryptography: Public key distribution and coin tossing}.
\newblock In {\em International Conference on Computers, Systems {\&} Signal Processing, Bangalore, India, Dec 9-12, 1984}, pages 175--179, 1984.

\bibitem[BBC{\etalchar{+}}93]{Bennett1993}
Charles~H Bennett, Gilles Brassard, Claude Cr{\'{e}}peau, Richard Jozsa, Asher Peres, and William~K Wootters.
\newblock {Teleporting an unknown quantum state via dual classical and Einstein-Podolsky-Rosen channels}.
\newblock {\em Physical Review Letters}, 70(13):1895--1899, Mar 1993.

\bibitem[BBFS21]{Berta_2021}
Mario Berta, Francesco Borderi, Omar Fawzi, and Volkher~B. Scholz.
\newblock Semidefinite programming hierarchies for constrained bilinear optimization.
\newblock {\em Mathematical Programming}, 194(1–2):781–829, April 2021.

\bibitem[BCR11]{Berta_2011}
Mario Berta, Matthias Christandl, and Renato Renner.
\newblock The quantum reverse shannon theorem based on one-shot information theory.
\newblock {\em Communications in Mathematical Physics}, 306(3):579–615, Aug 2011.

\bibitem[BDH{\etalchar{+}}14]{Bennett2014}
Charles~H. Bennett, Igor Devetak, Aram~W. Harrow, Peter~W. Shor, and Andreas Winter.
\newblock The quantum reverse shannon theorem and resource tradeoffs for simulating quantum channels.
\newblock {\em IEEE Transactions on Information Theory}, 60(5):2926--2959, 2014.

\bibitem[BDW18]{Stefan2018}
Stefan B\"auml, Siddhartha Das, and Mark~M. Wilde.
\newblock Fundamental limits on the capacities of bipartite quantum interactions.
\newblock {\em Physical Review Letters}, 121:250504, Dec 2018.

\bibitem[BDWW19]{Stefan2019}
Stefan B\"auml, Siddhartha Das, Xin Wang, and Mark~M. Wilde.
\newblock Resource theory of entanglement for bipartite quantum channels, 2019.

\bibitem[Bel66]{Bell_1966}
John~S Bell.
\newblock On the problem of hidden variables in quantum mechanics.
\newblock {\em Reviews of Modern physics}, 38(3):447, 1966.

\bibitem[BF18]{Barman_2018}
Siddharth Barman and Omar Fawzi.
\newblock Algorithmic aspects of optimal channel coding.
\newblock {\em IEEE Transactions on Information Theory}, 64(2):1038–1045, Feb 2018.

\bibitem[BG17]{Buscemi_2017}
Francesco Buscemi and Gilad Gour.
\newblock Quantum relative lorenz curves.
\newblock {\em Physical Review A}, 95(1), Jan 2017.

\bibitem[BGNP01]{Beckman2001}
David Beckman, Daniel Gottesman, M.~A. Nielsen, and John Preskill.
\newblock Causal and localizable quantum operations.
\newblock {\em Physical Review A}, 64:052309, Oct 2001.

\bibitem[BHLS03]{Bennett_2003}
C.H. Bennett, A.W. Harrow, D.W. Leung, and J.A. Smolin.
\newblock On the capacities of bipartite hamiltonians and unitary gates.
\newblock {\em IEEE Transactions on Information Theory}, 49(8):1895--1911, 2003.

\bibitem[BP05]{Barrett2005}
Jonathan Barrett and Stefano Pironio.
\newblock Popescu-rohrlich correlations as a unit of nonlocality.
\newblock {\em Physical Review Letters}, 95:140401, Sep 2005.

\bibitem[BSST99]{Bennett1999}
Charles~H. Bennett, Peter~W. Shor, John~A. Smolin, and Ashish~V. Thapliyal.
\newblock Entanglement-assisted classical capacity of noisy quantum channels.
\newblock {\em Physical Review Letters}, 83:3081--3084, Oct 1999.

\bibitem[BSST02]{Bennett2002}
C.H. Bennett, P.W. Shor, J.A. Smolin, and A.V. Thapliyal.
\newblock Entanglement-assisted capacity of a quantum channel and the reverse shannon theorem.
\newblock {\em IEEE Transactions on Information Theory}, 48(10):2637--2655, 2002.

\bibitem[BW92]{Bennett1992}
Charles~H Bennett and Stephen~J Wiesner.
\newblock {Communication via one-and two-particle operators on Einstein-Podolsky-Rosen states}.
\newblock {\em Physical Review Letters}, 69(20):2881--2884, nov 1992.

\bibitem[CDP08a]{Chiribella2008}
G.~Chiribella, G.~M. D'Ariano, and P.~Perinotti.
\newblock {Transforming quantum operations: Quantum supermaps}.
\newblock {\em Europhysics Letters}, 83(3), 2008.

\bibitem[CDP08b]{Chiribella2008a}
G.~Chiribella, G.~M. D’Ariano, and P.~Perinotti.
\newblock Quantum circuit architecture.
\newblock {\em Physical Review Letters}, 101(6), Aug 2008.

\bibitem[CG23]{Cheng2023a}
Hao-Chung Cheng and Li~Gao.
\newblock {Error exponent and strong converse for quantum soft covering}.
\newblock {\em IEEE Transactions on Information Theory}, 2023.

\bibitem[CGB23]{Cheng2017b}
Hao-Chung Cheng, Li~Gao, and Mario Berta.
\newblock {Quantum Broadcast Channel Simulation via Multipartite Convex Splitting}.
\newblock {\em IEEE Transactions on Information Theory}, 64(2):1385--1403, apr 2023.

\bibitem[CHSH69]{Clauser1969}
John~F. Clauser, Michael~A. Horne, Abner Shimony, and Richard~A. Holt.
\newblock Proposed experiment to test local hidden-variable theories.
\newblock {\em Physical Review Letters}, 23:880--884, Oct 1969.

\bibitem[CLL06]{CHILDS_2006}
Andrew~M. Childs, Debbie~W. Leung, and Hoi-Kwong Lo.
\newblock Two-way quantum communication channels.
\newblock {\em International Journal of Quantum Information}, 04(01):63–83, Feb 2006.

\bibitem[CLMW11]{Cubitt_2011}
Toby~S. Cubitt, Debbie Leung, William Matthews, and Andreas Winter.
\newblock Zero-error channel capacity and simulation assisted by non-local correlations.
\newblock {\em IEEE Transactions on Information Theory}, 57(8):5509–5523, Aug 2011.

\bibitem[CMW16]{Cooney_2016}
Tom Cooney, Milán Mosonyi, and Mark~M. Wilde.
\newblock Strong converse exponents for a quantum channel discrimination problem and quantum-feedback-assisted communication.
\newblock {\em Communications in Mathematical Physics}, 344(3):797–829, May 2016.

\bibitem[Dat09]{datta2009min}
Nilanjana Datta.
\newblock Min-and max-relative entropies and a new entanglement monotone.
\newblock {\em IEEE Transactions on Information Theory}, 55(6):2816--2826, 2009.

\bibitem[Dev05]{Devetak2005}
I.~Devetak.
\newblock The private classical capacity and quantum capacity of a quantum channel.
\newblock {\em IEEE Transactions on Information Theory}, 51(1):44--55, 2005.

\bibitem[DFW{\etalchar{+}}18]{D_az_2018}
María~García Díaz, Kun Fang, Xin Wang, Matteo Rosati, Michalis Skotiniotis, John Calsamiglia, and Andreas Winter.
\newblock Using and reusing coherence to realize quantum processes.
\newblock {\em Quantum}, 2:100, Oct 2018.

\bibitem[DKQ{\etalchar{+}}23]{Ding2023}
Dawei Ding, Sumeet Khatri, Yihui Quek, Peter~W. Shor, Xin Wang, and Mark~M. Wilde.
\newblock Bounding the forward classical capacity of bipartite quantum channels.
\newblock {\em IEEE Transactions on Information Theory}, 69(5):3034--3061, 2023.

\bibitem[DTW16]{Datta2016}
Nilanjana Datta, Marco Tomamichel, and Mark~M Wilde.
\newblock {On the second-order asymptotics for entanglement-assisted communication}.
\newblock {\em Quantum Information Processing}, 15(6):2569--2591, jun 2016.

\bibitem[DW16]{Duan2016}
Runyao Duan and Andreas Winter.
\newblock {No-signalling-assisted zero-error capacity of quantum channels and an information theoretic interpretation of the lov{\'{a}}sz number}.
\newblock {\em IEEE Transactions on Information Theory}, 62(2):891--914, 2016.

\bibitem[EGK11]{el2011network}
Abbas El~Gamal and Young-Han Kim.
\newblock {\em Network information theory}.
\newblock Cambridge university press, 2011.

\bibitem[Eke91]{Ekert1991}
Artur~K. Ekert.
\newblock {Quantum cryptography based on Bell's theorem}.
\newblock {\em Physical Review Letters}, 67(6):661--663, aug 1991.

\bibitem[ESW02]{Eggeling_2002}
T~Eggeling, D~Schlingemann, and R.~F Werner.
\newblock Semicausal operations are semilocalizable.
\newblock {\em Europhysics Letters}, 57(6):782–788, Mar 2002.

\bibitem[FBB19]{Faist_2019}
Philippe Faist, Mario Berta, and Fernando Brandão.
\newblock Thermodynamic capacity of quantum processes.
\newblock {\em Physical Review Letters}, 122(20), May 2019.

\bibitem[FF24a]{Omar2024}
Omar Fawzi and Paul Ferm\'e.
\newblock Broadcast channel coding: Algorithmic aspects and non-signaling assistance.
\newblock {\em IEEE Transactions on Information Theory}, pages 1--1, 2024.

\bibitem[FF24b]{Omar2024a}
Omar Fawzi and Paul Ferm\'e.
\newblock Multiple-access channel coding with non-signaling correlations.
\newblock {\em IEEE Transactions on Information Theory}, 70(3):1693–1719, Mar 2024.

\bibitem[FHS{\etalchar{+}}12]{Fawzi2012a}
Omar Fawzi, Patrick Hayden, Ivan Savov, Pranab Sen, and Mark~M. Wilde.
\newblock {Classical Communication Over a Quantum Interference Channel}.
\newblock {\em IEEE Transactions on Information Theory}, 58(6):3670--3691, jun 2012.

\bibitem[FR18]{Faist2018}
Philippe Faist and Renato Renner.
\newblock Fundamental work cost of quantum processes.
\newblock {\em Physical Review X}, 8:021011, Apr 2018.

\bibitem[FRS05]{Fan_2005}
Heng Fan, Vwani Roychowdhury, and Thomas Szkopek.
\newblock Optimal two-qubit quantum circuits using exchange interactions.
\newblock {\em Physical Review A}, 72(5), Nov 2005.

\bibitem[FWTB20]{Fang_2020}
Kun Fang, Xin Wang, Marco Tomamichel, and Mario Berta.
\newblock Quantum channel simulation and the channel’s smooth max-information.
\newblock {\em IEEE Transactions on Information Theory}, 66(4):2129–2140, Apr 2020.

\bibitem[GB14]{cvx}
Michael Grant and Stephen Boyd.
\newblock {CVX}: Matlab software for disciplined convex programming, version 2.1.
\newblock \url{http://cvxr.com/cvx}, March 2014.

\bibitem[GHR{\etalchar{+}}16]{Goold_2016}
John Goold, Marcus Huber, Arnau Riera, Lídia~del Rio, and Paul Skrzypczyk.
\newblock The role of quantum information in thermodynamics—a topical review.
\newblock {\em Journal of Physics A: Mathematical and Theoretical}, 49(14):143001, Feb 2016.

\bibitem[Gou19]{Gour_2019}
Gilad Gour.
\newblock Comparison of quantum channels by superchannels.
\newblock {\em IEEE Transactions on Information Theory}, 65(9):5880–5904, Sep 2019.

\bibitem[GS20a]{Gilad_2020a}
Gilad Gour and Carlo~Maria Scandolo.
\newblock Dynamical entanglement.
\newblock {\em Physical Review Letters}, 125:180505, Oct 2020.

\bibitem[GS20b]{Gilad_2020}
Gilad Gour and Carlo~Maria Scandolo.
\newblock Dynamical resources, 2020.

\bibitem[GS21]{Gour2021}
Gilad Gour and Carlo~Maria Scandolo.
\newblock Entanglement of a bipartite channel.
\newblock {\em Phys. Rev. A}, 103:062422, Jun 2021.

\bibitem[GW19]{Gour_2019a}
Gilad Gour and Andreas Winter.
\newblock How to quantify a dynamical quantum resource.
\newblock {\em Physical Review Letters}, 123(15), Oct 2019.

\bibitem[Hir23]{Hirche2023}
Christoph Hirche.
\newblock {Quantum network discrimination}.
\newblock {\em Quantum}, 7:1064, 2023.

\bibitem[HK81]{Han1981}
Te~Han and K.~Kobayashi.
\newblock A new achievable rate region for the interference channel.
\newblock {\em IEEE Transactions on Information Theory}, 27(1):49--60, 1981.

\bibitem[HL11]{Harrow_2011}
Aram~W. Harrow and Debbie~W. Leung.
\newblock A communication-efficient nonlocal measurement with application to communication complexity and bipartite gate capacities.
\newblock {\em IEEE Transactions on Information Theory}, 57(8):5504--5508, 2011.

\bibitem[Inc22]{MATLAB}
The~MathWorks Inc.
\newblock Matlab version: 9.13.0 (r2022b), 2022.

\bibitem[Joh16]{qetlab}
Nathaniel Johnston.
\newblock {QETLAB}: A {MATLAB} toolbox for quantum entanglement, version 0.9.
\newblock \url{https://qetlab.com}, Jan 2016.

\bibitem[Kon76]{Konno1976ACP}
Hiroshi Konno.
\newblock A cutting plane algorithm for solving bilinear programs.
\newblock {\em Mathematical Programming}, 11:14--27, 1976.

\bibitem[KS13]{K_nig_2013}
Robert König and Graeme Smith.
\newblock Limits on classical communication from quantum entropy power inequalities.
\newblock {\em Nature Photonics}, 7(2):142–146, Jan 2013.

\bibitem[KW24]{khatri2024}
Sumeet Khatri and Mark~M. Wilde.
\newblock Principles of quantum communication theory: A modern approach, 2024.

\bibitem[LALS20]{Leditzky_2020}
Felix Leditzky, Mohammad~A. Alhejji, Joshua Levin, and Graeme Smith.
\newblock Playing games with multiple access channels.
\newblock {\em Nature Communications}, 11(1), Mar 2020.

\bibitem[Li14]{Li2014a}
Ke~Li.
\newblock {Second-order asymptotics for quantum hypothesis testing}.
\newblock {\em The Annals of Statistics}, 42(1):171--189, feb 2014.

\bibitem[LKDW18]{Leditzky_2018}
Felix Leditzky, Eneet Kaur, Nilanjana Datta, and Mark~M. Wilde.
\newblock Approaches for approximate additivity of the holevo information of quantum channels.
\newblock {\em Physical Review A}, 97(1), Jan 2018.

\bibitem[LM15]{Leung2015}
Debbie Leung and William Matthews.
\newblock {On the power of PPT-preserving and non-signalling codes}.
\newblock {\em IEEE Transactions on Information Theory}, 61(8):4486--4499, 2015.

\bibitem[LW19]{liu2019resource}
Zi-Wen Liu and Andreas Winter.
\newblock Resource theories of quantum channels and the universal role of resource erasure, 2019.

\bibitem[Mat12]{Matthews_2012}
W.~Matthews.
\newblock A linear program for the finite block length converse of polyanskiy–poor–verdú via nonsignaling codes.
\newblock {\em IEEE Transactions on Information Theory}, 58(12):7036–7044, Dec 2012.

\bibitem[MBC22]{Milz_2022}
Simon Milz, Jessica Bavaresco, and Giulio Chiribella.
\newblock Resource theory of causal connection.
\newblock {\em Quantum}, 6:788, Aug 2022.

\bibitem[MK09]{Motahari2009}
Abolfazl~Seyed Motahari and Amir~Keyvan Khandani.
\newblock Capacity bounds for the gaussian interference channel.
\newblock {\em IEEE Transactions on Information Theory}, 55(2):620--643, 2009.

\bibitem[MLDS{\etalchar{+}}13]{M_ller_Lennert_2013}
Martin M\"uller-Lennert, Fr\'ed\'eric Dupuis, Oleg Szehr, Serge Fehr, and Marco Tomamichel.
\newblock On quantum r\'enyi entropies: A new generalization and some properties.
\newblock {\em Journal of Mathematical Physics}, 54(12), Dec 2013.

\bibitem[OCB12]{Oreshkov_2012}
Ognyan Oreshkov, Fabio Costa, and {\v{C}}aslav Brukner.
\newblock Quantum correlations with no causal order.
\newblock {\em Nature Communications}, 3(1), Oct 2012.

\bibitem[Pet86]{Petz1986}
D\'enes Petz.
\newblock Quasi-entropies for finite quantum systems.
\newblock {\em Reports on Mathematical Physics}, 23(1):57--65, 1986.

\bibitem[PHHH06]{Piani_2006}
M.~Piani, M.~Horodecki, P.~Horodecki, and R.~Horodecki.
\newblock Properties of quantum nonsignaling boxes.
\newblock {\em Physical Review A}, 74(1), Jul 2006.

\bibitem[PR97]{Popescu1997}
Sandu Popescu and Daniel Rohrlich.
\newblock Causality and nonlocality as axioms for quantum mechanics, 1997.

\bibitem[PV10]{Polyanskiy2010}
Yury Polyanskiy and Sergio Verd{\'u}.
\newblock Arimoto channel coding converse and r{\'e}nyi divergence.
\newblock {\em 2010 48th Annual Allerton Conference on Communication, Control, and Computing (Allerton)}, pages 1327--1333, 2010.

\bibitem[Ren05]{renner2005security}
Renato Renner.
\newblock {\em Security of Quantum Key Distribution}.
\newblock PhD thesis, ETH Zurich, 2005.

\bibitem[SCG20]{Saxena_2020}
Gaurav Saxena, Eric Chitambar, and Gilad Gour.
\newblock Dynamical resource theory of quantum coherence.
\newblock {\em Physical Review Research}, 2:023298, Jun 2020.

\bibitem[Sha48]{Shannon1948}
C.~E. Shannon.
\newblock A mathematical theory of communication.
\newblock {\em The Bell System Technical Journal}, 27(3):379--423, 1948.

\bibitem[Sha61]{shannon1961two}
Claude~E Shannon.
\newblock Two-way communication channels.
\newblock In {\em Proceedings of the Fourth Berkeley Symposium on Mathematical Statistics and Probability, Volume 1: Contributions to the Theory of Statistics}, volume~4, pages 611--645. University of California Press, 1961.

\bibitem[SW13]{Sharma_2013}
Naresh Sharma and Naqueeb~Ahmad Warsi.
\newblock Fundamental bound on the reliability of quantum information transmission.
\newblock {\em Physical Review Letters}, 110(8), Feb 2013.

\bibitem[TT15]{Tomamichel2015}
Marco Tomamichel and Vincent Y~F Tan.
\newblock {Second-Order Asymptotics for the Classical Capacity of Image-Additive Quantum Channels}.
\newblock {\em Communications in Mathematical Physics}, 338(1):103--137, aug 2015.

\bibitem[TV05]{Tse_Viswanath_2005}
David Tse and Pramod Viswanath.
\newblock {\em MIMO IV: multiuser communication}, page 425–495.
\newblock Cambridge University Press, 2005.

\bibitem[Wat09]{watrous2009semidefinite}
John Watrous.
\newblock Semidefinite programs for completely bounded norms.
\newblock {\em Theory of Computing}, 5(11):217--238, 2009.

\bibitem[WBHK20]{Wilde_2020}
Mark~M. Wilde, Mario Berta, Christoph Hirche, and Eneet Kaur.
\newblock Amortized channel divergence for asymptotic quantum channel discrimination.
\newblock {\em Letters in Mathematical Physics}, 110(8):2277–2336, June 2020.

\bibitem[WFT19]{Wang2017a}
Xin Wang, Kun Fang, and Marco Tomamichel.
\newblock {On Converse Bounds for Classical Communication Over Quantum Channels}.
\newblock {\em IEEE Transactions on Information Theory}, 65(7):4609--4619, jul 2019.

\bibitem[WW01]{Werner2001}
Reinhard~F. Werner and Michael~M. Wolf.
\newblock Bell inequalities and entanglement.
\newblock {\em Quantum Info. Comput.}, 1(3):1–25, October 2001.

\bibitem[WWY14]{Wilde_2014}
Mark~M. Wilde, Andreas Winter, and Dong Yang.
\newblock Strong converse for the classical capacity of entanglement-breaking and hadamard channels via a sandwiched r\'enyi relative entropy.
\newblock {\em Communications in Mathematical Physics}, 331(2):593–622, Jul 2014.

\bibitem[WXD18]{Wang_2018}
Xin Wang, Wei Xie, and Runyao Duan.
\newblock Semidefinite programming strong converse bounds for classical capacity.
\newblock {\em IEEE Transactions on Information Theory}, 64(1):640–653, Jan 2018.

\end{thebibliography}

\appendix
\setcounter{subsection}{0}
\setcounter{table}{0}
\setcounter{figure}{0}

\section{Dual SDP of the max-relative entropy of bidirectional classical communication}
Consider the primal SDP:
\be
\begin{aligned}
    \mathfrak{D}^{\leftrightarrow}_{\max}(\cN_{A_0B_0\rightarrow A_1B_1}) = \log\, \min\;\;& \lambda\\
     \textrm{s.t.} 
     & J_{A_0A_1B_0B_1}^{\cN} \leq Y_{A_0A_1B_0B_1}, \\
     & Y_{A_0A_1B_0B_1}\geq 0, Y_{A_0B_0} = \lambda \idop_{A_0B_0},\\
     & Y_{A_0B_0B_1} = \pi_{A_0}\ox Y_{B_0B_1}, Y_{A_0A_1B_0} = \pi_{B_0}\ox Y_{A_0A_1}.
\end{aligned}
\ee
The Lagrange function is
\be
\begin{aligned}
    &L\big(\lambda, Y_{A_0A_1B_0B_1}, M_{A_0A_1B_0B_1}, N_{A_0B_0}, P_{A_0B_0B_1}, Q_{A_0A_1B_0}\big)\\
    =\; & \lambda + \big\langle M, J_{A_0A_1B_0B_1}^{\cN} - Y_{A_0A_1B_0B_1} \big\rangle + \big\langle N, Y_{A_0B_0}-\lambda \idop_{A_0B_0}\big\rangle \\
    &+ \big\langle P, Y_{A_0B_0B_1}-\pi_{A_0} \ox Y_{B_0B_1}\big\rangle + \big\langle Q, Y_{A_0A_1B_0}-\pi_{B_0}\ox Y_{A_0A_1}\big\rangle\\
    =\; & \lambda + \tr\big[M J_{A_0A_1B_0B_1}^{\cN}\big] - \tr\big[M Y_{A_0A_1B_0B_1}\big] + \tr\big[(N_{A_0B_0}\ox \idop_{A_1B_1})Y_{A_0A_1B_0B_1}\big]\\
    & - \lambda \tr N_{A_0B_0} +\tr\big[(P_{A_0B_0B_1}\ox \idop_{A_1})Y_{A_0A_1B_0B_1}\big] + \tr\big[(Q_{A_0A_1B_0}\ox \idop_{B_1})Y_{A_0A_1B_0B_1}\big]\\
    & - \frac{1}{d_{A_0}}\tr[(P_{B_0B_1}\ox \idop_{A_0A_1}) Y_{A_0A_1B_0B_1}] - \frac{1}{d_{B_0}}\tr\big[(Q_{A_0A_1}\ox \idop_{B_0B_1}) Y_{A_0A_1B_0B_1}\big]\\
    =\; & \lambda(1-\tr N_{A_0B_0}) + \tr[M J_{A_0A_1B_0B_1}^{\cN}] + \Big\langle Y_{A_0A_1B_0B_1}, N_{A_0B_0}\ox \idop_{A_1B_1} - M_{A_0A_1B_0B_1}\\
    & + P_{A_0B_0B_1}\ox \idop_{A_1} + Q_{A_0A_1B_0}\ox \idop_{B_1} - \frac{1}{d_{A_0}} P_{B_0B_1}\ox \idop_{A_0A_1} - \frac{1}{d_{B_0}} Q_{A_0A_1}\ox \idop_{B_0B_1}\Big\rangle,
\end{aligned}
\ee
where we have introduced $M_{A_0A_1B_0B_1}, N_{A_0B_0}, P_{A_0B_0B_1}, Q_{A_0A_1B_0}$ as dual variables. The corresponding Lagrange dual function is 
\be
    g(M,N,P,Q) = \inf_{Y\geq 0}L(Y, \lambda, M, N, P, Q).
\ee
For $Y_{A_0A_1B_0B_1}\geq 0$, it must hold that $\tr N_{A_0B_0}\leq 1, M_{A_0A_1B_0B_1}\geq 0$, and
\begin{equation*}
    N_{A_0B_0}\ox \idop_{A_1B_1}+ P_{A_0B_0B_1}\ox \idop_{A_1} + Q_{A_0A_1B_0}\ox \idop_{B_1} \geq M_{A_0A_1B_0B_1} + \frac{1}{d_{A_0}}P_{B_0B_1}\ox \idop_{A_0A_1} + \frac{1}{d_{B_0}} Q_{A_0A_1}\ox \idop_{B_0B_1}.
\end{equation*}
Thus, we arrive at the dual SDP as
\begin{subequations}
\begin{align}
    \log\, \max &\; \tr\big[M_{A_0A_1B_0B_1} J_{A_0A_1B_0B_1}^{\cN}\big]\\
     \textrm{s.t.} &\;\; M_{A_0A_1B_0B_1}\geq 0, \, \tr N_{A_0B_0}\leq 1,\\
     &\;\; M_{A_0A_1B_0B_1} \leq N_{A_0B_0}\ox \idop_{A_1B_1}+ (P_{A_0B_0B_1}-\pi_{A_0}\ox P_{B_0B_1})\ox \idop_{A_1}\nonumber\\ 
     &\quad\quad + (Q_{A_0A_1B_0} - \pi_{B_0}\ox Q_{A_0A_1})\ox \idop_{B_1}.\label{Eq:const}
\end{align}
\end{subequations}
It is easy to see that the strong duality holds by Slater's condition. We further note that the inequality $\tr N_{AB}\leq 1$ can be restricted to equality. This is because for any feasible solution $\{M,N,P,Q\}$, we can replace $N$ with $\widehat{N}\coloneqq N_{A_0B_0}+(1-\tr N_{A_0B_0})\idop_{A_0B_0}/d_{A_0B_0}$ to get $\{M,\widehat{N},P,Q\}$. This is also a feasible solution as such shift also satisfies Eq.~\eqref{Eq:const}. Meanwhile, the objective function $\tr\big[MJ_{A_0A_1B_0B_1}^{\cN}\big]$ remains unchanged. Rewrite $P_{A_0B_0B_1} := P_{A_0B_0B_1}-\pi_{A_0}\ox P_{B_0B_1}$ and $Q_{A_0A_1B_0} := Q_{A_0A_1B_0} - \pi_{B_0}\ox Q_{A_0A_1}$. Thus, the final dual SDP is
\begin{equation*}
\begin{aligned}
   \mathfrak{D}^{\leftrightarrow}_{\max}(\cN_{A_0B_0\rightarrow A_1B_1}) = \log\, \max &\; \tr\big[M_{A_0A_1B_0B_1} J_{A_0A_1B_0B_1}^{\cN}\big]\\
    \textrm{s.t.} &\; \tr N_{A_0B_0} = 1,~P_{B_0B_1} = 0,~Q_{A_0A_1} = 0,\\
    &\;\; 0\leq M_{A_0A_1B_0B_1} \leq N_{A_0B_0}\ox \idop_{A_1B_1}+ P_{A_0B_0B_1}\ox \idop_{A_1} + Q_{A_0A_1B_0}\ox \idop_{B_1}.
\end{aligned}
\end{equation*}

\end{document}